%% file: main.tex
\documentclass[11pt]{article}
\usepackage[utf8]{inputenc}

\input{packages.tex}

\input{macros.tex}

\title{Random Access in Grammar-Compressed Strings:\\
Optimal Trade-Offs in Almost All Parameter Regimes}

\author[1,2]{Anouk Duyster \orcidlink{0009-0009-6027-9377} \email{aduyster@mpi-inf.mpg.de} }
\author[1]{Tomasz Kociumaka \orcidlink{0000-0002-2477-1702} \email{tomasz.kociumaka@mpi-inf.mpg.de} }

\affil[1]{Max Planck Institute for Informatics, SIC, Saarbrücken, Germany}
\affil[2]{Saarbrücken Graduate School of Computer Sciences, SIC, Saarbrücken, Germany}

\begin{document}
\date{\vspace{-1cm}}
\maketitle

\input{abstract}
\thispagestyle{empty}
\setcounter{page}{0}
\newpage

\input{introduction}

\input{preliminaries}

\input{upper/overview}

\input{lowerbound}

\input{upper/prelim}

\input{upper/contracting-weighted-slgs}

\input{upper/contracting-rlslgs}

\input{upper/nice}

\input{upper/leafy}

\input{upper/bound}

\input{upper/traversal}

\input{upper/rankandselect}

\input{appendix}

\bibliographystyle{alphaurl}
\bibliography{paper}

\end{document}

%% file: packages.tex
\usepackage[T1]{fontenc}
\usepackage{csquotes}
\usepackage{footmisc}
\usepackage{authblk}

\usepackage{amsmath}
\usepackage{amsthm}
\usepackage[margin=1in]{geometry}
\usepackage{amssymb}
\usepackage{mathtools}
\usepackage[inline]{enumitem}
\usepackage{microtype}
\setlist[itemize]{noitemsep,topsep=0.25\baselineskip,partopsep=0pt}
\setlist[enumerate]{noitemsep,topsep=0.25\baselineskip,partopsep=0pt}
\setlist[description]{noitemsep,topsep=0.25\baselineskip,partopsep=0pt}

\usepackage{thmtools}

\usepackage{tikz}
\usetikzlibrary{arrows.meta}
\usepackage{xcolor}

\usepackage[noend,ruled,linesnumbered]{algorithm2e}

\usepackage[colorlinks]{hyperref}
\usepackage[capitalize]{cleveref} 

\usepackage{orcidlink} 
\usepackage{bbding} 

%% file: macros.tex
\newcommand{\Z}{\mathbb{Z}}
\newcommand{\Zz}{\Z_{\ge 0}}

\newcommand{\Zp}{\Z_{>0}}
\newcommand{\Oh}{\mathcal{O}}

\newcommand{\poly}{\operatorname{poly}}
\newcommand{\dd}{.\,.}
\newcommand{\sub}{\subseteq}

\newcommand{\G}{\mathcal{G}} 

\newcommand{\Symb}{\mathcal{A}}
\newcommand{\Var}{\mathcal{V}}

\newcommand{\rle}{\mathsf{rle}}

\newcommand{\rhs}{\mathsf{rhs}}
\newcommand{\emptystring}{\varepsilon}

\newcommand{\cH}{\mathcal{H}}
\newcommand{\cK}{\mathcal{G}''}
\newcommand{\height}{h}
\newcommand{\Parse}{P}
\newcommand{\Lang}{L}

\DeclarePairedDelimiter\floor{\lfloor}{\rfloor}
\DeclarePairedDelimiter\ceil{\lceil}{\rceil}

\NewDocumentCommand\expand{o m}{
\mathsf{expand}{\IfValueTF {#1} {_{#1}} {} }(#2)
}
\NewDocumentCommand\weight{o m}{
 \lVert #2 \rVert{\IfValueTF {#1} {_{#1}} {} }
}
\NewDocumentCommand\explen{o m}{
|\mathsf{expand}{\IfValueTF {#1} {_{#1}} {} }(#2)|
}

\newcommand{\RA}{\textsc{Random Access}\xspace}
\newcommand{\fragmentcc}[2]{[#1\dd #2]}
\newcommand{\fragmentco}[2]{[#1\dd #2)}
\newcommand{\iv}{\fragmentco}
\newcommand{\rev}[1]{\overline{#1}}

\newcommand{\HL}{\cH_{\mathsf{L}}}
\newcommand{\HR}{\cH_{\mathsf{R}}}

\newcommand{\Vartop}[1]{\Var_{#1}^{\textsf{top}}}
\newcommand{\Varleaf}[1]{{\Var_{#1}^{\textsf{leaf}}}}
\newcommand{\Top}[1]{#1^{\textsf{top}}}
\newcommand{\Leaf}[1]{#1^{\textsf{leaf}}}

\newcommand{\rank}{\textsf{rank}}
\newcommand{\select}{\textsf{select}}

\newcommand{\child}{\mathsf{child}}
\newcommand{\new}{\mathsf{new}}
\newcommand{\forward}{\mathsf{forward}}
\newcommand{\backward}{\mathsf{backward}}
\newcommand{\pos}{\mathsf{pos}}

\newcommand{\node}{\mathsf{node}}

\newcommand{\pop}[1]{{#1}.\mathsf{pop}()}
\newcommand{\pushLeft}[1]{{#1}.\mathsf{pushLeft}()}
\newcommand{\pushRight}[1]{{#1}.\mathsf{pushRight}()}
\newcommand{\popLeft}[1]{{#1}.\mathsf{popLeft}()}
\newcommand{\popRight}[1]{{#1}.\mathsf{popRight}()}
\newcommand{\pushChild}[2]{{#1}.\mathsf{pushChild}({#2})}
\newcommand{\rightSibling}[1]{{#1}.\mathsf{rightSibling}()}

\newcommand{\symb}{\mathsf{symb}}
\newcommand{\leafPos}{\mathsf{leafPos}}

\newcommand{\PrefixSum}{\textsc{PrefixSum}\xspace}

\newtheorem{theorem}{Theorem}[section]
\newtheorem{lemma}[theorem]{Lemma}
\newtheorem{proposition}[theorem]{Proposition}
\newtheorem{corollary}[theorem]{Corollary}
\newtheorem{observation}[theorem]{Observation}
\newtheorem{fact}[theorem]{Fact}
\theoremstyle{definition}
\newtheorem{definition}[theorem]{Definition}
\newtheorem{claim}[theorem]{Claim}
\theoremstyle{remark}
\newtheorem{remark}[theorem]{Remark}
\newtheorem{example}[theorem]{Example}

\crefname{lemma}{Lemma}{Lemmas}
\crefname{proposition}{Proposition}{Propositions}
\crefname{fact}{Fact}{Facts}
\crefname{corollary}{Corollary}{Corollaries}
\crefname{observation}{Observation}{Observations}
\crefname{definition}{Definition}{Definitions}
\crefname{claim}{Claim}{Claims}

\crefname{algocf}{alg.}{algs.}
\Crefname{algocf}{Algorithm}{Algorithms}

\newcommand{\defpardsproblem}[4]{
\vspace{0.5em}
\noindent\fbox{
  \begin{minipage}{\linewidth-14pt}
  \vspace{0.3em}
  \textsc{#1}\\
  {\bf{Parameters:}} #2\\
  {\bf{Input:}} #3  \\
  {\bf{Queries:}} #4
  \end{minipage}
  }\\
}

\newcommand{\defdsproblem}[3]{
\vspace{0.5em}
\noindent\fbox{
  \begin{minipage}{\linewidth-14pt}
  \vspace{0.3em}
  \textsc{#1}\\
  {\bf{Input:}} #2  \\
  {\bf{Queries:}} #3

  \end{minipage}
  }\\
}

\newcommand{\cX}{\mathcal{X}}
\newcommand{\cY}{\mathcal{Y}}
\newcommand{\cZ}{\mathcal{Z}}
\newcommand{\cM}{\mathbf{M}}
\newcommand{\rM}{\mathbf{r}_\mathbf{M}}
\newcommand{\cZero}{\mathbf{0}}
\newcommand{\oper}{\mathbin{\boldsymbol{\oplus}}}
\newcommand{\Oper}{\mathop{\boldsymbol{\bigoplus}}}

\definecolor{orcidlogocol}{HTML}{A6CE39}
\newcommand{\email}[1]{\href{mailto:#1}{\Envelope}}

%% file: abstract.tex
\begin{abstract}
  A \RA query to a string \(T\) asks for the character \(T[i]\) at a given position $i\in [0\dd |T|)$.
  This fundamental task admits a straightforward solution with constant-time queries and $\Oh(n \log \sigma)$ bits of space when $T\in [0\dd \sigma)^n$.
  While this is the best one can achieve in the worst case, much research has focused on the compressed setting: if $T$ is compressible, one can hope for a much smaller data structure that still answers \RA queries efficiently.

  In this work, we investigate the grammar-compressed setting, where $T$ is represented by a context-free grammar that produces only $T$.
  Our main result is a general trade-off that optimizes \RA time as a function of the string length $n$, the grammar size (the total length of productions) $g$, the alphabet size $\sigma$, the data structure size $M$, and the word size $w\ge\Omega(\log n)$ of the word RAM model.
  For any data structure size $M$ satisfying \(g\log n < Mw < n\log \sigma\), we show an $\Oh(M)$-size data structure that answers \RA queries in time
  \[\Oh\left(\frac{\log \frac{n \log \sigma}{Mw}}{\log\frac{Mw}{g\log n}}\right).\]
  We also prove a matching unconditional lower bound that holds for all parameter regimes except very small grammars ($g \le w^{1+o(1)}\log n$) and relatively small data structures ($Mw \le g \log n \cdot w^{o(1)}$).
  The lower bound applies to word-RAM query time and, more strongly, to the worst-case cell-probe complexity of nondeterministic or bounded-error randomized query~algorithms.

  Previous work focused on optimizing the query time as a function of $n$ only, achieving $\Oh(\log n)$ time using $\Oh(g)$ space [Bille, Landau, Raman, Sadakane, Satti, Weimann; SIAM J. Comput.~2015] and $\Oh(\frac{\log n}{\log \log n})$ time using $\Oh(g\log^{\epsilon} n)$ space for any constant $\epsilon > 0$ [Belazzougui, Cording, Puglisi, Tabei; ESA 2015], [Ganardi, Jeż, Lohrey; J. ACM 2021].
  Our result improves upon these bounds (strictly for $g=n^{1-o(1)}$) and generalizes them beyond $M\le\Oh(g \poly\!\log n)$, yielding a smooth interpolation with the uncompressed setting of $Mw = n\log\sigma$ bits.

  Thus far, the only tight lower bound [Verbin and Yu; CPM 2013] was $\Omega(\smash{\frac{\log n}{\log \log n}})$ for $w=\Theta(\log n)$, $\smash{n^{\Omega(1)}} \le g \le \smash{n^{1-\Omega(1)}}$, and $M=\smash{g\cdot \log^{\Theta(1)} n}$.
  In contrast, our result yields a tight bound that accounts for all relevant parameters and is valid for almost all parameter regimes.

  Our bounds remain valid for run-length grammars, where production sizes use run-length encoding.
  This lets us recover (and, for strings with small run-length grammars, improve) the trade-offs achieved by block trees, formulated in terms of the LZ77 size $z$ [Belazzougui, Cáceres, Gagie, Gawrychowski, Kärkkäinen, Navarro, Ordóñez, Puglisi, Tabei; J.\ Comput.\ Syst.\ Sci.\ 2021] and substring complexity~$\delta$ [Kociumaka, Navarro, Prezza; IEEE Trans.\ Inf.\ Theory 2023].

  Our data structure admits an efficient deterministic construction algorithm.
  Beyond \RA, its variants also support substring extraction (with optimal additive overhead $\Oh(\frac{m\log \sigma}{w})$ for a length-$m$ substring, provided that $M \ge g$), as well as rank and select queries.

  All our results rely on novel grammar transformations that generalize contracting grammars [Ganardi; ESA 2021] and achieve the optimal trade-off between grammar size and height while enforcing extra structure crucial for constant-time navigation in the parse tree.
\end{abstract}

%% file: introduction.tex
\section{Introduction}\label{sec:intro}

\RA is arguably the most basic query one can ask about a string and a fundamental building block of classic string-processing algorithms.
For a string $T\in \Sigma^*$ and a position $i\in [0\dd |T|)$, it asks for the $i$-th character $T[i]\in \Sigma$ of $T$.\footnote{We adhere to $0$-based indexing of strings, that is, $T=T[0]\cdots T[|T|-1]$.
For $a,b\in \mathbb{R}$, we denote $[a\dd b]=\{k\in \mathbb{Z} : \allowbreak a \le k \le b\}$ and $(a\dd b)=\{k\in \mathbb{Z}: a < k < b\}$.
Half-open integer intervals $[a\dd b)$ and $(a\dd b]$ are defined analogously.}
In the standard representation of strings as arrays, \RA trivially takes constant time.
For strings over integer alphabets $\Sigma=[0\dd \sigma)$, constant-time \RA can still be supported using a \emph{packed} representation in $\Theta(n\log\sigma)$ bits thanks to the bit-manipulation operations available in the standard word RAM model.
Up to a constant factor, this simple data structure matches the information-theoretic lower bound of $\ceil{n\log\sigma}$ bits.\footnote{All logarithms in this work are binary unless another base is provided.}
More advanced tools~\cite{DPT10} come even closer, achieving $n\log\sigma + \Oh(\log^2 n)$ bits.

The complexity landscape of \RA queries becomes much more interesting in the \emph{compressed setting}, where we aim to achieve much smaller data structures for some ``compressible'' strings at the price of marginally larger data structures for the remaining ``incompressible'' strings.
Every injective function $C: \Sigma^* \to \{0,1\}^*$ can be seen as a compression method, but few such functions admit an efficient query algorithm for \RA.
In particular, the most popular compressors such as arithmetic coding~\cite{HV94}, Lempel--Ziv parsing~\cite{LZ77}, or run-length encoded Burrows--Wheeler transform~\cite{BW94} do not natively support \RA in \(n^{o(1)}\) time.
Since \RA is needed for most string-processing algorithms, a lot of research has been devoted to understanding how much extra space is needed to enable such support.
This question formally asks what trade-offs are possible between the data structure size and its \RA query time, where both complexities are measured not only in terms of the length $n$ and alphabet size $\sigma$ of the input string $T$ but also a parameter $s(T)$ capturing the size of its compressed representation.

In this work, we focus on the so-called highly-compressible regime, where typical space complexities are of the form $\Oh(s(T))$ or $\Oh(s(T) \poly \log n)$ rather than $s(T) + o(n)$.
In this setting, a clean formalization that captures many compressors is grammar-compression~\cite{R76,KY00}: the input string $T$ is encoded using a context-free grammar $\G$ that produces $T$ and no other string.
Such a grammar is a \emph{straight-line grammar} (SLG) because it cannot use circular dependencies between variables or define multiple productions for the same variable.
The underlying compressibility measure, denoted $g(T)$, is the size (total production length) of the smallest grammar generating $T$.
A long line of works has aimed to understand the complexity of \RA in SLG-compressed~strings:

\defpardsproblem{SLG Random Access}{$n,g,\sigma\in \Zp$}{
  A string $T\in [0\dd \sigma)^n$ represented by an SLG of size at most $g$.
}{
  Given a position $i\in \iv{0}{n}$, return $T[i]$.
}

Time-space trade-offs for this problem are inherently related to the trade-offs between the grammar height (the height of the underlying parse tree, i.e., the worst-case number of productions needed to derive a given character of $T$) and its size.
In 2002, parallel works of Charikar et al.~\cite{CLL05} and Rytter~\cite{Ryt03} showed that if $T$ is produced by an arbitrary SLG of size $g$, it is also produced by a \emph{balanced} SLG of height $\Oh(\log n)$ and size $\Oh(g \log\frac{n}{g})$.
This immediately implies $\Oh(\log n)$-time \RA using a data structure of size $\Oh(g\log \frac{n}{g})$.\footnote{We state all results in the word RAM model. In the introduction, we typically use word size $w=\Theta(\log (n\sigma))$, and we measure space complexities in machine words. Later on, we switch to bits and treat $w$ as a separate parameter.}
A decade later, a breakthrough of Bille et al.~\cite{BLR15} achieved $\Oh(\log n)$-time \RA using $\Oh(g)$ space.
On the lower-bound side, an influential work of Verbin and Yu~\cite{VY13} showed that an $\Oh(g \poly\!\log n)$-space data structure cannot support $o(\smash{\frac{\log n}{\log \log n}})$-time \RA in their hard regime.
Shortly afterward, Belazzougui et al.~\cite{BCJT15} proved that $\Oh(\smash{\frac{\log n}{\log \log n}})$ query time is achievable in $\Oh(g \log^{\epsilon} n)$ space for balanced SLGs and any constant $\epsilon > 0$.
In a recent milestone, Ganardi, Jeż, and Lohrey~\cite{GJL21} showed every SLG can be balanced while keeping its size at $\Oh(g)$; this immediately recovers the trade-off of \cite{BLR15} and proves that the trade-off of \cite{BCJT15} is valid for every SLG.
More generally, one can answer \RA queries in $\Oh(\log_\tau n)$ time using $\Oh(g \tau)$ space for $2\le \tau \le \poly\!\log n$.

A common limitation of all these previous works is that they measure the \RA time as a function of $n$ only.
The lower bound of \cite{VY13} is only valid when $n^{\Omega(1)} \le g \le n^{1-\Omega(1)}$,\footnote{\label{foot:VY13}%
Verbin and Yu~\cite{VY13} also prove a separate lower bound for a particular value $g \le \log^{2+o(1)} n$, which depends on $n$ and the data structure size $M \le n^{o(1)}$. That result only shows that the query time is at least $(\log n)^{1-o(1)}/\log M$.}
so we can hope for faster queries for the cases of $g\le n^{o(1)}$ and, much more importantly for applications, $g\ge n^{1-o(1)}$.
Moreover, neither the lower bound nor the upper bound applies to data structures of size exceeding $\Oh(g \poly\!\log n)$.
In this work, we aim to characterize what trade-offs are possible.

\begin{quote}%
\emph{%
What is the optimal \RA time in grammar-compressed strings, as a function of string length $n$, grammar size $g$, alphabet size $\sigma$, and data structure size $M$?
}
\end{quote}

On the positive side, we show that, for any $\tau \ge 2$, \RA queries can be answered in $\Oh(\log_\tau\frac{n \log \sigma}{g \tau \log n})$ time\footnote{In the introduction, we follow the convention that time complexity $\Oh(t)$ should be read as $\Oh(\max(1,t))$.} using $\Oh(g\tau)$ space.
Phrased in terms of $M$ rather than $\tau$, this reads as follows:

\begin{theorem}\label{thm:upper}
  Let $\G$ be an SLG of size $g$ generating a string $T\in [0\dd \sigma)^n$.
  In the word RAM model with word size $w\ge\Omega(\log (n\sigma))$, given $\G$ and any value $M$ with $g\log n < Mw < n\log \sigma$, one can in $\Oh(\frac{Mw}{\log n})$ time construct an $\Oh(M)$-size data structure supporting \RA queries in time\vspace{-.25cm}
    \[\Oh\left(\frac{\log \frac{n \log \sigma}{Mw}}{\log\frac{Mw}{g\log n}}\right).\]
\end{theorem}

Already in the regime of $\Oh(g)$ space, we improve the query time from $\Oh(\log n)$ to $\Oh(\log \smash{\frac{n\log \sigma}{g \log n}})$.
Our time bound smoothly interpolates between logarithmic for $g \le n^{1-\Omega(1)}$ and constant for $g \ge \Omega(\smash{\frac{n\log \sigma}{\log n}})$.
Since every string can be produced by an SLG of size $g\le \Oh(\smash{\frac{n\log \sigma}{\log n}})$, we recover the standard uncompressed bound: $\Oh(1)$ time using $\Oh(n\log \sigma)$ bits.
Allowing more space, we generalize the trade-off of \cite{BCJT15} beyond $\tau\le \poly\!\log n$ and strictly improve upon it when $g \ge n^{1-o(1)}$.
In particular, we get $\Oh(1)$ query time already using $\Oh((g \log n)^{1-\epsilon}(n \log \sigma)^\epsilon)$ bits for a constant $\epsilon > 0$.

On the lower-bound side, our main result is that the query time of \cref{thm:upper} is optimal for all parameter regimes except for very small grammars and relatively small data structures.
\begin{theorem}\label{thm:lower}
Consider integers $n,g,\sigma,M,w,t\in \Zp$ and a real constant $\epsilon > 0$ such that $n\ge g$, $n\ge\sigma \ge 2$, $Mw \ge g\cdot \log n \cdot (w \log n)^\epsilon$, and $g \ge 25\cdot w^{1+\epsilon}\cdot\log n$.
Suppose that, for every instance of SLG \RA with parameters $n,g,\sigma$, there is a data structure of $M$ $w$-bit machine words that answers each query by accessing $t$ of these machine words.
Then, \vspace{-.2cm} \[t\ge \Omega\left(\frac{\log \frac{n\log\sigma}{Mw}}{\log \frac{Mw}{g\log n}}\right).\]
This bound remains valid for nondeterministic and randomized data structures with two-sided error.
\end{theorem}
In the most studied setting of $w=\Theta(\log n)$, our tight lower bound applies to all grammar sizes except $g\le \log^{2+o(1)} n$ and all data structure sizes except $M\le g \log^{o(1)} n$.
As noted already in \cite{VY13}, $\log^{2} n$ is a natural barrier: if $g\le \log^{2-\delta} n$ for some constant $\delta>0$, then the query algorithm could read the entire grammar using $\log^{1-\delta+o(1)} n$ memory accesses.
The limitation $M \le g \log^{o(1)} n$ is shared with most static data-structure lower bounds, which typically capture space complexity only up to sub-logarithmic factors, except for a handful of isolated problems; see~\cite{Lar12}.
Whether $o(\log n)$ query time can be achieved in $\Oh(g)$ space remains a long-standing open problem.
\paragraph{Beyond SLGs.}
In \cite{CEK21,NOU22}, the $\Oh(\log n)$-time $\Oh(g)$-space trade-off for SLG \RA~\cite{BLR15,GJL21} has been generalized to run-length straight-line grammars (RLSLGs)~\cite{NII16}, which differ from SLGs in that production sizes are measured in terms of run-length encoding size rather than length.
Equivalently, in RLSLGs, productions of the form $A\to B^k$ for $k\in \Zp$ are assumed to be of size one.
The size of the smallest RLSLG generating $T$ satisfies $g_{\rle}(T)\le g(T)$, so RLSLG \RA is at least as hard as SLG \RA, and the lower bound of \cref{thm:lower} generalizes automatically.
Even though we are not aware of any previous work generalizing the trade-off of \cite{BCJT15} to RLSLGs, in \cref{sec:ub} we state and prove our upper bound of \cref{thm:upper} already in the stronger RLSLG variant (without any penalty).

A notable advantage of $g_\rle(T)$ is that it is a tighter upper bound of further measures, including the LZ77 size $z(T)$~\cite{LZ77} and the substring complexity $\delta(T)$~\cite{RRRS13,KNP23}, which satisfy $\delta(T)\le z(T) \le g_\rle(T) \le g(T)$.
In particular, $g_{\rle}(T) \le \Oh(\delta(T) \log\frac{n \log \sigma}{\delta(T) \log n})\le \Oh(z(T) \log\frac{n \log \sigma}{z(T) \log n})$, and the first of these bounds fails for $g(T)$~\cite{KNP23}.
Thanks to this characterization of $g_{\rle}(T)$, \cref{thm:upper} recovers the \RA trade-offs of \emph{block trees} formulated in~\cite{BCG21,KNP23} using $z(T)$ and $\delta(T)$, respectively.
Moreover, we derive from \cref{thm:lower} that these trade-offs are optimal, as functions of $z(T)$ and $\delta(T)$, respectively, in a wide range of parameter regimes.

\begin{restatable}[see \cref{section:appendix}]{corollary}{thmBlock}\label{thm:block}
  For every string $T\in [0\dd\sigma)^n$ and all parameters $d\ge \delta(T)$, $\tau\ge 2$, and $d\tau \le s \le \smash{d\tau \log_\tau \frac{n\log \sigma}{d\log n}}$, there is a data structure of size $\smash{\Oh(s+d \tau \log_\tau \frac{n\log \sigma}{s \log n})}$ supporting \RA queries in time $\smash{\Oh(\log_\tau \frac{n\log \sigma}{s \log n})}$.
  This query time is optimal if $n\ge \sigma\ge 2$, $\tau \ge \log^{\epsilon} n$, and $d \ge \log^{2+\epsilon} n$ for a constant $\epsilon > 0$, even restricted to strings satisfying $d\ge g(T) \ge z(T) \ge \delta(T)$.
\end{restatable}

\paragraph{Beyond Random Access.}
Many previous data structures for (RL)SLG \RA support extracting not only individual characters but also substrings of $T$ of arbitrary length $m$,
with the additive time overhead of $\Oh(m)$~\cite{BLR15,CEK21,NOU22} or, if the substring is returned in the packed representation, $\Oh(m/\log_\sigma n)$~\cite{BCJT15} or $\Oh(m/\log_\sigma n\cdot \log_\tau \frac{n\log\sigma}{s\log n})$~\cite{BCG21,KNP23}.
We generalize \cref{thm:upper} to achieve the optimal overhead of $\Oh(m/\log_\sigma n)$.
Moreover, ours is the first solution that outputs the packed representation and comes with an efficient construction.

\begin{theorem}\label{thm:retrieve}
  Let $\G$ be an RLSLG of size $g$ generating a string $T\in [0\dd \sigma)^n$.
  In the word RAM model with word size $w\ge\Omega(\log (n\sigma))$, given $\G$ and any value $M$ with $g\log n < Mw < n\log \sigma$, one can in $\Oh(\frac{Mw}{\log n})$ time construct an $\Oh(M)$-size data structure that extracts any length-$m$ \mbox{substring in time}
    \[\Oh\left(\frac{\log \frac{n \log \sigma}{Mw}}{\log\frac{Mw}{g\log n}} + \tfrac{m\log \sigma}{w}\cdot \tfrac{M+g}{M}\right).\]
\end{theorem}
More generally, we support fast character iterators that, upon initialization at position $i$, can in constant time move in either direction by $\Oh(\log_\sigma n)$ positions and output the characters at the intermediate positions.
None of the previous works achieves this without amortization for $\sigma\le n^{o(1)}$.

Our approach readily generalizes to computing various aggregate information about the prefixes of $T$, modeled by the prefix sum problem in which every character $a\in \Sigma$ is mapped to an element $\Phi(a)$ of a monoid.
In particular, similarly to \cite{BLR15,BCJT15,GJL21,NOU22,CEK21,KNP23}, we support $\rank$ and $\select$ queries, asking for the number of occurrences of a given character $a\in \Sigma$ in the prefix $T[0\dd i)$ and the position of the $r$-th leftmost occurrence of $a$ in $T$, respectively.

\begin{theorem}\label{thm:ranksel}
  Let $\G$ be an RLSLG of size $g$ generating a string $T\in [0\dd \sigma)^n$.
  In the word RAM model with $w\ge\Omega(\log (n\sigma))$, given $\G$ and any value $M$ with $g\log n < Mw < n$, one can in $\Oh(\frac{Mw\sigma}{\log n}+g\sigma\log n)$ time construct an $\Oh(M\sigma)$-size data structure supporting $\rank_{T}$ and $\select_{T}$ queries in time
    \[\Oh\left(\frac{\log \frac{n}{Mw}}{\log\frac{Mw}{g\log n}}\right).\]
\end{theorem}
For constant alphabet size $\sigma$, this result matches our time-space trade-off for \RA.
It is also straightforward to see that, already for $\sigma=2$, $\rank$ and $\select$ queries are at least as hard as \RA,\footnote{For $\rank$, note that $T[i]=1$ if and only if $\rank_{T,1}(i)<\rank_{T,1}(i+1)$. For $\select$, we map $T$ through a morphism $0\mapsto 10$ and $1\mapsto 01$, preserving $n$ and $g$ up to a constant factor.
The resulting string $T'$ satisfies $\select_{T',1}(i)=2i+T[i]$.}
so our bounds remain tight for almost all parameter regimes.
The only penalty that we suffer already for $\sigma\le\Oh(1)$ is the extra $\Oh(g\log n)$ term in the construction algorithm.
As we discuss in \cref{sec:rank_select}, this term arises only for $\select$, and only when $g \ge n/\log^{1+o(1)} n$.
It is an artifact of our usage of Elias--Fano encoding~\cite{E74,F71}, which currently lacks an efficient construction from the packed input representation using variable-length gap encoding.

In general, similarly to all previous data structures in grammar-compressed space,
ours suffers from a $\sigma$-factor overhead in space complexity (and construction time).
Reducing this overhead, possibly at the cost of slight increase in the query time, remains a major open question.
As noted in \cite[Section 6]{BCJT15}, this would require breakthroughs for counting paths between DAG vertices.

\subsection{Our Techniques}\label{sec:techniques}
We prove \cref{thm:upper,thm:retrieve,thm:ranksel} by converting the input (RL)SLG into an equivalent \emph{shallow} grammar that admits fast \emph{navigation} in its parse tree.
Let $\G$ be an (RL)SLG and let $A\in\Symb_\G$ be a symbol (variable or terminal).
We write $\rhs_\G(A)$ for the right-hand side of its production (a string over $\Symb_\G$) and $\expand[\G]{A}$ for its full expansion (the unique string over terminals derived from $A$).

\RA can be implemented recursively.
Suppose $A$ is a variable with production $A\to B_0\cdots B_{k-1}$, and we are asked to retrieve $\expand[\G]{A}[i]$.
Then there is a unique $j\in[0\dd k)$ such that
$\explen[\G]{B_0\cdots B_{j-1}} \le i < \explen[\G]{B_0\cdots B_{j}}$,
and the query recurses on the symbol $B_j$ and the shifted index $i-\explen[\G]{B_0\cdots B_{j-1}}$.
We denote the task of finding this $j$ by $\child_A$.
It is a $\rank$ query over the prefix sums $\explen[\G]{B_0\cdots B_{j}}$, asking for the number of such sums that are at most $i$.

Grammar balancing~\cite{GJL21} yields grammars of height $\Oh(\log n)$, size $\Oh(g)$, and production lengths $k\le 2$, leading to $\Oh(\log n)$-time \RA in $\Oh(g)$ space.
To reduce the query time to $\Oh(\log \frac{n}{g})$, we would like to decrease the height to $\Oh(\log \frac{n}{g})$ by creating a new length-$\Theta(g)$ right-hand side of the starting symbol $S$.
This is where our first obstacle appears: even if every variable $A$ has recursion depth $\Oh(\log \explen[\G]{A})$, a natural process starting with $S$ and exhaustively replacing all symbols longer than $\frac{n}{g}$ by their right-hand sides may create a sequence of $\Theta(g\log n)$ symbols, exceeding $\Oh(g)$ space.
We avoid this by working with \emph{contracting} grammars~\cite{G21}, where a variable $B$ can appear on the right-hand side of a variable $A$ only if $\explen[\G]{B}\le \frac12\explen[\G]{A}$.
We show in \cref{lemma:nice-var} how this implies that the same process produces a sequence of $\Oh(g)$ symbols of length at most $\frac{n}{g}$, thus reducing $\child_S$ to a $\rank$ query on an $\Oh(g)$-element set of prefix sums.
This query can then be implemented in $\Oh(\log \frac{n}{g})$ time~\cite{BLR15} or even in $\Oh(\log \log \frac{n}{g})$ time~\cite{BBV12}.

For a more general trade-off with $\Oh(\log_\tau \frac{n}{g})$ query time, we want to similarly increase the fan-out of every other variable $A$ to $\Theta(\tau)$ so that the height becomes $\Oh(\log_\tau \frac{n}{g})\le \Oh(1+\log_\tau \frac{n}{g\tau})$.
Here, an obstacle is that we have to implement $\child_A$ in constant time in order to avoid extra factors in \RA time.
This is possible for $\tau \le \poly\!\log n$ using fusion trees~\cite{Fredman1993}, which is why the previous trade-offs~\cite{BCJT15} were limited to small $\tau$.
Our key structural notion is $\tau$-\emph{niceness} (\cref{def:nice}): as for the start symbol, we replace $\rhs_\G(A)$ by a longer right-hand side $\rhs_{\tau}(A)$ of $\Oh(\tau)$ symbols of length at most $\frac{1}{\tau}\explen[\G]{A}$ each.
Crucially (see \cref{lemma:nice-var}), the contracting property of $\G$ implies a strong \emph{locality} condition: any substring of $\expand[\G]{A}$ of length at most $\frac1\tau\explen[\G]{A}$ intersects only $\Oh(\log\tau)$ consecutive symbols of $\rhs_{\tau}(A)$.
Thus, by storing the answers to $\Theta(\tau)$ evenly spaced $\rank$ queries, we reduce every $\child_A$ to a $\rank$ query in a set of size $\Oh(\log\tau)$, which we answer in constant time using fusion trees (see \cref{lemma:nicechild}).
Combining this with a standard shortcut (storing $\expand[\G]{A}$ explicitly whenever $\explen[\G]{A}\le \log_\sigma n$) yields the trade-off of \cref{thm:upper}.
The approach extends to RLSLGs once we generalize the transformation of \cite{G21} appropriately (\cref{cor:contracting_rlslg}; a parallel work~\cite{CLG26} achieves almost the same generalization).

\paragraph{Substring Extraction.}
To extract any substring of $T$ with constant delay per character using an $\Oh(g)$-size data structure, we could simply use the character iterators of \cite{GKPS05} or \cite{Lohrey2017}, with slight adaptations to support RLSLGs and have construction time matching our \RA.
Outputting blocks of $b=\Theta(\log_\sigma n)$ characters with constant delay is much more challenging; see \cite{BCJT15}.
The underlying issue is that a long variable $A$ may contain a very short child $B$ (even of constant expansion length), so a character-by-character traversal cannot be turned into a block-by-block traversal by local shortcuts alone.
We address this by converting $\G$ into a \emph{$b$-leafy} grammar: \emph{leaf variables} have expansion length $\Theta(b)$, while the remaining \emph{top variables} have constant-size right-hand sides over top and leaf variables only.
Treating leaf variables as terminals yields a grammar for a shorter string of length $\Theta(n/b)$, whose characters represent blocks of $\Theta(b)$ characters of $T$.
Hence, standard character iterators can output one leaf per step, i.e., $\Theta(b)$ characters at a time.
To preserve lengths under this reinterpretation, we lift most constructions to \emph{weighted strings}, assigning each leaf a weight equal to its original expansion length.
This is how we prove \cref{thm:retrieve}.

\paragraph{Prefix Aggregation, $\rank$, and $\select$.}
Prefix aggregation queries can be handled similarly to \RA: during the same root-to-leaf traversal, we maintain the monoid aggregate of the already traversed prefix, yielding $\Oh(\log_\tau \frac{n}{g\tau})$ time once $\child_A$ is supported in constant time.
The $\rank_{T,a}$ query is the special case where each character is mapped to $0$ or $1$ and the aggregate is the sum.
As in previous works~\cite{BCJT15,P19}, for $\select_{T,a}$, we first transform $T$ so that each character encodes a block of non-$a$ characters followed by an $a$, reducing $\select$ to a prefix aggregation query on the transformed string.
Obtaining the bound in \cref{thm:ranksel} (with $g\tau \log n$ in the denominator) requires a succinct structure for the leaf variables; we use Elias--Fano encoding~\cite{E74,F71}.

\paragraph{Lower Bound.}
We prove \cref{thm:lower} in the \emph{nondeterministic cell-probe model}, formulated in terms of \emph{certificates} by Wang and Yin~\cite{Wang2014},
as well as in the \emph{randomized cell-probe model with two-sided bounded error}.
We use a reduction from Blocked Lopsided Set Disjointness (BLSD)~\cite{P11}, extending the construction of Verbin and Yu~\cite{VY13} behind the lower bound discussed in Footnote~\ref{foot:VY13} by introducing a tunable splitting parameter~$Q$.
The key idea is that, instead of answering a single BLSD query using one \RA query in a highly compressible string ($g\approx \log^2 n$),
we answer it using $Q$ \RA queries in a moderately compressible string (roughly $g \approx Q \log^2 \frac{n}{Q}$).
This simple idea not only recovers the other lower bound of~\cite{VY13} (for $\smash{n^{\Omega(1)}\le g \le n^{1-\Omega(1)}}$) but also readily generalizes it to almost all parameter regimes.
The original proof of Verbin and Yu~\cite{VY13} also ultimately relies on the hardness of BLSD, but it passes through a several-step chain of reductions from~\cite{P11,VY13}.
In the randomized setting, errors could accumulate across all $Q$ \RA queries, so our reduction uses an extra subroutine that repeats some queries using multiple copies of the data structure in order to derive a bounded-error communication protocol for BLSD that violates the lower bound of~\cite{P11}.
No such subroutine appears in~\cite{P11}, which is why their randomized cell-probe lower bounds only work in the zero-error models.
In the nondeterministic setting, this issue disappears; the remaining annoyance, shared with the randomized setting, is the amount of calculations needed to map between all the parameters in the reduction.

\subsection{Related Work}
\RA has been studied from a multitude of perspectives for dozens of compressibility measures $s(T)$.
Early milestone results provide constant-time random access to entropy-compressed strings, achieving $nH_0(T)+o(n)$ bits for binary strings~\cite{Raman2007} and $nH_k(T)+o(n \log \sigma)$ bits for strings over larger alphabets~\cite{SG06} (for $k<o(\log_\sigma n)$); see also~\cite{Grossi2013} for a survey.
Unfortunately, low-order entropy does not capture large-scale repetitions, and the lower-order terms in these data structures are relatively large, which makes them unsuitable for highly compressible strings.
That regime brings a whole zoo of compressibility (repetitiveness) measures; see~\cite{Navarro2021} for a dedicated survey, which discusses which measures support efficient random access.
In space $\Oh(s(T)\poly\!\log n)$, the best one can typically hope for (in highly compressible strings) is $\Oh(\smash{\frac{\log n}{\log \log n}})$ \RA time: see~\cite{DeK24} for a work that generalizes the lower bound of~\cite{VY13} to many measures that are not bounded by $\Oh(g(T)\poly\!\log n)$.
An exception is LZ78~\cite{LZ78}, for which $\Oh(\log \log n)$-time access is possible~\cite{DLR13}.
Unfortunately, it is not known how to support \RA queries efficiently (e.g., in $n^{o(1)}$ time) in space proportional to the size $z(T)$ of the LZ77 parsing or $r(T)$ of the run-length encoded Burrows--Wheeler transform.
The latter is particularly surprising given $\Oh(r(T))$-space data structures supporting pattern-matching queries~\cite{GNP20}, which often makes \RA the sole bottleneck; see~\cite{Becker2025}.
There exists, however, a trade-off between \RA time and an additive term on top of a succinct encoding of $\rle(\mathsf{BWT}(T))$~\cite{Sinha2019}.
The lack of efficient access in $\Oh(z(T))$ space motivated the introduction of more restricted variants like LZ-End~\cite{KN10}, which support \RA in $\Oh(\poly\log n)$ time~\cite{KS22}, height-bounded LZ~\cite{Bannai2024}, bounded-access-time LZ~\cite{Liptak2024}, and LZBE~\cite{SNY25}.
Further works for \RA in grammar-compressed strings include schemes providing faster access to characters that are close to static bookmarks~\cite{CGW16} or dynamic fingers~\cite{BCC18,G21}, or located in incompressible regions~\cite{CLG26}.
Very recently, \RA has also been studied for two-dimensional grammar-compressed strings~\cite{DK26}.

\input{open-problems.tex}

%% file: open-problems.tex
\subsection{Open Problems}
The most important open problem about grammar-compressed \RA remains whether $o(\log n)$-time queries can be achieved in $\Oh(g)$ space for any $g\le n^{1-\Omega(1)}$.
We hope that our direct reduction from BLSD will inspire renewed efforts to tackle this question from the lower-bound side.
A more structured variant that seems as difficult as the general case is when each production is of length two and each expansion length is a power of two; in that case, no balancing is needed and $\child_A$ queries remain trivial for every $\tau$.
With appropriate parameters, our lower bound of \cref{thm:lower} only produces grammars of this form, albeit ones supporting $\Oh(\smash{\frac{\log n}{\log \log n}})$-time queries.

Another task stemming from the limitations of \cref{thm:lower} is to understand whether SLG \RA becomes easier for very small $g$ (e.g., $g\le \Oh(\log n)$) in the word RAM model.

Since every SLG can be encoded in $\Oh(g\log g)$ bits, for very small grammars ($g\le n^{o(1)}$) it is also open to achieve efficient (e.g., polylogarithmic-time) \RA queries in $\Oh(g\log g)$ bits.

Compressibility measures closely related to $g(T)$ pose many further interesting and long-standing problems. In particular, it is not known how to answer \RA queries in $n^{o(1)}$ time and in $\Oh(z(T))$ or $\Oh(r(T))$ space (where $r(T)$ is the number of runs in the Burrows--Wheeler transform).

Very recently~\cite{KK26}, $\Omega(\smash{\frac{\log n}{\log \log n}})$ lower bounds for moderately compressible strings have been proved for many queries beyond \RA, including suffix array and inverse suffix array queries.
Can one generalize these lower bounds to almost all parameter regimes, similarly to~\cref{thm:lower}?
Further queries, such as \textsc{Longest Common Extension} and \textsc{Internal Pattern Matching}, are at least as hard as \RA, so the lower bounds translate automatically.
However, for those queries it remains open if the existing upper bounds~\cite{I17,KK23,DK24} (formulated in terms of $\delta(T)$ or $z(T)$) can be improved for weakly compressible strings and generalized to smooth interpolations with the uncompressed setting, in the spirit of \cref{thm:block}.

Finally, we leave for future work whether our techniques can help in settings where the characters at some positions are supposed to be more easily accessible than others, e.g., due to bookmarks~\cite{CGW16} or fingers~\cite{BCC18,G21}.
Currently, we only achieve this functionality by varying character weights, but the weight needs to be the same across occurrences of the character.

%% file: preliminaries.tex
\section{Preliminaries}
Let \(\Sigma\) be some finite set.
Let \(T = t_0\cdots t_{n-1}\in \Sigma^*\). We call \(T\) a \emph{string} (or \emph{sequence}) of length \(|T|\coloneqq n\) over the \emph{alphabet} \(\Sigma\).
We use character notation \(T[i] \coloneqq t_i\) for \(i \in \iv{0}{|T|}\) and substring notation \(T\iv{i}{j} \coloneqq t_i\cdots t_{j-1}\) for \(0\leq i \leq j \leq |T|\).

Let \(a^k\) for character \(a\in\Sigma\) and integer \(k\in\Zp\) denote the \(k\)-fold repetition of \(a\); we also refer to it as the $k$-th \emph{power} of $a$.
Given any sequence \(T\coloneqq a_0^{k_0}a_1^{k_1}\cdots a_{m-1}^{k_{m-1}}\in \Sigma^*\), where \(a_i\neq a_{i+1}\) for $i\in \iv{0}{m-1}$ and \(k_i\in \Zp\) for $i\in \iv{0}{m}$, we define its run-length encoding \(\rle(T) = (a_0,k_0)(a_1,k_1)\cdots (a_{m-1},k_{m-1})\). Note that \(|\rle(T)| \coloneqq m\).

\paragraph{Straight-Line Grammars.}

We call \(\G \coloneqq (\Var_\G,\Sigma_\G,\rhs_\G : \Var_\G\rightarrow \left(\Var_\G\cup \Sigma_\G\right)^*,
S,
\weight[\G]{\cdot} : \Sigma_\G \rightarrow \mathbb{Z}_{\geq 1})\) a \emph{weighted straight-line grammar} (SLG) with variables \(\Var_\G\), terminals \(\Sigma_\G\), right-hand sides \(\rhs_\G\), start symbol \(S\in \Var_\G\cup \Sigma_\G\), and weight function \(\weight[\G]{\cdot}\) if there is an order \(\prec_\G\) on \(\Var_\G\) such that \(B\prec_\G A\) holds whenever \(B\) appears in $\rhs_\G(A)$.
We call \(\Symb_\G\coloneqq\Var_\G\cup\Sigma_\G\) the set of symbols.
When clear from context we drop the subscript \(\cdot_\G\).
The \emph{size} $|\G|$ of an SLG $\G$ is the sum of the lengths of all right-hand sides.

Straight-line grammars whose right-hand sides are stored in run-length encoded form are called run-length straight-line grammars (RLSLGs).
Their semantics are unchanged; only the size is measured as the sum of the lengths of the run-length-encoded right-hand sides.
Whenever we construct an (RL)SLG from another (RL)SLG, we assume that if the input is an SLG (i.e., without runs), then the output is also an SLG.
Usually, we bound the size of every individual rule. Note that this size is measured differently for SLGs and RLSLGs.

The \emph{expansion} of a terminal \(a\in\Sigma_\G\) is \(\expand[\G]{a}\coloneqq a\).
We extend expansions recursively to variables and sequences of symbols:
\[
  \expand[\G]{A} \coloneqq
  \begin{cases}
    \expand[\G]{\rhs(A)} & \text{if } A\in\Var_\G,\\
    \bigodot_{i\in\iv{0}{|A|}}\expand[\G]{A[i]} & \text{if } A \in \Symb_\G^*.
  \end{cases}
\]
We say that a symbol or a sequence of symbols \(A\) \emph{expands to} the string \(\expand[\G]{A}\).
The grammar \(\G\) \emph{produces} the string \(\expand[\G]{S}\).
It \emph{defines} the strings \(\expand[\G]{A}\) for all \(A\in\Symb_\G\).

Let \(\Lang_\G \coloneqq \{\expand[\G]{A} : A\in\Symb_\G\}\) denote the set of strings defined by \(\G\).
If \(\Lang_\G \subseteq \Lang_\cH\), then we say that \(\cH\) defines all strings defined by \(\G\).
We call the function \(f: \Symb_{\G}\rightarrow \Symb_{\cH}\) a grammar homomorphism from \(\G\) to \(\cH\) if \(\expand[\G]{A} = \expand[\cH]{f(A)}\) holds for every \(A\in \Symb_{\G}\).
The existence of such a homomorphism implies~\(\Lang_\G \subseteq \Lang_\cH\).

If the weight function \(\weight[\G]{\cdot}:\Sigma_\G \to \Zp\) is not stated explicitly, we assume it is the unit function \(\weight[\G]{\cdot} \coloneqq 1\). We call such a grammar unweighted.
We extend the weight function recursively to variables and sequences of symbols:
\[
  \weight[\G]{A} \coloneqq
  \begin{cases}
    \weight[\G]{\rhs(A)} & \text{if } A\in\Var_\G,\\
    \sum_{i\in\iv{0}{|A|}}\weight[\G]{A[i]} & \text{if } A \in \Symb_\G^*.
  \end{cases}
\]
For every sequence of symbols \(A\), the property \(\weight{A} = \weight{\expand[\G]{A}}\geq \explen[\G]{A}\) holds.

The \emph{height} $\height_\G(A)$ of every symbol $A\in \Symb_\G$ is also defined recursively:
\[ \height_\G(A) \coloneqq \begin{cases}
0 & \text{if }A\in \Sigma_\G,\\
0 & \text{if }A\in \Var_\G\text{ and }\rhs_\G(A) = \emptystring,\\
1 + \max_i \height_\G(A_i) & \text{if }A\in \Var_\G\text{ and }\rhs_\G(A)=A_0\cdots A_{k-1} \text{ for }k \ge 1.
\end{cases}\]

\paragraph{Parse Trees.}

We call the following directed graph the \emph{parse tree} of a weighted SLG.
Here, each node is a pair \((A,a)\) consisting of a symbol \(A\) and a weighted offset \(a\) into the expansion of the start symbol.

\begin{definition}\label{def:parsetree}
  Let \(\G=(\Var,\Sigma,\rhs,S,\weight{\cdot})\) be an (RL)SLG.
  We call the directed graph \(\Parse_{\G}\) the parse tree of \(\G\) if it is the smallest graph such that:
  \begin{enumerate}
    \item  \(\Parse_{\G}\) contains a node \((S,0)\);
    \item for every node \(u \coloneqq (A,a)\) with $A\in \Var$ and \(\rhs(A)= B_0\cdots B_{k-1}\), and for every $i\in [0\dd k)$, the parse tree \(\Parse_{\G}\) contains a node \(v_i \coloneqq \left(B_i,a+\sum_{j=0}^{i-1} \weight{B_j}\right)\) and an edge \((u,v_i)\).
  \end{enumerate}
\end{definition}

Note that, for every node $(A,a)$ of the parse tree $\Parse_\G$, the height $\height(A)$ is equal to the number of edges of the longest path from $(A,a)$ to a leaf of $\Parse_\G$ (the height of the subtree rooted at $(A,a)$).

%% file: upper/overview.tex
\section{The Upper Bounds: Overview}\label{sec:overview}
This section overviews \cref{sec:structured,sec:ub,sec:traversal,sec:rank_select}, which contain the upper-bound proofs behind \cref{thm:upper,thm:retrieve,thm:ranksel}.
As discussed in \cref{sec:techniques}, the core theme is to transform an input grammar into a form that admits \emph{constant-time local navigation}, while keeping the grammar \emph{shallow} and the overall size under control.
All transformations are phrased in terms of grammar homomorphisms, so they preserve expansions and let us switch to more structured grammars producing the same strings.

\paragraph{\RA in Weighted Strings and Parse-Tree Viewpoint.}
Most of our results are stated for strings over weighted alphabets $(\Sigma,\weight{\cdot})$.
In that setting, a random access query is given $i\in \iv{0}{\weight{T}}$ and returns $T[j]$ for the unique position $j\in \iv{0}{|T|}$ such that
$\weight{T[0\dd j)} \le i < \weight{T[0\dd (j+1))}$.
If needed, we can also output the offset $\weight{T[0\dd j)}$ and the position $j$.

\RA in a string $T$ produced by an (RL)SLG $\G$ corresponds to descending the parse tree $\Parse_\G$:
we maintain a node $(A,a)$ with $i\in [a\dd a+\weight{A})$ and repeatedly move to the child $(B,b)$ whose interval $[b\dd b+\weight{B})$ contains $i$, until a terminal is reached (\cref{alg:random}).
This transition step is implemented by a $\child_A$ query.
For $\rhs_\G(A)=B_0\cdots B_{k-1}$, $\child_A$ reduces to a rank query in
$X_A=\{\weight{B_0\cdots B_{j-1}} : j\in [1\dd k]\}$ on input $i-a$ (\cref{lemma:nicechild}).
Thus, the overall query time is governed by (i)~the per-node cost of $\child_A$ and (ii)~the depth of $\Parse_\G$; the goal is to optimize both.

\paragraph{Contracting Grammars.}
The first step is to obtain a grammar in which every variable shrinks in a controlled way along parse-tree edges.
Ganardi~\cite{G21} introduced \emph{contracting} SLGs, requiring that every child $B$ of a variable $A$ satisfies $\weight{B}\le \weight{A}/2$.
He proved that every SLG $\G$ can be homomorphically embedded into a contracting SLG of size $\Oh(|\G|)$ and production lengths bounded by some constant $d$.
Although the final result in \cite{G21} is established for the unweighted case, the key ingredients of its proof already use weighted SLGs.
In \cref{theorem:weighted-contracting-slg}, we extend the statement and the remaining arguments to fully incorporate the weighted perspective.
This extension makes it easy to lift the result to RLSLGs in \cref{cor:contracting_rlslg}.\footnote{A parallel independent work~\cite{CLG26} proves the same result, albeit for unweighted RLSLGs only.}
We treat these results as black boxes, relying only on the constant bound $d$ on right-hand side size (number of symbols or runs of symbols).

\paragraph{From Contracting to Nice.}
The contracting property already implies logarithmic height $h(A)\le \Oh(\log\explen[\G]{A})$ for every symbol $A\in \Symb_\G$, but to obtain the desired trade-off we need faster shrinkage.
For $\tau>1$, a child $B$ of $A$ is \emph{$\tau$-heavy} if $\weight{B}>\frac1\tau\weight{A}$, and $\tau$-light otherwise (\cref{def:tau-heavy}).
We also impose a local regularity condition that later yields $\Oh(1)$-time $\child_A$ implementation.

In \cref{def:nice}, we call a variable $A$ \emph{\(\tau\)-nice} if:
\begin{enumerate*}[label=(\arabic*)]
  \item\label{itov:nice-light} every variable child in $\rhs(A)$ is \(\tau\)-light;
  \item\label{itov:nice-short} $\rhs(A)$ has at most $2d\tau$ runs (or symbols, for SLGs);
  \item\label{itov:nice-local} every substring \(X\) of \(\rhs(A)\) of weight at most \(\frac1\tau\weight{A}\) has at most \(2d\ceil{\log \tau}\) runs (or symbols).
\end{enumerate*}
The key construction (\cref{lemma:nice-var}) starts from a \emph{contracting} grammar and replaces $\rhs(A)$ by a derived sequence $\rhs_\tau(A)\in \Symb_\G^*$, obtained by exhaustively replacing all \(\tau\)-heavy children with their right-hand sides.
The construction runs in time $\Oh(\tau)$ and ensures~\ref{itov:nice-short} and~\ref{itov:nice-local}; its proof is an induction on $\ceil{\log\tau}$ via the recursion $\tau\mapsto \tau/2$.
The only difference between $\rhs_{\tau/2}(A)$ and $\rhs_{\tau}(A)$ is that $\tau$-heavy children are replaced by their right-hand sides; thanks to the contracting property, all these replacements can happen in parallel, and there are $\Oh(\tau)$ of them in total and $\Oh(1)$ in a substring of weight at most \(\frac1\tau\weight{A}\).
Each replacement produces at most $d$ runs, so the bounds in~\ref{itov:nice-short} and~\ref{itov:nice-local} grow by $\Oh(d\tau)$ and $\Oh(d)$ compared to $\tau/2$.

We apply this process to all variables, with $\tau_r$ for the starting symbol and $\tau_v$ for the remaining ones, to obtain what we call a \((\tau_r,\tau_v)\)-nice grammar of size $\Oh(\tau_r+|\G|\tau_v)$.
In the resulting parse tree, the depth of every node $(A,a)$ is at most
$ 2+\max\!\left(0,\ \log_{\tau_v}\frac{\weight{S_\G}}{\tau_r\weight{A}}\right)$; see \cref{lem:nice-height}.

\paragraph{Leafy Grammars: Packing \texorpdfstring{\boldmath $b$}{b} Characters per Leaf.}
To exploit bit-parallelism when $\sigma$ is small, we introduce \emph{$b$-leafy} grammars (\cref{def:leafy}).
Here variables split into:
\begin{enumerate*}[label=(\roman*)]
  \item \emph{leaf variables} $\Varleaf{\cH}$ expanding to explicit terminal strings of length in $[b\dd 2b)$, and
  \item \emph{top variables} $\Vartop{\cH}$ whose right-hand sides contain only (top and leaf) variables.
\end{enumerate*}
The top part $\Top{\cH}$ treats leaf variables as terminals and captures the global structure, while leaf variables store the explicit blocks (packed into \(\Oh(b\log\sigma)\) bits, so \(\Oh(1)\) words for our choices of $b$).
The construction in \cref{lemma:leafyconstruct} guarantees $|\Top{\cH}|,|\Varleaf{\cH}|\in \Oh(|\G|)$ and runs in time \(\Oh(|\G|\cdot (1+b\log\sigma/w))\): we maintain only $\Oh(1)$ helper variables per original variable and generate each explicit block directly as a packed string.
The produced top rules have a deliberately restricted form: the homomorphic image of each variable of $\G$ with expansion length at least $b$ is a top variable whose right-hand side consists of at most two leaf variables plus possibly a middle top variable in between, and every helper top rule has constant run-length size.
This prevents the parse tree from containing long chains of tiny explicit pieces and is the structural reason why the traversal routines in \cref{sec:traversal} can output $\Theta(b)$ characters per step.

We then combine leafiness with niceness:
\cref{corr:constructleafyandnice} builds, from an unweighted RLSLG, a $b$-leafy grammar whose weighted top part $\Top{\cH}$ is \((\tau_r,\tau_v)\)-nice; this is where we need the preceding niceness construction to work for weighted strings.
Leafiness yields a further height reduction, giving a height bound of $3+\max\!\left(0,\ \log_{\tau_v}\frac{\weight{S_\G}}{\tau_r b}\right)$.

\paragraph{Constant-Time \texorpdfstring{$\child$}{child} Queries from Niceness.}
After the above transformations, the remaining algorithmic bottleneck is answering $\child_A$ at a \(\tau\)-nice variable $A$.
\cref{lemma:nicechild} shows that we can preprocess $A$ in $\Oh(\tau)$ time into \(\Oh(\tau\log \weight{A})\) bits so that \(\child_A\) is answered in \(\Oh(1)\) time.
At a high level, we store the run-length encoding $C_A=\rle(\rhs(A))$ and the cumulative run boundaries $P_A$ (prefix weights of runs).
Then $\child_A((A,a),i)$ reduces to computing the rank of $i-a$ in the set of run boundaries plus arithmetic operations (integer division) to locate $i$ within the run.
Condition~\ref{itov:nice-local} in \(\tau\)-niceness implies a \emph{local sparsity} property:
every interval of weighted length $\frac1\tau\weight{A}$ intersects only $\Oh(\log\tau)$ run boundaries.
This enables a standard bucketing strategy:
we partition the universe $\iv{0}{\weight{A}}$ into $\tau$ buckets, and inside each bucket store the $\Oh(\log \tau)$ boundaries in a constant-time rank structure (fusion tree~\cite{Fredman1993,PT14}).
The total size of the $\tau$ buckets is still $\Oh(\tau)$, so these rank structures take $\Oh(\tau \log\weight{A})$ bits in total, and they can be constructed in $\Oh(\tau)$ total time.

\paragraph{Random Access: Putting It Together.}
For weighted strings, \cref{lem:ra_simple} constructs a \((|\G|\tau,\tau)\)-nice grammar and the per-variable \(\child_A\) structures, and then answers queries by the parse-tree traversal above.
The time becomes the nice height bound, namely
$\Oh\!\left(\max\!\left(1,\ \log_\tau\frac{\weight{T}}{|\G|\tau\weight{c}}\right)\right)$ for the returned character $c$.
For unweighted $T\in[0\dd\sigma)^n$, we additionally apply the leafy transformation with $b=\Theta(\log_\sigma n)$ so each leaf block fits in $\Oh(1)$ words (when $w=\Theta(\log n)$), yielding \cref{lem:ra_unweighted} and the trade-off in \cref{thm:upper}.
We also prove a weighted analog of \cref{lem:ra_unweighted} in \cref{thm:ra}.

\paragraph{Traversal and Substring Extraction.}
Section~\ref{sec:traversal} generalizes random access to a pointer interface that supports constant-time forward/backward traversal.
For weighted strings, we show how to traverse characters one by one, building upon~\cite{Lohrey2017}.
For unweighted strings over small alphabets, we crucially exploit $b$-leafy grammars.
The key is to apply character pointers to the weighted top string $\Top{T}$ while maintaining an offset inside the current explicit leaf block; a character pointer to $T$ is represented as ``(a pointer to $\Top{T}$) + (an offset in the current block)''.
Because each leaf block has length $\Theta(b)$ and is stored packed, each traversal step can return $\Theta(b)$ characters in constant time.
As a result, \cref{cor:traverse} supports extracting \(m\) consecutive characters around a position in time
\(\Oh(1+m/b)\), with construction time matching the random-access trade-off; \cref{cor:substring} turns this directly into substring extraction.
The final bounds choose \(b=\Theta(\min(w,\tau\log n)/\log\sigma)\), so the block operations fit in $\Oh(1)$ words and the additive extraction term becomes \(m\log\sigma/\min(w,\tau\log n)\).

\paragraph{Rank/Select via Prefix Sums.}
Finally, \cref{sec:rank_select} lifts the entire framework from extracting characters to computing arbitrary \emph{monoid-valued} aggregates.
Given a mapping \(\Phi:\Sigma\to (\cM,\oper,\cZero)\), prefix-sum queries ask for \(\Phi(T[0\dd i))\).
The main \cref{lemma:prefsum,lemma:prefsum-leafy} mirror \RA: we store for each variable the monoid value of its expansion and combine it along the traversed path.
Rank follows by instantiating \(\Phi\) appropriately (mapping a queried character to \(1\) and the rest to \(0\)).
For select, we first transform the string so that each character encodes the number of non-occurrences between consecutive occurrences (i.e., a block of non-\(a\) characters followed by an~\(a\)), reducing \(\select_{T,a}\) to a prefix-sum query on the transformed string (\cref{claim:selectgrammar}).
To obtain the tightest bounds, we additionally build succinct rank/select structures for the short leaf blocks:
we use constant-time $\rank$ within packed bitmasks and constant-time $\select$ in the Elias--Fano encoding (\cref{lemma:bitvector_rankselect,claim:EF}), so the extra work and storage at leaves do not introduce additional overhead.

%% file: lowerbound.tex
\newcommand{\out}{\mathsf{out}}
\newcommand{\cR}{\mathcal{R}}

\section{The Lower Bound}
In this section, we prove cell-probe lower bounds for SLG \RA against both nondeterministic and bounded-error randomized query algorithms.
The lower bounds share the high-level structure and many steps, so we present them in parallel.
We first explain the formal settings.

Any data structure problem can be modeled as a function $f:\cX\times \cY\to \cZ$, where $x\in \cX$ is a query, $y\in \cY$ is an input, and $f(x,y)$ is the correct answer.
A deterministic cell-probe data structure with $M$ cells of $w$ bits consists of an encoding
\[
  D:\cY\to (\{0,1\}^w)^M
\]
and a query algorithm that probes at most $t$ cells.
Formally, the query algorithm is given by probe functions $A_0,\ldots,A_{t-1}$ and an output function $A_{\out}$, where
\[
  A_j : \cX \times (\{0,1\}^w)^j \to \iv{0}{M}
  \quad\text{for } j\in \iv{0}{t}
  \qquad\text{and}\qquad
  A_{\out} : \cX \times (\{0,1\}^w)^t \to \cZ.
\]
For a fixed query $x\in \cX$ and input $y\in \cY$, the probed addresses are
\[
  i_j \coloneqq A_j(x,D(y)[i_0],\ldots,D(y)[i_{j-1}])
  \qquad\text{for } j\in \iv{0}{t},
\]
and the returned answer is
\[
  D_{\out}(x,y) \coloneqq A_{\out}(x,D(y)[i_0],\ldots,D(y)[i_{t-1}]).
\]
Correctness means \(D_{\out}(x,y)=f(x,y)\) for all \(x\in \cX\) and \(y\in \cY\).

In the nondeterministic cell-probe model, the query algorithm additionally receives a hint \(h\) from an unbounded domain.
Equivalently, the functions \(A_0,\ldots,A_{t-1},A_{\out}\), and hence also the derived outcome function \(D_{\out}\), receive \(h\) as an extra parameter.
Moreover, the output function is allowed to return an extra value \(\bot\).
Correctness requires that for every query \(x\in \cX\) and input \(y\in \cY\),
\begin{itemize}
  \item there exists a hint \(h\) such that \(D_{\out}(x,y,h)=f(x,y)\), and
  \item for every hint \(h\), either \(D_{\out}(x,y,h)=f(x,y)\) or \(D_{\out}(x,y,h)=\bot\).
\end{itemize}
Since hints are unbounded, we may assume that a hint specifies the probed cells together with their expected contents, and the algorithm outputs \(\bot\) as soon as one probe disagrees with the hint.
Thus, a successful hint is simply a set of \(t\) cells whose contents determine the answer uniquely.
This leads to the notion of \emph{certificates} due to Wang and Yin~\cite{Wang2014}.
\begin{definition}[Certificates~\cite{Wang2014}]
A data structure problem $f:\cX\times \cY\to \cZ$ has \emph{$(M,w,t)$-certificates} if there exists a function $D:\cY\to (\{0,1\}^w)^M$ such that, for every query $x\in \cX$ and every input $y\in \cY$, there exists a set \(C\subseteq\iv{0}{M}\) with \(|C|=t\) such that, for every \(y'\in \cY\), if \(D(y)[i]=D(y')[i]\) for all \(i\in C\), then \(f(x,y')=f(x,y)\).
\end{definition}

In the randomized setting, the data structure and the query algorithm share a random seed \(r\) drawn uniformly from a finite domain \(\cR\).
Formally, both \(D\) and the functions \(A_0,\ldots,A_{t-1},A_{\out}\) receive \(r\) as an extra argument.
Hence the derived outcome function \(D_{\out}\) also depends on \(r\).
The correctness requirement is that, for every \(x\in \cX\) and \(y\in \cY\), we have \(\Pr_{r\sim U(\cR)}[D_{\out}(x,y,r)=f(x,y)]\ge \frac23\).

\begin{definition}
A data structure problem \(f:\cX\times \cY\to \cZ\) admits an \((M,w,t)\)-\emph{bounded-error structure} if there exist a finite seed domain \(\cR\), an encoding \(D:\cY\times \cR\to (\{0,1\}^w)^M\), probe functions \(A_j : \cX\times \cR\times (\{0,1\}^w)^j \to \iv{0}{M}\) for \(j\in \iv{0}{t}\), and an output function \(A_{\out} : \cX\times \cR\times (\{0,1\}^w)^t \to \cZ\) such that \(\Pr_{r\sim U(\cR)}[D_{\out}(x,y,r)=f(x,y)]\ge \tfrac23\) for every \(x\in \cX\) and \(y\in \cY\).
\end{definition}

We will reduce Blocked Lopsided Set Disjointness to SLG \RA.
Recall that, in Lopsided Set Disjointness (LSD), the goal is to decide whether two sets $X$ and $Y$ over a common universe are disjoint.
In the data-structure version, one of the sets, say $Y$, can be preprocessed, while the other set $X$ is revealed at query time.
In Blocked LSD (BLSD), the universe is partitioned into \(N\) blocks of size~\(B\), and the set \(X\) must contain exactly one element from each block.
Equivalently, an instance of BLSD can be represented by functions $S$ and $T$, so that, for each \(i\in \iv{0}{N}\), the value $S(i)\in \iv{0}{B}$ specifies the unique element of $X$ in the \(i\)-th block, whereas $T(i)\subseteq \iv{0}{B}$ specifies the intersection of $Y$ with that block.

\defpardsproblem{Blocked Lopsided Set Disjointness (BLSD)~\cite{P11}}{$N,B\in \Zp$}{
  A function \(T : \iv{0}{N} \rightarrow 2^{\iv{0}{B}}\).
}{
  Given \(S : \iv{0}{N}\rightarrow \iv{0}{B}\), decide if there exists $i\in \iv{0}{N}$ such that $S(i)\in T(i)$.
}

Wang and Yin \cite[Theorem 26]{Wang2014} proved a certificate lower bound for 2-Blocked-LSD, which imposes an additional distinctness condition within every group of \(B\) consecutive blocks, namely \(S(i)\ne S(i')\) whenever \(\floor{i/B}= \floor{i'/B}\).
Since this is a more restrictive version of BLSD, it yields the following lower bound for BLSD as well:

\begin{corollary}\label{thm:blsd-certificate}
  Consider $N,B,M,w,t\in \Zp$ such that $M>t$, and a real constant $\lambda > 0$.
  If BLSD with parameters $N,B$ has $(M,w,t)$-certificates, then \(t\ge\Omega\Big(\frac{NB^{1-\lambda}}{w}\Big)\) or \(t\ge\Omega\Big(\frac{N\log B}{\log\frac{M}{t}}\Big)\).
\end{corollary}

For a lower bound against randomized algorithms, we interpret BLSD as a communication problem with two players: Alice, who receives $S$, and Bob, who receives $T$.
Pătrașcu~\cite{P11} proved asymmetric communication complexity lower bounds for this problem:
\begin{theorem}[{\cite[Theorem 1.4]{P11}}]
Consider $N,B\in \Zp$ and a real constant $\lambda > 0$.
In any bounded-error communication protocol for BLSD with parameters $N,B$, Alice sends at least $\Omega(N \log B)$ bits or Bob sends at least $\Omega(NB^{1-\lambda})$ bits.
\end{theorem}

Next, we reduce BLSD to SLG \RA.
We follow the encoding of Verbin and Yu~\cite{VY13} with a crucial extension that allows answering a single BLSD query using several \RA queries.
The following lemma captures the original reduction of~\cite{VY13}.

\begin{lemma}[{\cite[Lemma 2]{VY13}}]\label{lem:vyblsd}
  Let $N,B\in \Zp$.
  For every function $T: \iv{0}{N}\to 2^{\iv{0}{B}}$, there exists an SLG of size at most $4NB$ producing a binary string $V$ of length $B^N$ such that, for every function $S: \iv{0}{N}\to \iv{0}{B}$, we have \[V\!\left[\sum_{i\in \iv{0}{N}} S(i)\cdot B^i\right]=\mathtt{1} \quad \text{if and only if} \quad \bigvee_{i\in \iv{0}{N}}\; S(i)\in T(i).\]
\end{lemma}

Our generalization splits the input BLSD instances into smaller subinstances before applying \cref{lem:vyblsd}.
The parameter $Q$ controls the number of \RA queries needed to answer one BLSD query.
While keeping the compressed size of the produced string $V$ at $\Theta(NB)$, this allows reducing the length of $V$ from $B^N$ to $Q\cdot B^{N/Q}$.
As a result, we obtain a smooth trade-off between storing $T$ explicitly ($NB$ bits) and storing the answers to all queries explicitly ($B^N$ bits).

\begin{lemma}\label{lem:blsd-to-ra}
  Let \(P,Q,B\in \Zp\).
  For every function $T: \iv{0}{PQ}\to 2^{\iv{0}{B}}$, there exists an SLG of size at most \(5PQB\) that produces a binary string \(V\) of length \(Q\cdot B^P\) such that,
  for every function $S : \iv{0}{PQ}\to \iv{0}{B}$, we have
  \[\bigvee_{q\in \iv{0}{Q}}\;V\!\left[qB^P+ \sum_{p\in \iv{0}{P}} S(qP+p)\cdot B^p\right]=\mathtt{1} \quad \text{if and only if }\quad \bigvee_{i\in \iv{0}{PQ}}\; S(i)\in T(i).\]
\end{lemma}
\begin{proof}
Given a BLSD instance of \(PQ\) blocks of size \(B\) each, we split the instance into \(Q\) parts of \(P\) consecutive blocks each.
Each such part, indexed with $q\in \iv{0}{Q}$, corresponds to blocks $i\in [qP\dd (q+1)P)$ and,
after shifting this interval by $qP$, can be interpreted as a BLSD instance of size $(P,B)$.
\cref{lem:vyblsd} yields a binary string $V_q$ of length $B^P$ produced by a grammar $\G_q$ of size at most $4PB$ such that
\[V_q\!\left[\sum_{p\in \iv{0}{P}} S(qP+p)\cdot B^p\right]=\mathtt{1} \quad \text{if and only if} \quad \bigvee_{p\in \iv{0}{P}}\; S(qP+p)\in T(qP+p).\]
The final string $V$ is the concatenation $V=\bigodot_{q\in \iv{0}{Q}} V_q$;
since $V_q = V[qB^P\dd (q+1)B^P)$, this immediately yields the claimed equivalence.
  Moreover, $|V|=Q\cdot B^P$ and $V$ is produced by a grammar $\G$ obtained by taking a disjoint union of grammars $\G_q$ and introducing a new start symbol (with a production of length $Q$) capturing the top-level concatenation.
  The total size of $\G$ is at most $Q\cdot 4PB + Q \le 5PQB$.
\end{proof}

\begin{example}
If \(P=Q=B=2\), then the universe consists of four blocks of size two, which we split into two subinstances on blocks \(0,1\) and \(2,3\).
The reduction builds a string \(V=V_0V_1\) of length \(2\cdot 2^2=8\), where each \(V_q\) has length \(4\).
Consider the following instance:
\[
  T(0)=\{1\},\qquad T(1)=\emptyset,\qquad T(2)=\{0\},\qquad T(3)=\{1\}.
\]
Here, the first subinstance has hits exactly for \((S(0),S(1))\in\{(1,0),(1,1)\}\), so \(V_0=\mathtt{0101}\).
Likewise, the second subinstance has hits exactly for \((S(2),S(3))\in\{(0,0),(0,1),(1,1)\}\), so \(V_1=\mathtt{1011}\).
Thus, \(V=\mathtt{01011011}\).
For a query \(S:\iv{0}{4}\to \iv{0}{2}\), the two accessed positions are
\[
  v_0=S(0)+2S(1)
  \quad\text{and}\quad
  v_1=4+S(2)+2S(3).
\]
For instance, if \(S=(1,0,1,0)\), then \((v_0,v_1)=(1,5)\), so the BLSD answer is $V[1]\vee V[5]=\mathtt{1}\vee \mathtt{0}=\mathtt{1}$.
On the other hand, if \(S=(0,1,1,0)\), then \((v_0,v_1)=(2,5)\), so the BLSD answer is $V[2]\vee V[5]=\mathtt{0}\vee \mathtt{0}=\mathtt{0}$.
\end{example}

With this reduction, we can translate certificates for random access into certificates for~BLSD.

\begin{corollary}\label{cor:blsd-to-ra}
  Consider \(P,Q,B,M,w,t\in \Zp\) such that $M \ge tQ$.
  If SLG \RA with parameters $n\coloneqq Q\cdot B^P$, $g\coloneqq 5PQB$, and $\sigma\coloneqq 2$ has $(M,w,t)$-certificates, then BLSD with parameters $N\coloneqq PQ$ and $B$ has $(M,w,tQ)$ certificates.
\end{corollary}
\begin{proof}
  Consider a BLSD input $T: \iv{0}{PQ}\to 2^{\iv{0}{B}}$ and the corresponding
  string $V$ obtained using \cref{lem:blsd-to-ra}.
  By construction, this string is an SLG \RA instance with parameters $n=Q\cdot B^P$, $g = 5PQB$, and $\sigma=2$.
  We use the hypothetical SLG \RA data structure for $V$ as a data structure for BLSD queries in~$T$.
  By \cref{lem:blsd-to-ra}, every BLSD query $S$ to $T$ can be answered using $Q$ random access queries to $V$, asking for $V[v_q]$ for some positions $v_q\in \iv{0}{|V|}$ indexed with $q\in \iv{0}{Q}$.
  By our assumption, for each $q\in \iv{0}{Q}$, there is a certificate $C_q\subseteq \iv{0}{M}$ of size $t$ that determines the value of $V[v_q]$.
  Taking the union $C\coloneqq \bigcup_{q\in \iv{0}{Q}} C_q$, padded with arbitrary elements to reach size $t\cdot Q \le M$, we obtain a certificate that determines $V[v_q]$ for all $q\in \iv{0}{Q}$, and hence also their disjunction.
  By \cref{lem:blsd-to-ra}, this disjunction is exactly the BLSD answer, that is, whether $\bigvee_{i\in \iv{0}{N}} S(i)\in T(i)$ holds.
  Thus, BLSD with parameters $PQ$ and $B$ has $(M,w,tQ)$-certificates.
\end{proof}

In order to derive a counterpart of \cref{cor:blsd-to-ra} for bounded-error randomized data structures, extra care is needed to avoid accumulating error probabilities from multiple random access queries used to answer a single BLSD query.
In \cite{P11}, similar one-to-many reductions miss analogous error analysis and thus almost all cell-probe lower bounds hold only against zero-error query algorithms.
The central piece that allows supporting bounded-error structures is the following abstract result.
\begin{proposition}\label{prop:or}
  Consider a bit sequence $x_0,\ldots,x_{q-1}$ and a randomized oracle that, given any subset $A\subseteq \iv{0}{q}$, returns bits $(y_a)_{a\in A}$ such that $\Pr[y_a=x_a]\ge \frac23$ for every $a\in A$.
  The returned bits may be arbitrarily dependent within a single oracle call, but the randomness used by different oracle calls is independent.
  Then, the disjunction $\bigvee_{i\in \iv{0}{q}} x_i$ can be computed correctly with probability at least $\frac23$ using $\Oh(1+\log q)$ oracle calls, where the $j$th oracle call is made on a subset of size \(\Oh(q/2^{\Omega(j)})\).
\end{proposition}
\begin{proof}
  Our high-level strategy is to maintain and gradually shrink a subset $A\subseteq \iv{0}{q}$ such that $\bigvee_{i\in \iv{0}{q}} x_i = \bigvee_{a\in A} x_a$.
  Each step of this process is implemented using the following subroutine:
  \begin{claim}\label{claim:or-reduce}
    There exists a randomized procedure that,
    given a set \(A\subseteq \iv{0}{q}\) and an integer \(\rho\ge 2\), uses at most \(r\coloneqq \ceil{108\ln \rho}\)
    oracle calls, all on the set \(A\), and returns a subset \(B\subseteq A\) satisfying
    \[
      |B|\le \ceil{|A|/\rho}
      \qquad\text{and}\qquad
      \Pr\left[\bigvee_{i\in A}x_i=\bigvee_{i\in B}x_i\right] \ge 1-\tfrac{1}{6\rho}.
    \]
  \end{claim}
  \begin{proof}
    If \(A=\emptyset\), we return \(B=\emptyset\), which satisfies all requirements.
    Henceforth, suppose \(A\ne\emptyset\).

    We make \(r\) independent oracle calls on \(A\).
    For every \(i\in A\), let \(z_i\in \{0,1\}\) be the majority value among the \(r\) obtained bits \(y_i\).
    Construct \(C\coloneqq \{i\in A : z_i=1\}\),
    and return an arbitrary subset \(B\subseteq C\) of size \(\min(|C|,\ceil{|A|/\rho})\).

    For every \(i\in A\) and every \(\ell\in \iv{0}{r}\), let \(X_{i,\ell}\in\{0,1\}\) indicate whether the \(\ell\)-th returned bit for \(i\) is wrong.
    Then \(\mathbb{E}[X_{i,\ell}]\le \frac13\), and the random variables \(X_{i,0},\ldots,X_{i,r-1}\) are independent.
    Since \(z_i\ne x_i\) implies \(\sum_{\ell=0}^{r-1} X_{i,\ell}\ge \frac{r}{2}\), Hoeffding's inequality yields
    \[
      \Pr[z_i\ne x_i]
      \le
      \Pr\!\left[\frac{1}{r}\sum_{\ell=0}^{r-1} X_{i,\ell}-\mathbb{E}\!\left[\frac{1}{r}\sum_{\ell=0}^{r-1} X_{i,\ell}\right]\ge \frac16\right]
      \le
      e^{-2r(1/6)^2}
      =
      e^{-r/18}.
    \]
    Since \(r=\ceil{108\ln\rho}\), we conclude that the following holds for every \(i\in A\):
    \[
      \Pr[z_i\ne x_i] \le e^{-r/18}
      \le
      e^{-6\ln \rho}
      =
      \rho^{-6}
      \le
      \tfrac{1}{16\rho^2}.
    \]

    If \(\bigvee_{i\in A}x_i=0\), then every subset \(B\subseteq A\) satisfies \(\bigvee_{i\in B}x_i=0=\bigvee_{i\in A}x_i\).
    It remains to consider the case \(\bigvee_{i\in A}x_i=1\).
    Define the set of false positives
    \[
      A_{\mathsf{FP}}\coloneqq \{i\in A : x_i=0 \text{ and } z_i=1\}.
    \]
    Then,
    \(
      \mathbb{E}[|A_{\mathsf{FP}}|]\le \frac{|A|}{16\rho^2}
    \), and by Markov's inequality,
    \[
      \Pr\!\left[|A_{\mathsf{FP}}|\ge \tfrac{|A|}{\rho}\right]
      \le
      \tfrac{|A|}{16\rho^2}\cdot \tfrac{\rho}{|A|}
      \le
      \tfrac{1}{16\rho}.
    \]
    Moreover, for a fixed $j\in A$ with $x_j=1$, we have
    \[
      \Pr[z_j=0]\le \tfrac{1}{16\rho^2}\le \tfrac{1}{16\rho}.
    \]
    Hence, with probability at least \(1-\frac{1}{16\rho}-\frac{1}{16\rho}>1-\frac{1}{6\rho}\), we have both \(z_j=1\) and \(|A_{\mathsf{FP}}|< \frac{|A|}{\rho}\).
    Conditioned on this event, \(C\) contains at least one true positive and strictly fewer than \(\frac{|A|}{\rho}\) false positives.
    Therefore, every subset of \(C\) of size \(\min(|C|,\ceil{|A|/\rho})\) still contains a true positive, and so the returned set \(B\) satisfies \(\bigvee_{i\in A}x_i=\bigvee_{i\in B}x_i\).
  \end{proof}

  Let us continue with the proof of \cref{prop:or}.
  We define \(L\coloneqq \ceil{\log(1+\log q)}\)  and construct a descending sequence of sets \(\iv{0}{q}\eqqcolon A_0 \supseteq A_1 \supseteq \cdots \supseteq A_L\)
  so that, for each \(b\in \iv{0}{L}\), the set \(A_{b+1}\) is obtained from \(A_b\) by applying \cref{claim:or-reduce} with parameter  \(\rho_b\coloneqq 2^{2^b}\).

  Since \(\ceil{\ceil{m/\rho}/\rho'}=\ceil{m/(\rho\cdot \rho')}\) holds for all positive integers $m,\rho,\rho'\in \Zp$, an induction on $b\in \fragmentcc{0}{L}$ shows that
  \[
    |A_b|\le \left\lceil \frac{q}{\prod_{i=0}^{b-1} \rho_i}\right\rceil = \left\lceil \frac{q}{2^{2^b-1}}\right\rceil
  \]
  holds for every $b\in \fragmentcc{0}{L}$ and, in particular,
  \[
    |A_{L}| \le
    \left\lceil \frac{q}{2^{2^L-1}}\right\rceil
    \le
    \left\lceil \frac{q}{2^{2^{\log(1+\log q)}-1}}\right\rceil
    =
    \left\lceil \frac{q}{2^{1+\log q-1}}\right\rceil
    =
    1.
  \]

  If \(A_{L}=\emptyset\), we output \(0\).
  Otherwise, we make \(\Oh(1)\) independent oracle calls on the singleton set \(A_L\), and output the majority value among the obtained estimates for the unique element of \(A_L\).

  If all reduction steps succeed, then \(\bigvee_{i\in \iv{0}{q}}x_i=\bigvee_{i\in A_{L}}x_i\).
  By \cref{claim:or-reduce} and the union bound, the probability that all reduction steps succeed is at least
  \[
    1-\sum_{b=0}^{L-1}\tfrac{1}{6\rho_b}
    \ge
    1-\tfrac{1}{6}\sum_{b\ge 0}2^{-2^b}
    >
    1-\tfrac{1}{6}.
  \]
  The final majority vote on \(A_{L}\) is correct with probability at least \(\frac56\), after increasing the hidden constant in its \(\Oh(1)\) repetitions if necessary.
  Therefore, the overall success probability is at least \(1-\frac{1}{6}-\frac{1}{6}=\frac23\).

  It remains to justify the complexity bound.
  In batch \(b\in \iv{0}{L}\), the algorithm queries the set \(A_b\), which we know to be of size \(\Oh\!\left(\frac{q}{2^{2^b}}\right)\), and it makes \(\Oh(\log \rho_b)\) oracle calls.
  The number of oracle calls used in the reduction steps is proportional to
  \[
    \sum_{b=0}^{L-1} \log\rho_b = \sum_{b=0}^{L-1} 2^b \le \Oh(2^L)\le \Oh(\log q).
  \]
  Together with the final \(\Oh(1)\) singleton calls, this gives \(\Oh(1+\log q)\) oracle calls in total.

  Finally, enumerate all oracle calls in chronological order.
  If the \(s\)th oracle call belongs to batch~\(b\), then the first \(b+1\) batches contain only \(\Oh(2^b)\) oracle calls in total, so \(s\le \Oh(2^b)\), and therefore \(2^b\ge \Omega(s)\).
  Hence, this call is made on a subset of size \(\Oh\!\left(\frac{q}{2^{2^b}}\right)\le \frac{q}{2^{\Omega(s)}}\).
  If the \(s\)th oracle call is one of the final singleton calls, then \(s\le \Oh(1+\log q)\), and its queried set has size \(1\le \Oh(q/2^{\Omega(s)})\), after decreasing the constant hidden in the \(\Omega(s)\) exponent if necessary.
\end{proof}

We are now ready to provide a counterpart of \cref{cor:blsd-to-ra} for bounded-error data structures.
\begin{corollary}\label{cor:blsd-to-ra-bp}
  Consider \(P,Q,B,M,w,t\in \Zp\) such that $M \ge Q$.
  Suppose that SLG \RA with parameters $n\coloneqq Q\cdot B^P$, $g\coloneqq 5PQB$, and $\sigma\coloneqq 2$ admits an $(M,w,t)$-bounded-error data structure.
  Then, BLSD with parameters $N\coloneqq PQ$ and $B$ admits a bounded-error communication protocol in which Alice sends $\Oh(tQ \log\frac{2M}{Q})$ bits and Bob sends $\Oh(tQw)$ bits.
\end{corollary}
\begin{proof}
  Fix a BLSD input \(T : \iv{0}{PQ}\to 2^{\iv{0}{B}}\) held by Bob.
  Let \(V\) be the binary string produced from \(T\) using \cref{lem:blsd-to-ra}.
  This string constitutes an SLG \RA instance with parameters $n = Q\cdot B^P$, $g = 5PQB$, and $\sigma = 2$; by assumption, it thus admits an \((M,w,t)\)-bounded-error data structure.

  For the BLSD query \(S:\iv{0}{PQ}\to \iv{0}{B}\), held by Alice, and for every \(q\in \iv{0}{Q}\), define
  \[
    v_q \coloneqq qB^P+\sum_{p\in \iv{0}{P}} S(qP+p)\cdot B^p
    \qquad\text{and}\qquad
    x_q \coloneqq V[v_q].
  \]
  By \cref{lem:blsd-to-ra}, the BLSD answer is exactly \(\bigvee_{q\in \iv{0}{Q}} x_q\).
  Hence, by \cref{prop:or}, it suffices to implement the oracle from that lemma with communication costs matching the claimed bounds.

  Consider one oracle call on a subset \(A\subseteq \iv{0}{Q}\).
  Using public randomness, Alice and Bob sample one random seed \(r\) for the data structure.
  Bob constructs the data structure \(D(V,r)\), and together they simulate all \RA query executions asking for the positions \(v_q\) for \(q\in A\).
  The seed \(r\) is sampled independently for each oracle call, so different oracle calls are independent.
  Within one oracle call, the queried bits may be dependent, which is allowed by \cref{prop:or}.

  We simulate the \(t\) probe rounds in lockstep.
  In round \(\ell\in \iv{0}{t}\), Alice knows the queried positions \(v_q\), the seed \(r\), and all answers received in earlier rounds.
  Hence, she can determine the set \(C_\ell\subseteq \iv{0}{M}\) of all memory cells that at least one of the queries requests in round \(\ell\).
  Since \(|C_\ell|\le |A|\le Q\le M\), the number of possibilities for \(C_\ell\) is \(\binom{M}{\le |A|}\coloneqq \sum_{s=0}^{|A|}\binom{M}{s},\)
  so Alice can encode \(C_\ell\) using \(\big\lceil{\log \binom{M}{\le |A|}}\big\rceil\le \Oh(|A|\log\tfrac{2M}{|A|})\) bits and send this encoding to Bob.
  Bob replies with the contents of the cells in \(C_\ell\), in a canonical order, using \(\Oh(|C_\ell|w)\le \Oh(|A|w)\) bits.
  After \(t\) rounds, Alice can reconstruct the outputs \(y_q\) of all simulated \RA queries for \(q\in A\).

  For each \(q\in A\), the value \(y_q\) equals \(x_q\) with probability at least \(\frac23\), because each simulated query is exactly one execution of the bounded-error \RA data structure on \(V\).
  Thus, this simulation realizes the oracle required by \cref{prop:or}.

  Consequently, \cref{prop:or} yields a bounded-error algorithm for \(\bigvee_{q\in \iv{0}{Q}} x_q\), and hence for the BLSD answer, that uses \(\Oh(\log Q)\) oracle calls in total.
  By \cref{prop:or}, there exist constants \(C,c>0\) such that the subset size \(a_j\) in the \(j\)th oracle call satisfies \(a_j\le C\cdot Q/2^{cj}\) for each \(j\).
  The total number of bits sent by Bob is therefore at most
  \[
    \Oh\!\left(tw \cdot \sum_j a_j\right)\le \Oh\!\left(twQ\cdot \sum_{j=0}^\infty 2^{-cj}\right)\le \Oh(tQw).
  \]

  For Alice, the total communication is
  \[
    \Oh\!\left(t\sum_j a_j\log\tfrac{2M}{a_j}\right).
  \]
  We choose \(j_0\coloneqq \ceil{\frac{1}{c}\log(Ce)}\) so that \(a_j \le C\cdot Q / 2^{cj}\le \frac{Q}{e}\) for every \(j\ge j_0\)
  and utilize the fact that \(u\mapsto u\log\tfrac{Q}{u}\) is increasing on \((0,\frac{Q}{e}]\) and decreasing on \([\frac{Q}{e},Q]\).

  For \(j<j_0\), we use only the trivial bound \(a_j\le Q\) to obtain
  \[
    a_j\log\tfrac{2M}{a_j}
    =
    a_j\log\tfrac{2M}{Q}+a_j\log\tfrac{Q}{a_j}
    \le
    Q\log\tfrac{2M}{Q}+\tfrac{Q}{e}\log e
    \le
    \Oh\!\left(Q\log\tfrac{2M}{Q}\right),
  \]
  because \(M\ge Q\) implies \(\log\tfrac{2M}{Q}\ge 1\).

  For \(j\ge j_0\), we use \(a_j \le C\cdot Q/2^{cj}\) to obtain
  \[
    a_j\log\tfrac{2M}{a_j}
    =
    a_j\log\tfrac{2M}{Q}+a_j\log\tfrac{Q}{a_j}
    \le
    \tfrac{CQ}{2^{cj}}\log\tfrac{2M}{Q}
    +
    \tfrac{CQ(cj-\log C)}{2^{cj}}
  \le
    \Oh\!\left(\tfrac{Q}{2^{cj}}\log\tfrac{2M}{Q}\right)
    +
    \Oh\!\left(\tfrac{Qj}{2^{cj}}\right).
  \]
  Since the series \(\sum_{j=0}^\infty 2^{-cj}\) and \(\sum_{j=0}^\infty j2^{-cj}\) converge, the tail contributes
  \[
    \Oh\!\left(Q\log\tfrac{2M}{Q}\right)+\Oh(Q)
    =
    \Oh\!\left(Q\log\tfrac{2M}{Q}\right).
  \]
  Adding the \(\Oh(1)\) values of \(j<j_0\), we conclude that Alice sends \(\Oh(tQ\log\tfrac{2M}{Q})\) bits, as claimed.
\end{proof}

Before we proceed, we prove a technical claim about the monotonicity of SLG \RA.

\begin{lemma}\label{lem:monot}
  Consider $n,g,\sigma,n',g',\sigma',M,w,t\in \Zp$ with $n'\ge n$, $g'\ge g+3\log n'+2$, and $\sigma'\ge \sigma$.
  \begin{itemize}
    \item If SLG \RA with parameters $(n',g',\sigma')$ has $(M,w,t)$-certificates,
    then SLG \RA with parameters $(n,g,\sigma)$ also has $(M,w,t)$-certificates.
    \item If SLG \RA with parameters $(n',g',\sigma')$ admits an $(M,w,t)$-bounded-error data structure, then SLG \RA with parameters $(n,g,\sigma)$ also admits an $(M,w,t)$-bounded-error data structure.
  \end{itemize}
\end{lemma}
\begin{proof}
  \newcommand{\pad}{\mathsf{pad}}
  Consider a string $V\in \iv{0}{\sigma}^{n}$ produced by an SLG $\G$ of size at most $g$.
  Let \(m\coloneqq n'-n\), define \(\pad(V)\coloneqq V0^m\in \iv{0}{\sigma'}^{n'}\), and let \(\pad(\G)\) be an SLG \(\G'\) producing \(\pad(V)\), constructed next.

  If $m=0$, we set $\G'\coloneqq \G$.
  Otherwise, let us denote $L\coloneqq \floor{\log m}$, set $Z_0\coloneqq 0$, and introduce fresh variables $Z_1,\ldots,Z_{L}$ with rules $Z_i\to Z_{i-1}Z_{i-1}$ for $i\in (0\dd L]$ so that $\expand[\G']{Z_i}=0^{2^i}$ for all $i\in [0\dd L]$.
  The total size of these rules is $2L$.
  Let $m=2^{i_0}+\cdots+2^{i_k}$ be the binary decomposition of $m$ (so $k\le i_k \le L$).
  We introduce a new start variable $S'$ with rule $S' \to S Z_{i_0}\cdots Z_{i_k}$, where $S$ is the start variable of $\G$.
  As a result, $\expand[\G']{S'}=\expand[\G]{S}0^m=\pad(V)$.
  Overall, the padding increases the grammar size by at most $2L+k+2\le 3L+2\le 3\log n'+2$,
  and hence $|\G'|\le g+3\log n'+2\le g'$.

  We first translate certificates.
  Let \(D'\) witness \((M,w,t)\)-certificates for \RA with parameters \((n',g',\sigma')\), and define \(D(\G)\coloneqq D'(\pad(\G))\).
  Fix an input grammar \(\G\) producing \(V\in \iv{0}{\sigma}^{n}\) and a query position \(i\in \iv{0}{n}\).
  Let \(C\subseteq \iv{0}{M}\) be the certificate guaranteed for \((\pad(V),i)\) under \(D'\).
  If another grammar \(\widehat{\G}\) producing \(\widehat{V}\in \iv{0}{\sigma}^{n}\) satisfies \(D(\G)[j]=D(\widehat{\G})[j]\) for all \(j\in C\), then
  \(D'(\pad(\G))[j]=D'(\pad(\widehat{\G}))[j]\) for all \(j\in C\).
  Since \(C\) is a certificate for \((\pad(V),i)\), this implies \(\pad(V)[i]=\pad(\widehat{V})[i]\).
  As \(i<n\), padding does not affect position \(i\), so \(V[i]=\widehat{V}[i]\).
  Thus, \(C\) is also a certificate for \((V,i)\), and \RA with parameters \((n,g,\sigma)\) has \((M,w,t)\)-certificates.

  For bounded-error data structures, let \(\cR\) be the seed domain, let \(D'\) be the randomized encoding, and let \(D'_0,\ldots,D'_{t-1},D'_{\mathsf{out}}\) be the corresponding query algorithm for \RA with parameters \((n',g',\sigma')\).
  Define \(D(\G,r)\coloneqq D'(\pad(\G),r)\), and reuse the same probe and output functions \(D'_0,\ldots,D'_{t-1},D'_{\mathsf{out}}\) for the smaller problem.
  Fix an input grammar \(\G\) producing \(V\), a query position \(i<n\), and a seed \(r\in \cR\).
  The execution on \((\G,i,r)\) in the smaller problem is identical to the execution on \((\pad(\G),i,r)\) in the larger one: both use the same memory image and the same query algorithm, and \(i<n\le n'\) is a valid query position in the larger instance.
  Hence, both executions return the same output for every seed \(r\).
  Whenever the larger data structure answers position \(i\) of \(\pad(V)\) correctly, the smaller one answers position \(i\) of \(V\) correctly as well, because \(\pad(V)[i]=V[i]\).
  Therefore, the smaller data structure inherits the same success probability \(\frac23\), and \RA with parameters \((n,g,\sigma)\) admits an \((M,w,t)\)-bounded-error data structure.
\end{proof}

In the final step, we put everything together and change the parametrization to $n,g,\sigma$.

\begin{theorem}\label{thm:bin}
Consider integers $n,g,\sigma,M,w,t\in \Zp$ and a real constant $\epsilon > 0$ such that $n\ge g$, $n\ge\sigma \ge 2$, $Mw \ge g\cdot \log n \cdot (w \log n)^\epsilon$, and $g \ge 25\cdot w^{1+\epsilon}\cdot\log n$.

If \RA with parameters $n$, $g$, and $\sigma$ has \((M,w,t)\)-certificates or admits an \((M,w,t)\)-bounded-error data structure, then \[t\ge\Omega\left(\frac{\log \frac{n}{g}}{\log \frac{Mw}{g\log n}}\right) \ge \Omega\left(\frac{\log \frac{n\log\sigma}{Mw}}{\log \frac{Mw}{g\log n}}\right).\]
\end{theorem}
\begin{proof}
  Set \(B\coloneqq 1+\floor{w^{1+\epsilon}}\), \(P\coloneqq 1+\floor{\log_B\tfrac{n}{g}}\), and \(Q\coloneqq \floor{\tfrac{g-5 \log n}{5PB}}\).
  Due to $1 \le w$, we have
  \begin{equation}\label{eq:bin-B-bounds}
    1 \le w^{1+\epsilon} \le B \le 2w^{1+\epsilon}.
  \end{equation}
  Moreover, $2\le g\le n$ implies
  \begin{equation}\label{eq:bin-P-bounds}
    \max(1,\log_B\tfrac{n}{g}) \le P \le 1+ \log_B\tfrac{n}{g} \le 1 + \log\tfrac{n}{g} \le 1 + \log n \le 2 \log n.
  \end{equation}
  The assumption \(g \ge 25\cdot w^{1+\epsilon}\cdot\log n\), together with \eqref{eq:bin-B-bounds} and \eqref{eq:bin-P-bounds}, yields
  \begin{equation}\label{eq:bin-Q-positive}
    Q = \left\lfloor\frac{g-5 \log n}{5PB}\right\rfloor \ge \left\lfloor\frac{25\cdot w^{1+\epsilon}\cdot\log n-5 \log n}{20\cdot w^{1+\epsilon}\cdot \log n}\right\rfloor = \left\lfloor\frac{5\cdot w^{1+\epsilon}-1}{4\cdot w^{1+\epsilon}}\right\rfloor\ge 1.
  \end{equation}
  Since \(\frac12x < \floor{x}\le x\) holds for \(\floor{x}\ge 1\), the assumption \(g\ge 25 \cdot w^{1+\epsilon}\cdot \log n \ge 25 \log n\) further yields
  \begin{equation}\label{eq:bin-Q-bounds}
    \frac{2g}{25PB} = \frac{g-\tfrac15g}{10PB} \le \frac{g-5 \log n}{10PB} < Q \le \frac{g-5 \log n}{5PB} < \frac{g}{5PB}.
  \end{equation}

  Define \(n'\coloneqq Q\cdot B^P\) and \(g'\coloneqq 5PQB\).
  Combining \eqref{eq:bin-Q-bounds} and \eqref{eq:bin-P-bounds} implies
  \begin{equation}\label{eq:bin-nprime-bound}
    n'=Q\cdot B^P \le \frac{g}{5PB} \cdot B^{1+\log_B\frac{n}{g}} = \frac{g}{5PB}\cdot B\cdot \frac{n}{g} = \frac{n}{5P}\le n.
  \end{equation}
  Moreover, \eqref{eq:bin-Q-bounds} yields
  \begin{equation}\label{eq:bin-gprime-bound}
    g'=5PQB\le g-5 \log n \le g - 3\log n - 2.
  \end{equation}
  The bounds in \eqref{eq:bin-nprime-bound} and \eqref{eq:bin-gprime-bound}
  let us apply \cref{lem:monot} to conclude that \RA with parameters \((n',g',2)\) has \((M,w,t)\)-certificates or admits a bounded-error data structure.

  If \(M \le tQ\), then \eqref{eq:bin-Q-bounds}, \eqref{eq:bin-B-bounds}, \eqref{eq:bin-P-bounds}, and the theorem assumption \(Mw \ge g\log n \cdot (w \log n)^\epsilon\) imply
  \begin{equation}\label{eq:bin-small-M-help}
    t \ge \frac{M}{Q} > \frac{5MPB}{g} \ge \frac{5Mw^{1+\epsilon}}{g} > \frac{Mw}{g} \ge \log^{1+\epsilon}n \cdot w^\epsilon > \log n.
  \end{equation}
  Since \(\log n \ge \log \frac{n}{g}\) and \(\log \frac{Mw}{g\log n}\ge \epsilon\log (w\log n) \ge \Omega(1)\), this further yields
  \begin{equation}\label{eq:bin-small-M}
    t > \log n \ge \Omega\Bigg(\frac{\log \frac{n}{g}}{\log \frac{Mw}{g\log n}}\Bigg).
  \end{equation}

  We may henceforth assume \(M > tQ\).
  Set \(\lambda\coloneqq \frac{\epsilon}{1+\epsilon}\) and \(N\coloneqq PQ\).

  If \RA with parameters \((n',g',2)\) has \((M,w,t)\)-certificates, then \cref{cor:blsd-to-ra} implies that BLSD with parameters \((N,B)\) has \((M,w,tQ)\)-certificates.
  Hence, \cref{thm:blsd-certificate} yields
  \begin{equation*}
    tQ\ge\Omega\!\left(\frac{NB^{1-\lambda}}{w}\right)
    \quad\text{or}\quad
    tQ\ge\Omega\!\left(\frac{N\log B}{\log\frac{M}{tQ}}\right)
    \ge
    \Omega\!\left(\frac{N\log B}{\log\frac{2M}{Q}}\right).
  \end{equation*}

  If instead \RA with parameters \((n',g',2)\) admits an \((M,w,t)\)-bounded-error data structure, then \cref{cor:blsd-to-ra-bp} yields a bounded-error communication protocol for BLSD with parameters \((N,B)\) in which Alice sends \(\Oh\left(tQ\log\tfrac{2M}{Q}\right)\) bits and Bob sends \(\Oh(tQw)\) bits.
  The randomized BLSD lower bound therefore implies
  \[
    tQ\ge\Omega\!\left(\frac{NB^{1-\lambda}}{w}\right)
    \quad\text{or}\quad
    tQ\ge\Omega\!\left(\frac{N\log B}{\log\frac{2M}{Q}}\right).
  \]

  In either case, dividing by \(Q\) and using \(N=PQ\), we obtain
  \begin{equation}\label{eq:bin-two-cases}
    t\ge\Omega\!\left(\frac{PB^{1-\lambda}}{w}\right)
    \quad\text{or}\quad
    t\ge\Omega\!\left(\frac{P\log B}{\log\frac{2M}{Q}}\right).
  \end{equation}

  We first handle the former case.
  By \eqref{eq:bin-B-bounds} and the definition of \(\lambda=\frac{\epsilon}{1+\epsilon}\), we have
  \[
    B^{1-\lambda}\ge w^{(1+\epsilon)(1-\lambda)} = w,
  \]
  and hence \eqref{eq:bin-two-cases} implies \(t\ge\Omega(P)\).
  Moreover, the theorem assumption gives
  \begin{equation}\label{eq:bin-space-ratio}
    \tfrac{Mw}{g\log n} \ge (w\log n)^\epsilon \ge  (w\log n)^{\Omega(1)}\ge w^{\Omega(1)}.
  \end{equation}
  Since \eqref{eq:bin-B-bounds} implies \(B\le w^{\Oh(1)}\), \eqref{eq:bin-space-ratio} yields \(\log B \le \Oh(\log \frac{Mw}{g\log n})\).
  Therefore, using \eqref{eq:bin-P-bounds}, we obtain
  \begin{equation}\label{eq:bin-first-case}
    t \ge \Omega(P)\ge \Omega\Bigg(\frac{\log\frac{n}{g}}{\log B}\Bigg)\ge \Omega\Bigg(\frac{\log \frac{n}{g}}{\log \frac{Mw}{g\log n}}\Bigg).
  \end{equation}

  It remains to consider the second alternative in \eqref{eq:bin-two-cases}.
  Let us first focus on the logarithmic denominator.
  Due to \eqref{eq:bin-Q-bounds}, \eqref{eq:bin-B-bounds}, and \eqref{eq:bin-P-bounds}, we have
  \[
    \log\tfrac{2M}{Q} \le \log\tfrac{25MPB}{g}
      \le \log\tfrac{M \log n\cdot w^{1+\epsilon}}{g}+\Oh(1)\le \log\tfrac{M w}{g\log n}+\log(w^\epsilon \cdot \log^2 n)+\Oh(1).
  \]
  By \eqref{eq:bin-space-ratio},
  \begin{equation}\label{eq:bin-denom}
    \log\tfrac{2M}{Q} \le \log\tfrac{M w}{g\log n}+ \log(w^\epsilon \cdot \log^2 n)+\Oh(1) \le \Oh\!\left(\log\tfrac{M w}{g\log n}\right).
  \end{equation}
  Combining \eqref{eq:bin-two-cases}, \eqref{eq:bin-P-bounds}, and \eqref{eq:bin-denom}, we conclude that
  \begin{equation}\label{eq:bin-second-case}
    t\ge\Omega\bigg(\frac{P\log B}{\log \frac{2M}{Q}}\bigg)\ge \Omega\bigg(\frac{\log\tfrac{n}{g}}{\log \frac{Mw}{g\log n}}\bigg).
  \end{equation}

  Equations \eqref{eq:bin-small-M}, \eqref{eq:bin-first-case}, and \eqref{eq:bin-second-case} together prove the first inequality in the theorem.
  For the second inequality, it suffices to plug in the following bound:
  \begin{equation}\label{eq:bin-second-claim}
    \log\tfrac{n}{g}
    =
    \log\left(\tfrac{n\log \sigma}{Mw}\cdot \tfrac{Mw}{g \log n}\cdot \tfrac{\log n}{\log \sigma}\right)
    \ge \log\tfrac{n \log \sigma}{Mw},
  \end{equation}
  where the inequality uses the assumptions \(Mw > g \log n\) and \(n \ge \sigma \ge 2\).
\end{proof}

Note that \cref{thm:bin} does not apply in the following cases:
\begin{description}
	\item[\boldmath $g < 25 \cdot w^{1+\epsilon} \cdot \log n$.]
	This corresponds to very small SLGs, which do not model BLSD inputs.
	\item[\boldmath $Mw < g \log n \cdot (w\log n)^{\epsilon}$.]
	In that case, we can derive a weaker lower bound by monotonicity with respect to $M$ (the certificate size can only increase as the data structure size decreases, with everything else unchanged).
  Applying the bound for $M \coloneqq \ceil{g \cdot w^{\epsilon - 1}\cdot \log^{1+\epsilon} n}$ yields
	\[ t \ge \Omega\bigg(\frac{\log \frac{n}{g}}{\log (w\log n)}\bigg).\]
  The same bound could have been obtained by directly adapting the proof of \cref{thm:bin}.
\end{description}

Restating \cref{thm:bin} in terms of word RAM algorithms, we obtain \cref{thm:lower} from the introduction.
Here, we present its simplified version for the case of $w=\Theta(\log n)$.

\begin{corollary}\label{cor:bin-wordram}
Consider integers $n,g,\sigma,M,t\in \Zp$ and a real constant $\epsilon > 0$ such that $n\ge g$, $n\ge\sigma \ge 2$ and $g \ge \Omega(\log^{2+\epsilon} n)$.
Suppose that, for every instance of SLG \RA with parameters $n,g,\sigma$, there is a data structure of $M$ machine words of $\Theta(\log n)$ bits each whose query algorithm runs in time $t$ in the word RAM model.
Then
\[
  t\ge \Omega\left(\frac{\log \frac{n}{g}}{\log \max(\frac{M}{g},\log n)}\right).
\]
\end{corollary}
\begin{proof}
  Since \(g=\Omega(\log^{2+\epsilon} n)\), we can choose an integer \(w=\Theta(\log n)\) such that
  \[
    w \le \left(\frac{g}{25\log n}\right)^{\frac{1}{1+\epsilon}}.
  \]
  All word RAMs with word size \(\Theta(\log n)\) are equivalent up to constant factors in running times and the number of machine words, so we may assume that the data structure uses \(M\) words of \(w\) bits each and still answers queries in time \(t\), up to constant factors.
  For every query/input pair, the cells probed by the query algorithm during its execution form a certificate, because any other input agreeing on these cell contents yields the same probe sequence and the same answer.
  Hence, SLG \RA with parameters \(n,g,\sigma\) has \((M,w,t)\)-certificates.
  By the choice of \(w\), we have \(g \ge 25\cdot w^{1+\epsilon}\cdot \log n\).

  If \(Mw \ge g\log n \cdot (w\log n)^\epsilon\), then \cref{thm:bin} applies and yields
  \[
    t\ge \Omega\left(\frac{\log \frac{n}{g}}{\log \frac{Mw}{g\log n}}\right).
  \]
  Since \(w=\Theta(\log n)\), we have \(\frac{Mw}{g\log n}=\Theta(\frac{M}{g})\), and therefore
  \[
    t\ge \Omega\left(\frac{\log \frac{n}{g}}{\log \frac{M}{g}}\right).
  \]

  It remains to consider the case \(Mw < g\log n \cdot (w\log n)^\epsilon\).
  Then
  \[
    \tfrac{M}{g} < \tfrac{\log n \cdot (w\log n)^\epsilon}{w} = \log^{\Oh(1)} n,
  \]
  and hence \(\log \max(\frac{M}{g},\log n)\le\Oh(\log \log n)\).
  On the other hand, by the remark preceding the corollary,
  \[
    t \ge \Omega\left(\frac{\log \frac{n}{g}}{\log (w\log n)}\right)= \Omega\left(\frac{\log \frac{n}{g}}{\log \log n}\right).
  \]
  This proves the claimed bound also in this case.
\end{proof}

%% file: upper/prelim.tex
\section{Structured Run-Length Straight-Line Grammars}\label{sec:structured}

Sometimes it is useful to have a highly structured RLSLG.
An RLSLG in \emph{normal form} has three types of rules: the right-hand side is either empty, a single power, or a pair of exactly two symbols.

\begin{definition}\label{def:normal}
  Let \(\G= (\Var,\Sigma,\rhs,S,\weight{\cdot})\) be an RLSLG and let $\Symb_+$ be the set of symbols of $\G$ with non-empty expansions.
  If \(\rhs : \Var \rightarrow \{\emptystring\}\cup \Symb_+^2 \cup \{A^k : A\in\Symb_+, k\in \mathbb{Z}_{>2}\}\), then \(\G\) is in \emph{normal form}.
\end{definition}

A folklore result guarantees that we can efficiently convert any (RL)SLG into a normal form.

\begin{fact}\label{lemma:compute_normalform}
  Given an RLSLG \(\G\), we can compute in \(\Oh(|\G|)\) time an RLSLG \(\cH\) in normal form and a homomorphism from \(\G\) to \(\cH\) such that \(|\cH|\le \Oh(|\G|)\).
\end{fact}
\begin{remark}
  The existence of a homomorphism implies \(\Lang_\G\subseteq \Lang_\cH\), that is, \(\cH\) defines all strings that \(\G\) defines.
\end{remark}
\begin{proof}
  We process the symbols $A\in \Symb_\G$ in the topological order;
  after processing each symbol $A$, we ensure that $\cH$ contains a symbol $f(A)$ that can serve as the homomorphic image of $A$.
  If $A\in \Sigma_\G$ is a character, we add it to $\Sigma_\cH$ and set $f(A)\coloneqq A$.
  If $A\in \Var_\G$ is a variable, let $W_A$ be the sequence obtained from $\rhs_\G(A)$ by removing all symbols with empty expansions; note that \(\expand[\G]{W_A}=\expand[\G]{A}\).
  We then proceed as follows based on either $\rle(W_A)=(A_0,k_0)\cdots (A_{m-1},k_{m-1})$, where $A_0,\ldots,A_{m-1}\in \Symb_{\G,+}$ and $k_0,\ldots,k_{m-1}\in \Zp$, when $\G$ is interpreted as an RLSLG,
  or based on $W_A = A_0 \cdots A_{m-1}$ and assume $k_0=\cdots = k_{m-1}=1$, if $\G$ is interpreted as an SLG.
  \begin{enumerate}
    \item For each $i\in \iv{0}{m}$ with $k_{i}\ge 3$, we add a new variable $B_i$ to $\Var_\cH$ with $\rhs_\cH(B_i)\coloneqq (f(A_i))^{k_i}$.
    \item For each $i\in \iv{0}{m}$ with $k_{i}= 2$, we add a new variable $B_i$ to $\Var_\cH$ with $\rhs_\cH(B_i)\coloneqq f(A_i)\cdot f(A_i)$.
    \item For each $i\in \iv{0}{m}$ with $k_i = 1$, we just set $B_i\coloneqq f(A_i)$.
    \item If $m\ge 1$, we set $C_0 \coloneqq B_0$, for each $i\in \iv{1}{m}$, add a new variable $C_i$ to $\Var_\cH$ with $\rhs_\cH(C_i)\coloneqq C_{i-1} B_i$, and finally set $f(A)\coloneqq C_{m-1}$.
    \item Otherwise, we set $f(A)\coloneqq E$, where $E\in \Var_\cH$ with $\rhs_\cH(E)\coloneqq \emptystring$ is added to $\Var_\cH$ if necessary.
  \end{enumerate}
  It is easy to observe that each variable added to $\cH$ is in normal form, the total run-length size of the newly added rules is $\Oh(m)$.
  If $m=0$, then clearly \(\expand[\cH]{f(A)}=\expand[\cH]{E}=\emptystring=\expand[\G]{W_A}=\expand[\G]{A}\).
  Otherwise,
  \begin{align*}
  \expand[\cH]{f(A)}&=\expand[\cH]{B_0}\cdots \expand[\cH]{B_{m-1}}\\
  &=(\expand[\cH]{f(A_0)})^{k_0}\cdots (\expand[\cH]{f(A_{m-1})})^{k_{m-1}}\\
  &=(\expand[\G]{A_0})^{k_0}\cdots (\expand[\G]{A_{m-1}})^{k_{m-1}}\\
  &=\expand[\G]{W_A}=\expand[\G]{A},\end{align*} which means that $f$ is indeed a homomorphism.
\end{proof}

%% file: upper/contracting-weighted-slgs.tex
\subsection{Contracting SLGs for Weighted Strings}\label{sec:contracting_slg}
The key to fast traversal of a grammar is a strong bound on its height. Contracting grammars offer such a bound in every non-terminal, which makes them especially useful.
Below, we generalize them to weighted SLGs.

\begin{definition}[Contracting variables and SLGs; see \cite{G21}]\label{def:contracting}
  Consider a weighted SLG $\G$.
  A symbol $B\in \Symb_\G$ is a \emph{child} of a variable $A\in \Var_\G$ if $B$ occurs in $\rhs(A)$,
  and we call it a \emph{heavy child} of $A$ if $\weight{B}>\frac12\weight{A}$.
  A~variable $A\in \Var_\G$ is \emph{contracting} if it does not have any variable\footnote{A character (terminal symbol) is allowed as a heavy child of a contracting variable.}
 $B\in \Var_\G$ as a heavy child.
  The SLG $\G$ is \emph{contracting} if every variable $A\in \Var_\G$ is contracting.
\end{definition}

The contracting property was introduced by Ganardi for SLGs over unweighted alphabets.
He showed in \cite[Theorem 1]{G21} that, given an SLG \(\G\), one can in linear time construct a contracting SLG of size \(\Oh(|\G|)\) with constant-size right-hand sides that defines all strings that \(\G\) does.
In this section, we extend this result to SLGs over weighted alphabets, formalized as follows:

\begin{theorem}\label{theorem:weighted-contracting-slg}
  Given an SLG $\G$ over a weighted alphabet, in \(\Oh(|\G|)\) time, one can construct a homomorphism from $\G$ to a contracting SLG \(\cH\) of size \(\Oh(|\G|)\) with constant-size right-hand sides.
\end{theorem}
\begin{remark}
  Thanks to the support of weighted alphabets, this theorem can be used as a building block for lifting the \enquote{contracting} property to generalizations of SLGs, similar to how \cite{NOU22,NU24} generalize balancing SLGs~\cite{GJL21} to RLSLGs and beyond.
\end{remark}

Our proof uses most of the inner machinery of Ganardi's proof \cite[Section 3]{G21}.
Although the final result of \cite{G21} applies to unweighted SLGs only, the key ingredient we need is already formulated and proved for weighted alphabets:
\begin{theorem}[{\cite[Theorem 6]{G21}}]\label{thm:prefix_tree}
Given a tree $T$ with $m$ edges, each labeled by a string of length at most $\ell$ over a weighted alphabet, one can compute in $\Oh(m\ell)$ time a contracting SLG with $\Oh(m)$ variables and right-hand sides of length $\Oh(\ell)$ each so that, for every node $\nu$ of $T$, the SLG defines the concatenation of the edge labels on the path from the root of $T$ to the node $\nu$.
\end{theorem}

To prove \cref{theorem:weighted-contracting-slg}, we begin with a weighted SLG \(\G = (\Var_\G,\Sigma_\G,\rhs_\G,S_\G,\weight{\cdot})\).
Without loss of generality (see \cref{lemma:compute_normalform}), we assume that every right-hand side of $\G$ is of constant size.
Define a directed graph \(H\) with vertex set $\Var_\G$ and arc set $\{(A,B)\in \Var_\G \times \Var_\G : B\text{ is a heavy child of }A\}$ so that, for every non-contracting variable $A$ with heavy child $B$, there is an arc from $A$ to $B$.
Observe that every variable $A\in \Var_\G$ has out-degree at most one in $H$, and the unique outgoing arc, if any, goes to the unique heavy child $B$ with $\frac12\weight{A}<\weight{B}\le \weight{A}$.
Indeed, two heavy children of weight exceeding $\frac12\weight{A}$ cannot coexist.
Moreover, every arc $(A,B)$ follows a grammar dependency, so \(B\prec_\G A\); hence, $H$ is acyclic.
Thus, $H$ is a forest in which arcs are oriented towards the roots of~$H$.
It resembles the \emph{heavy forest} of \cite{G21}, except that we only allow variables as vertices, so that the roots of $H$ are the contracting variables~of~$\G$.

For every $A\in \Var_\G$ with $\rhs_\G(A)=B_0\cdots B_{k-1}$ and a heavy child $B_i\in \Var_\G$, the arc $(A,B_i)$ of $H$ is decorated with labels $\lambda(A)\coloneqq \rev{B_0\cdots B_{i-1}}=B_{i-1}\cdots B_{0}$ and $\rho(A)\coloneqq {B_{i+1}\cdots B_{k-1}}$.

We apply \cref{thm:prefix_tree} to each of the trees of the reverse of $H$, both with labels $\lambda$ and $\rho$.
Taking the disjoint unions across the trees, this gives contracting SLGs $\HL=(\Var_{\HL},\Symb_\G,\rhs_{\HL},S,\weight{\cdot})$ and $\HR=(\Var_{\HR},\Symb_\G,\rhs_{\HR},S,\weight{\cdot})$ over $\Symb_\G$ with constant-size right-hand sides that define the concatenations of $\lambda$ and $\rho$ labels, respectively, for every root-to-node path in the reverse of $H$.

Observe that every variable $A\in \Var_\G$ can be associated to a path $A_0 \to \cdots \to A_m$ in $H$ from $A_0\coloneqq A$ to a contracting variable $A_m\coloneqq R_A\in \Var_\G$.
Note that
\[\expand[\G]{A} = \expand[\G]{\rev{\lambda(A_{0})}\cdots \rev{\lambda(A_{m-1})} \cdot R_A \cdot \rho(A_{m-1})\cdots \rho(A_{0})}\]
follows by unfolding the productions along the path from $A$ to $R_A$ in $H$, collecting the left and right contexts at each step; since the left contexts $\lambda(A_i)$ are stored reversed, taking $\rev{\lambda(A_i)}$ recovers the correct order.
Moreover, \cref{thm:prefix_tree} guarantees that $\HL$ contains a variable $P_A\in \Var_{\HL}$ with $\expand[\HL]{P_A}=\lambda(A_{m-1})\cdots \lambda(A_{0})$ and $\HR$ contains a variable $S_A\in \Var_{\HR}$ with $\expand[\HR]{S_A}=\rho(A_{m-1})\cdots \rho(A_{0})$.
Thus, $\expand[\G]{A}=\expand[\G]{\rev{\expand[\HL]{P_A}}\cdot R_A \cdot \expand[\HR]{S_A}}$.

We assume without loss of generality that $\Var_{\HL}$ and $\Var_{\HR}$ are disjoint and define a grammar $\cH = (\Var_{\G}\cup \Var_{\HL}\cup \Var_{\HR},\Sigma_\G, \rhs_{\cH},S_\G,\weight{\cdot})$ so that
\[\rhs_{\cH}(A) = \begin{cases}
\rev{\rhs_{\HL}(A)} & \text{if }A\in \Var_{\HL},\\
\rhs_{\HR}(A) & \text{if }A \in \Var_{\HR},\\
\rev{\rhs_{\HL}(P_A)} \cdot \rhs_{\G}(R_A) \cdot \rhs_{\HR}(S_A) & \text{if }A\in\Var_{\G}.\\
\end{cases}\]
Now, the right-hand sides are of constant size (up to three times as long as in $\G$, $\HL$, and $\HR$) and $\expand[\cH]{A}=\expand[\G]{A}$ holds for all $A\in \Var_\G$, which means that the inclusion map $\Var_\G\hookrightarrow \Var_\cH$ is a homomorphism from $\G$ to $\cH$.
We next argue that variables $A\in \Var_\G$ are contracting in $\cH$.
Consider a child $B\in \Var_{\cH}$ of $A$ in $\cH$.
If $B$ occurs in $\rhs_\G(R_A)$, then $\weight{B} \le \frac12\weight{R_A}\le \frac12\weight{A}$ because $R_A$ is contracting.
If $B$ occurs in $\rhs_{\HR}(S_A)$, then either $\weight{B}\le \frac12\weight{S_A} \le \frac12\weight{A}$ or $B$ is also a heavy child of $S_A$ in $\HR$.
Since $\HR$ is contracting, the latter means that $B\notin \Var_{\HR}$.
Thus, $B$ is a symbol from $\Symb_\G$ appearing in the label $\rho(A_i)$ for some ancestor $A_i$ of $A$ in~$H$; its weight is bounded by the weight of that label, so $\weight{B}\le \weight{\rho(A_i)} \le \weight{A_i}-\weight{A_{i+1}} < \frac12 \weight{A_i} \le \frac12\weight{A}$.
If $B$ occurs in $\rev{\rhs_{\HL}(P_A)}$, the argument is~symmetric.

Next, consider a variable $A\in \Var_{\HL}\cup \Var_{\HR}$.
Although the contracting property of $\HL$ and $\HR$ guarantees that it does not have a heavy child in $\Var_{\HL}\cup\Var_{\HR}$, it may have a heavy child in $\Var_\G$.
In that case, we replace the heavy child $B$ with $\rhs_\cH(B)$ in $\rhs_\cH(A)$.
Since we already proved $B$ to be contracting in $\cH$, this fix makes $A$ contracting.
Moreover, since $A$ and $B$ belong to disjoint sets $\Var_{\HL}\cup \Var_{\HR}$ and $\Var_\G$, respectively, the fix can be applied in parallel to all non-contracting variables $A\in \Var_{\HL}\cup \Var_{\HR}$.
As we replace a single heavy child by its right-hand side and $\rhs_\cH(B)$ has constant size, the right-hand side size remains constant.

This completes the proof of the existential part of \cref{theorem:weighted-contracting-slg}.
It is straightforward to verify that all steps of our construction can be implemented in linear time.

%% file: upper/contracting-rlslgs.tex
\subsection{Contracting RLSLGs}

We claim that, for every RLSLG \(\G\) over a weighted alphabet, there is a contracting RLSLG \(\mathcal{H}\) of size \(\Oh(|\G|)\) that produces all strings that \(\G\) produces. We prove this by constructing a homomorphism from \(\G\) to such an \(\cH\).

\begin{corollary}\label{cor:contracting_rlslg}
Given an RLSLG $\G$ over a weighted alphabet, in \(\Oh(|\G|)\) time, one can construct a homomorphism from $\G$ to a contracting RLSLG \(\cH\) of size \(\Oh(|\G|)\) with constant-size (run-length encoded) right-hand sides.
\end{corollary}
\newcommand{\Runs}{\mathcal{R}}
\begin{proof}
Without loss of generality, let \(\G\) be in normal form (\cref{lemma:compute_normalform}).
Let \(\Runs_\G\) be the set of \emph{run variables} in \(\G\), that is, variables \(A\in \Var_\G\) with \(\rhs_\G(A)=B^k\) for some \(B\in\Symb_{\G}\) and \(k>2\).
We construct a weighted SLG \[\G' \coloneqq (\Var_{\G}\setminus \Runs_\G,\Sigma_\G\cup \Runs_\G,\rhs_{\G}|_{\Var_\G\setminus \Runs_\G}, S_\G, \weight[\G]{\cdot}|_{\Sigma_\G\cup \Runs_\G})~.\]
According to \cref{theorem:weighted-contracting-slg}, there is a contracting SLG \(\mathcal{H}'\) with constant-size right-hand sides defining all strings that \(\G'\) defines.
By renaming variables in \(\cH'\), we may assume without loss of generality that the homomorphism is the inclusion map \(\Var_{\G}\setminus \Runs_\G \hookrightarrow \Var_{\cH'}\).

We define the RLSLG \(\mathcal{H} := (\Var_{\mathcal{H}'}\cup \Runs_\G, \Sigma_\G, \rhs_{\cH}, S_\G,\weight{\cdot}_{\G})\),
where
\[
\rhs_{\cH}(A) \coloneqq
\begin{cases}
\rhs_{\cH'}(A) & \text{if } A\in \Var_{\cH'},\\
\rhs_{\G}(A) & \text{if } A\in \Runs_\G.
\end{cases}
\]
Now, the right-hand sides are of constant size, as in $\cH'$.
Furthermore, for every \(A\in\Var_\G\), we have
\[
  \expand[\cH]{A}=\expand[\G]{\expand[\cH']{A}}=\expand[\G]{\expand[\G']{A}}=\expand[\G]{A}.
\]
Thus, the inclusion map $\Var_\G\hookrightarrow \Var_\cH$ is a homomorphism from $\G$ to~$\cH$.
We also observe that variables $A\in \Runs_\G$ are contracting in $\cH$ because
if \(A\in \Runs_\G\), then \(\rhs_{\cH}(A)=\rhs_\G(A)=B^k\) for some \(k\ge 2\), so \(\weight{B}=\frac1k\weight{A}\le \frac12\weight{A}\).
Next, consider a variable $A\in \Var_{\cH'}$.
Although the contracting property of $\cH'$ guarantees that it does not have a heavy child in $\Var_{\cH'}$, it may have a heavy child in $\Runs_\G$.
In that case, we replace the heavy child $B$ with $\rhs_\cH(B)$ in $\rhs_\cH(A)$.
Since we already proved $B$ to be contracting in $\cH$, this fix makes $A$ contracting.
This replacement preserves expansions and increases the right-hand side length by only a constant (in run-length encoding), so it can be applied to all such occurrences.
Moreover, since $A$ and $B$ belong to disjoint sets $\Var_{\cH'}$ and $\Runs_\G$, respectively, the fix can be applied in parallel to all non-contracting variables $A\in \Var_{\cH'}$.

This completes the proof of the existential part of the corollary.
It is straightforward to verify that all steps of our construction can be implemented in linear time.
\end{proof}

%% file: upper/nice.tex
\subsection{Nice Grammars}
We generalize the notion of contracting grammars further. Based on the previous lemma we allow larger right-hand sides, which makes expansion lengths decay even faster.

\begin{definition}\label{def:tau-heavy}
  Consider a weighted (RL)SLG $\G$ and a real parameter $\tau>1$.
  We say a child $B\in \Symb_\G$ of a variable $A\in \Var_\G$ is \emph{$\tau$}-heavy if $\weight{B}> \frac1\tau \weight{A}$; otherwise, $B$ is a \emph{$\tau$-light} child of $A$.
\end{definition}
Note that the notion of a 2-heavy child coincides with that of a heavy child (\cref{def:contracting}).

The remainder of this section depends on a constant $d\in \mathbb{Z}_{>1}$, which we will ultimately require to be large enough so that the (RL)SLGs constructed using \cref{theorem:weighted-contracting-slg} and \cref{cor:contracting_rlslg} have right-hand sides satisfying $|\rhs(A)|\le d$ and $|\rle(\rhs(A))|\le d$, respectively.

\begin{definition}\label{def:nice}
  Consider a weighted RLSLG $\G$, a real parameter $\tau>1$, and a constant $d\in \Zp$.
  We say that a variable $A$ of $\G$ is \emph{$\tau$-nice} if it satisfies the following conditions:
  \begin{enumerate}[label=(\arabic*)]
    \item\label{it:light} every variable \(B\in \rhs(A)\) is \(\tau\)-light for \(A\);
    \item\label{it:short} \(|\rle(\rhs(A))|\leq 2d\tau\);
    \item\label{it:local} every substring \(X\) of \(\rhs(A)\) with \(\weight{X}\le \frac1\tau\weight{A}\) satisfies
    \(|\rle(X)|\le 2d\ceil{\log \tau}\).
  \end{enumerate}
  If we interpret $\G$ as an SLG, we replace $\rle(\rhs(A))$ in~\ref{it:short} with \(\rhs(A)\) and $\rle(X)$ in~\ref{it:local} with \(X\).
\end{definition}

We show that, for a parameter $\tau>1$ and every variable in a contracting (RL)SLG, we can construct a new right-hand side that makes the variable $\tau$-nice without changing its expansion.

\begin{lemma}\label{lemma:nice-var}
  Let \(\G\) be a contracting (RL)SLG with (run-length encoded) right-hand sides of size bounded by a constant $d\in \Zp$.
  Given a variable $A\in \Var_\G$ and a real parameter \(\tau>1\), in time \(\Oh(\tau)\),\footnote{The algorithm assumes that the weight of every symbol can be compared with $\frac1\tau\weight{A}$ in constant time.} we can compute a (run-length encoded) sequence \(\rhs_{\tau}(A)\in \Symb_\G^*\) such that replacing \(\rhs(A)\) by \(\rhs_{\tau}(A)\) turns \(A\) into a \(\tau\)-nice variable and does not change \(\expand{A}\).
\end{lemma}
\begin{proof}
  In the proof, we focus on the case of $\G$ being an RLSLG. In the case of an SLG, we simply replace the run-length encoding with the standard string representation everywhere.

  To define $\rhs_{\tau}(A)$, we start with $\rhs(A)$ and exhaustively replace every variable $B$ with $\weight{B}>\frac1\tau \weight{A}$ by $\rhs(B)$.
  We maintain the sequence in run-length encoded form, so the running time is proportional to $|\rle(\rhs_{\tau}(A))|$, which we will show to be at most $2d\tau\le \Oh(\tau)$ in the proof of~\ref{it:short}.
  If such a variable \(B\) occurs in a run \(B^k\), then \(k<\tau\) because \(k\weight{B}\le \weight{A}\), so replacing the run by \(k\) copies of the constant-size run-length-encoded right-hand side \(\rhs(B)\) still takes time proportional to the emitted encoding.

  To argue about $\rhs_{\tau}(A)$, we provide an equivalent construction by induction on \(\ceil{\log \tau}\).
  If \(\ceil{\log \tau}=1\), then \(\rhs_{\tau}(A)= \rhs(A)\).
  This is equivalent to the original definition since $\tau\le 2$ and every variable $B$ in $\rhs(A)$ satisfies $\weight{B}\le \frac12\weight{A}\le \frac1\tau\weight{A}$ because $A$ is contracting.
  If \(\ceil{\log \tau}>1\), then \(\rhs_{\tau}(A)\) can be obtained from \(\rhs_{\tau/2}(A)\) by replacing every occurrence of a \(\tau\)-heavy variable \(B\) (i.e., \(\weight{B}>\frac1\tau\weight{A}\)) in \(\rhs_{\tau/2}(A)\) with \(\rhs(B)\).
  This is equivalent because, by the definition of \(\rhs_{\tau/2}(A)\), every such variable satisfies $\weight{B}\le \frac2\tau\weight{A}$ and, since $B$ is contracting, all variables in $\rhs(B)$ are of weight at most $\frac12\weight{B}\le \frac1\tau\weight{A}$.
  Overall, the inductive definition of $\rhs_{\tau}(A)$ is equivalent to the original, which satisfies \ref{it:light} and \(\expand{\rhs_{\tau}(A)}=\expand{A}\) by~construction.

  It remains to prove \ref{it:short} and \ref{it:local}.
  If $\ceil{\log \tau}=1$, they are both true because every substring $X$ of $\rhs(A)$ satisfies $|\rle(X)|\le|\rle(\rhs(A))|\le d \le \min(2d\tau,2d\ceil{\log \tau})$.
  Assume now that $\ceil{\log \tau}>1$.
  The total weight of \(\rhs_{\tau/2}(A)\) is \(\weight{A}\), so there are strictly less than \(\tau\) occurrences of \(\tau\)-heavy symbols (counting multiplicity in runs).
  Replacing a single occurrence of \(B\) increases the run length by at most \(|\rle(\rhs(B))|\le d\).
  With at most \(\tau\) such occurrences and \(|\rle(\rhs_{\tau/2}(A))|\le d\tau\) by induction,
  we obtain \(|\rle(\rhs_\tau(A))|\le d\tau+d\tau=2d\tau\), proving \ref{it:short}.
  For \ref{it:local}, let \(X\) be any substring of \(\rhs_\tau(A)\) with \(\weight{X}\le \frac1\tau\weight{A}\).
  Since each replaced \(\tau\)-heavy symbol \(B\) has weight larger than \(\frac1\tau\weight{A}\), the substring \(X\)
  cannot fully contain \(\rhs(B)\).
  Thus, \(X\) consists of at most two partial expansions (each contributing at most \(d\) runs)
  and a substring \(X'\) of \(\rhs_{\tau/2}(A)\).
  By induction, \(|\rle(X')|\le 2d\ceil{\log \frac\tau2}\), and therefore
  \(|\rle(X)|\le 2d + 2d\ceil{\log\frac\tau2} = 2d\ceil{\log \tau}\).
\end{proof}

From this local property we get a global property.

\begin{definition}
  Consider real parameters \(\tau_r\geq\tau_v> 1\).
  We say that a weighted (RL)SLG $\G$ is \((\tau_r,\tau_v)\)-nice if the starting symbol $S_\G$ is \(\tau_r\)-nice, and the remaining variables are \(\tau_v\)-nice.
\end{definition}

By \cref{lemma:nice-var} we can, for an RLSLG \(\G\), compute such a \((\tau_r,\tau_v)\)-nice RLSLG \(\G'\) that defines all strings that \(\G\) defines in time \(\Oh(\tau_r+|\G|\cdot \tau_v)\). It has size \(\Oh(\tau_r+|\G|\cdot\tau_v)\).

\begin{corollary}\label{lemma:niceconstruct}
  Consider an (RL)SLG $\G$ producing strings of weight $2^{\Oh(w)}$, where $w$ is the machine word size.
  Given such $\G$ and real parameters \(\tau_r\geq\tau_v> 1\), in time \(\Oh(\tau_r +|\G|\cdot \tau_v)\) we can construct a homomorphism from $\G$ to a \((\tau_r,\tau_v)\)-nice grammar \(\cH\) of size \(|\cH| \le \Oh(\tau_r + |\G|\cdot\tau_v)\). It has \(|\Var_\cH|\le\Oh(|\G|)\) variables.
\end{corollary}
\begin{proof}
  First, apply the contracting transformation (\cref{theorem:weighted-contracting-slg} or \cref{cor:contracting_rlslg}) to obtain an equivalent contracting (RL)SLG of size \(\Oh(|\G|)\)
  with right-hand sides bounded by a constant~\(d\). Notice that this implies \(\Oh(|\G|)\) many variables.
  Then, apply \cref{lemma:nice-var} with \(\tau=\tau_r\) for the start symbol and with \(\tau=\tau_v\) for every other variable. This changes the right-hand sides of each variable, but does not introduce new variables.
  The resulting grammar \(\cH\) is \((\tau_r,\tau_v)\)-nice by construction.
  Since each application of \cref{lemma:nice-var} takes $\Oh(\tau)$ time to replace \(\rhs(A)\) by a sequence of length \(\Oh(\tau)\), the total size of $\cH$ is \(\Oh(\tau_r+|\G|\cdot\tau_v)\), and the total construction time is the same.
\end{proof}

We conclude this section with an upper bound on the height of a $(\tau_r,\tau_v)$-nice (RL)SLG.

\begin{lemma}\label{lem:nice-height}
  Let \(\G\) be a $(\tau_r,\tau_v)$-nice (RL)SLG for some $\tau_r \ge \tau_v >1$.
  For every node $(A,a)$ in the parse tree $\Parse_\G$, the path from the root to $(A,a)$ is of length at most \(2+\max\left(0,\;\log_{\tau_v}\frac{\weight{S_\G}}{\tau_r \weight{A}}\right)\).
\end{lemma}
\begin{proof}
  Suppose that the path is of length $\ell$ and its subsequent nodes are $(A_d,a_d)_{d=0}^\ell$, where $(A_0,a_0)=(S_\G,0)$ and $(A_\ell,a_\ell)=(A,a)$.
  We henceforth assume that $\ell \ge 3$ (otherwise, the claim holds trivially).
  The symbols $A_d$ are variables for $d\in [0\dd \ell-1]$;
  Consequently, the following hold since $\G$ is $(\tau_r,\tau_v)$-nice:
  \begin{enumerate*}[label=(\alph*)]
  \item $\weight{S_\G}=\weight{A_0} \ge \tau_r \weight{A_1}$;
  \item for $d\in[1\dd \ell-2]$, we have $\weight{A_d} \ge \tau_v \weight{A_{d+1}}$; and
  \item $\weight{A_{\ell-1}} \ge \weight{A_\ell}= \weight{A}$.
  \end{enumerate*}
  Chaining all these inequalities, we get $\weight{S_\G}\ge \tau_r \cdot \tau_v^{\ell-2}\cdot \weight{A}$,
  which immediately implies
  $\ell \le 2+\log_{\tau_v}\frac{\weight{S_\G}}{\tau_r \weight{A}}$.
\end{proof}

%% file: upper/leafy.tex
\subsection{Leafy Grammars}
Given a grammar over a small alphabet (of size $\sigma\ll 2^{w}$), we can further reduce the height of the parse tree by packing multiple consecutive characters into a single machine word.
This will be especially useful for retrieving substrings, where several characters can be reported at once.

The naive strategy is to simply replace the right-hand side $\rhs(A)$ with the expansion $\expand{A}$ for every variable $A$ for which $|\expand{A}|\le b$ holds for some threshold $b$, which we usually set to be at most $\frac{w}{\log\sigma}$ so that $\expand{A}$ can be stored in a single machine word.
This is sufficient to achieve our results for random access (\cref{sec:ub}).
Nevertheless, in the context of traversal (\cref{sec:traversal}), we want to avoid visiting parse tree nodes $(A,a)$ for which $|\expand{A}|< o(b)$, because they hinder the goal of outputting $\Theta(b)$ characters per time unit.

Hence, we transform the input RLSLG so that a character $c\in \Sigma$ can occur in $\rhs(A)$ for a variable $A\in\Var$ only if $\rhs(A)=\expand{A}$ is a string of length $\Theta(b)$.

\begin{definition}\label{def:leafy}
  Consider a (RL)SLG $\G$ and an integer $b\in [1\dd \explen[\G]{S_\G}]$. We say $\G$ is \emph{$b$-leafy} if its variables $A\in \Var_{\G}$ can be decomposed into two disjoint classes:
  \begin{itemize}
    \item \emph{leaf} variables $A\in \Varleaf{\G}$ for which $\rhs_\G(A)\in \Sigma_\G^*$ and $|\rhs_\G(A)|\in [b\dd 2b)$;
    \item \emph{top} variables $A\in \Vartop{\G}$ for which $\rhs_\G(A)\in \Var_\G^*$.
  \end{itemize}
  We define the \emph{top part} of $\G$ in which all leaf variables become terminal symbols. Formally,
  $\Top{\G} = (\Vartop{\G}, \Varleaf{\G},\rhs_\G |_{\Vartop{\G}},S_\G,\weight[\Top{\G}]{\cdot})$, where \(\weight[\Top{\G}]{A}\coloneqq\explen[\G]{A}\) for every \(A\in\Varleaf{\G}\).
\end{definition}
\begin{remark}
  A $b$-leafy (RL)SLG $\G$ can be stored in $|\Top{\G}|\log \explen[\G]{S_\G} + |\Varleaf{\G}| b \log \sigma$ bits.
\end{remark}

\medskip
We show now that we can create a \(b\)-leafy grammar from any given grammar.

\begin{lemma}\label{lemma:leafyconstruct}
  Let \(\G\) be an (RL)SLG that produces a string \(T\in \iv{0}{\sigma}^n\) and let \(b\in [1\dd n]\).

  Given $\G$ and~$b$, in time \(\Oh(|\G|\cdot (1+b \log \sigma/w))\), where $w\ge \Omega(\log n)$ is the machine word size, we can construct a \(b\)-leafy (RL)SLG \(\cH\) such that $|\Top{\cH}|,|\Varleaf{\cH}|\le \Oh(|\G|)$ and
  a partial homomorphism \(f : \{ A \in \Symb_\G : \explen{A}\geq b\} \rightarrow \Symb_\cH\) that preserves expansions.
\end{lemma}
\begin{proof}
  \cref{lemma:compute_normalform} allows us to assume without loss of generality that $\G$ is in normal form.

  We process the variables $A \in \Var_\G$ in topological order.
  For variables $A$ with $\explen[\G]{A}<2b$, we write down $\expand[\G]{A}$; these explicit strings are used when we build new leaf variables below.
  This takes $\Oh(1+b\log \sigma/w)$ time per variable using the standard implementation of packed strings, with $\ceil{\log \sigma}$ bits per character.\footnote{We use the word-RAM operation of (integer) multiplication to compute string powers of length at most $w/\log\sigma$.}
  For each variable $A$ with $\explen[\G]{A}\ge b$, we place $A$ in $\Vartop{\cH}$, and add $\Oh(1)$ helper variables to $\Vartop{\cH}$ and $\Varleaf{\cH}$ so as to maintain the following invariants:
  \begin{itemize}
    \item $\expand[\cH]{A}=\expand[\G]{A}$;
    \item $\rhs_\cH(A)\in \Varleaf{\cH}\cup \Varleaf{\cH}\cdot \Varleaf{\cH} \cup \Varleaf{\cH}\cdot \Vartop{\cH} \cdot \Varleaf{\cH}$;
    \item every newly created helper leaf variable has expansion length in $[b\dd 2b)$;
    \item every newly created helper top variable has a right-hand side of constant size (in $\rle$).
  \end{itemize}
  Note that the first two invariants involve $A\in \Var_\G$, and the last two involve helper variables.

  We achieve this using the following cases.
  \begin{description}
    \item[{\boldmath Case 1: $b \le |\expand[\G]{A}| < 2b$.}]
    We create a new variable $\Leaf{A}\in \Varleaf{\cH}$ with $\rhs_{\cH}(\Leaf{A})\coloneqq \expand[\G]{A}$ and set $\rhs_{\cH}(A)\coloneqq \Leaf{A}$.
    \item[{\boldmath Case 2: $2b \le |\expand[\G]{A}| < 3b$.}]
    We create new variables $\Leaf{A_L},\Leaf{A_R}\in \Varleaf{\cH}$ with $\rhs_{\cH}(\Leaf{A_L})\coloneqq \expand[\G]{A}[0\dd b)$
    and $\rhs_{\cH}(\Leaf{A_R})\coloneqq \expand[\G]{A}[b\dd \explen[\G]{A})$.
    We also set $\rhs_{\cH}(A)\coloneqq \Leaf{A_L}\Leaf{A_R}$.
    \item[{\boldmath Case 3: $|\expand[\G]{A}| \ge 3b$ and \(\rhs_\G(A) = BC\).}]~
    \begin{description}
      \item[{\boldmath Case 3.1: \(\explen[\G]{B},\explen[\G]{C} \ge b\).}]
      By the inductive invariant for $B$ and $C$, we can write $\rhs_\cH(B)=\Leaf{B_L}\beta$ and $\rhs_\cH(C)=\gamma \Leaf{C_R}$ for some $\Leaf{B_L},\Leaf{C_R}\in \Varleaf{\cH}$ and (possibly empty) strings $\beta,\gamma$ of at most two variables.
      If $\beta\gamma=\emptystring$, then set $\rhs_{\cH}(A)\coloneqq \Leaf{B_L}\Leaf{C_R}$.
      Otherwise, create a new variable $\Top{A}\in \Vartop{\cH}$ with $\rhs_{\cH}(\Top{A})\coloneqq \beta \gamma$ and set $\rhs_{\cH}(A)\coloneqq \Leaf{B_L}\Top{A}\Leaf{C_R}$.
      Then $\rhs_{\cH}(A)$ has the required form, $\rhs_{\cH}(\Top{A})$ has constant size when it is created, and $\expand[\cH]{A}=\expand[\G]{B}\expand[\G]{C}=\expand[\G]{A}$.
      \item[{\boldmath Case 3.2: \(\explen[\G]{B} \ge b > \explen[\G]{C}\).}]
      Then, $\explen[\G]{B}\ge 2b$, so the invariant guarantees $\rhs_{\cH}(B)=\Leaf{B_L} \beta \Leaf{B_R}$ for some $\Leaf{B_L},\Leaf{B_R}\in \Varleaf{\cH}$ and $\beta\in \Var_{\cH}\cup\{\emptystring\}$.
      \begin{description}
        \item[{\boldmath Case 3.2.1: \(\explen[\cH]{\Leaf{B_R}}+\explen[\G]{C}< 2b\).}]
        We create a new variable $\Leaf{A}\in \Varleaf{\cH}$ with $\rhs_\cH(\Leaf{A})\coloneqq \expand[\cH]{\Leaf{B_R}}\cdot \expand[\G]{C}$ and set $\rhs_{\cH}(A)\coloneqq \Leaf{B_L}\beta \Leaf{A}$.
        The new leaf has length in $[b\dd 2b)$, and $\rhs_{\cH}(A)$ has a required form.
        \item[{\boldmath Case 3.2.2: \(\explen[\cH]{\Leaf{B_R}}+\explen[\G]{C}\ge 2b\).}]
        We create new variables $\Leaf{A_L},\allowbreak \Leaf{A_R}\in \Varleaf{\cH}$ with
        $\rhs_\cH(\Leaf{A_L})\coloneqq \expand[\cH]{\Leaf{B_R}}[0\dd b)$ and
        \[\rhs_\cH(\Leaf{A_R})\coloneqq \expand[\cH]{\Leaf{B_R}}[b\dd \explen[\cH]{\Leaf{B_R}})\cdot \expand[\G]{C},\] as well as $\Top{A}\in \Vartop{\cH}$ with $\rhs_\cH(\Top{A})\coloneqq \beta\Leaf{A_L}$.
        Moreover, we set $\rhs_{\cH}(A)\coloneqq \Leaf{B_L}\Top{A}\Leaf{A_R}$.
        Here $\explen[\cH]{\Leaf{A_R}}=\explen[\cH]{\Leaf{B_R}}-b+\explen[\G]{C}\in [b\dd 2b)$, so both new leaves are valid, and $\rhs_{\cH}(\Top{A})$ has constant size.
      \end{description}
      \item[{\boldmath Case 3.3: \(\explen[\G]{B}< b \le \explen[\G]{C}\).}]
      This case is symmetric to Case 3.2.
    \end{description}
    \item[{\boldmath Case 4: $|\expand[\G]{A}| \ge 3b$ and \(\rhs_\G(A) = B^k\) for $k\ge 3$.}]~
     \begin{description}
      \item[{\boldmath Case 4.1:  \(\explen[\G]{B}\ge 2b\).}]
      The invariant guarantees that $\rhs_{\cH}(B)=\Leaf{B_L} \beta \Leaf{B_R}$ holds for some $\Leaf{B_L},\Leaf{B_R}\in \Varleaf{\cH}$ and $\beta\in \Vartop{\cH}\cup\{\emptystring\}$.
      We create a new variable $\Top{A}\in \Vartop{\cH}$ with $\rhs_{\cH}(\Top{A})\coloneqq \beta \Leaf{B_R} B^{k-2} \Leaf{B_L} \beta$ and set $\rhs_{\cH}(A)\coloneqq \Leaf{B_L}\Top{A}\Leaf{B_R}$.
      Then $\rhs_{\cH}(\Top{A})$ has constant size in run-length encoding, and $\expand[\cH]{A}=\expand[\G]{B}^k=\expand[\G]{A}$.
      \item[{\boldmath Case 4.2: \(\explen[\G]{B}< 2b\).}]
      Define $m\coloneqq \ceil{b/\explen[\G]{B}}$ so that $b \le \explen[\G]{B^m} \le 2b-1$
      and set $q \coloneqq \floor{\tfrac{\explen[\G]{A}-2b}{\explen[\G]{B^m}}}$
      so that $2b \le \explen[\G]{A}-q\cdot \explen[\G]{B^m} \le 4b-2$.
      This allows writing $\explen[\G]{A}=a_L+q\cdot \explen[\G]{B^m}+a_R$ for some $a_L,a_R\in [b\dd 2b)$, since any value in $[2b\dd 4b-2]$ can be split into two numbers in $[b\dd 2b)$.
      Create new variables $\Leaf{A_L},\Leaf{A_R}\in \Varleaf{\cH}$ and define
      $\rhs_{\cH}(\Leaf{A_L})\coloneqq \expand[\G]{A}[0\dd a_L)$ and
      $\rhs_{\cH}(\Leaf{A_R})\coloneqq \expand[\G]{A}[\explen[\G]{A}-a_R\dd \explen[\G]{A})$.
      If $q\ge 1$, create a new leaf variable $\Leaf{B}\in \Varleaf{\cH}$ and a new top variable $\Top{A}\in \Vartop{\cH}$ with
      $\rhs_{\cH}(\Leaf{B})\coloneqq \expand[\G]{A}[a_L\dd a_L+\explen[\G]{B^m})$ and
      $\rhs_{\cH}(\Top{A})\coloneqq (\Leaf{B})^q$.
      When $q\ge 1$, because $\explen[\G]{B^m}$ is a multiple of $\explen[\G]{B}$, the block $\expand[\G]{A}[a_L\dd a_L+\explen[\G]{B^m})$ repeats every $\explen[\G]{B^m}$ positions in $\expand[\G]{A}=\expand[\G]{B}^k$, so $(\Leaf{B})^q$ expands to the correct middle segment.
      Finally, set $\rhs_{\cH}(A)=\Leaf{A_L}\Top{A}\Leaf{A_R}$ when $q\ge 1$, and $\rhs_{\cH}(A)=\Leaf{A_L}\Leaf{A_R}$ when $q=0$.
  \end{description}
\end{description}
In all cases, $\rhs_{\cH}(A)$ has one of the required forms, all leaf variables have expansion length in $[b\dd 2b)$, and all new top variables have constant-size right-hand sides in run-length encoding.
Each leaf string is obtained by concatenating at most two already computed strings and/or taking a substring of length $<2b$, so it can be written in $\Oh(b\log \sigma/w)$ time.
Since we create $\Oh(|\G|)$ leaf variables, the total construction time is $\Oh(|\G|\cdot (1+b \log \sigma/w))$, as claimed.

  For the normalized grammar, the identity function on \(\{ A \in \Symb_\G : \explen{A}\geq b\}\) is the claimed partial homomorphism.
  For the original input grammar, this map should be composed with the homomorphism produced by \cref{lemma:compute_normalform}.
\end{proof}

This construction guarantees large leaves, while keeping the general structure of the grammar. We can now further manipulate this grammar to make the top part $(\tau_r,\tau_v)$-nice.

\begin{corollary}\label{corr:constructleafyandnice}
Let \(\G\) be an (RL)SLG that produces a string \(T\in [0\dd \sigma)^n\), let \(b\in [1\dd n]\), and let $\tau_r\ge \tau_v >1$.
Given $\G$, $b$, $\tau_r$, and $\tau_v$, in time \(\Oh(\tau_r + |\G|\cdot (\tau_v + b \log \sigma/w))\), where $w\ge\Omega(\log n)$ is the machine word size, we can construct a \(b\)-leafy (RL)SLG \(\cH\) such that $\Top{\cH}$ is $(\tau_r,\tau_v)$-nice and of size $|\Top{\cH}|\le\Oh(\tau_r+|\G|\tau_v)$, $|\Varleaf{\cH}|\le \Oh(|\G|)$, and
a partial homomorphism \(f : \{ A \in \Symb_\G : \explen{A}\geq b\} \rightarrow \Symb_\cH\) that preserves expansions.
\end{corollary}
\begin{proof}
  First, use \cref{lemma:leafyconstruct} to obtain a $b$-leafy (RL)SLG \(\G'\) that defines all strings of length at least $b$ that \(\G\) defines, and the partial homomorphism \(f_1\) from \(\{ A \in \Symb_\G : \explen{A}\geq b\}\) to \(\Symb_{\G'}\).
  This step takes $\Oh(|\G| (1 + b \log\sigma/w))$ time and the top part of \(\G'\) is of size $\Oh(|\G|)$.
  Next, apply \cref{lemma:niceconstruct} to the top part \(\Top{\G'}\).
  This yields a \((\tau_r,\tau_v)\)-nice grammar \(\cK\) together with a homomorphism \(f_2\) from \(\Top{\G'}\) to \(\cK\).
  We obtain \(\cH\) by taking \(\Top{\cH}\coloneqq \cK\) and, for every \(A\in \Varleaf{\G'}\), keeping \(A\) as a leaf variable of \(\cH\) with the same right-hand side as in \(\G'\).
  Since the terminals of \(\Top{\G'}\) are exactly the leaf variables of \(\G'\), the map \(f_2\) can be viewed as a map into \(\Symb_\cH\).
  We set \(S_\cH\coloneqq f_2(f_1(S_\G))\); since \(b\le n=\explen[\G]{S_\G}\), this is well-defined and expands to \(T\).
  This step takes $\Oh(\tau_r + |\G|\tau_v)$ time and results in a $b$-leafy (RL)SLG whose top part is of size $\Oh(\tau_r + |\G|\tau_v)$.
  Consequently, \(f \coloneqq f_2 \circ f_1\) is the claimed partial homomorphism.
\end{proof}

We conclude this section with an upper bound $\Oh(\max(1,\log_{\tau_v}\frac{\weight{S_\G}}{\tau_r b}))$ on the height of a $b$-leafy (RL)SLG with a $(\tau_r,\tau_v)$-nice top part.

\begin{lemma}\label{lem:leafy-nice-height}
  Consider an integer $b\ge 1$ and real numbers $\tau_r\ge \tau_v>1$.
  If \(\G\) is a $b$-leafy (RL)SLG producing an unweighted string, and \(\G\) has a $(\tau_r,\tau_v)$-nice top part,
  then its height does not exceed \(3+\max\left(0,\log_{\tau_v}\frac{\weight{S_\G}}{\tau_r b}\right)\).
\end{lemma}
\begin{proof}
  For every leaf $(A,a)$ of $\Parse_{\Top{\G}}$, we have $\weight{A}\ge \explen[\G]{A}\ge b$,
  so \cref{lem:nice-height} guarantees that the height of $\Parse_{\Top{\G}}$ does not exceed
  $2+\max(0,\log_{\tau_v}\frac{\weight{S_\G}}{\tau_r b})$.
  Since $\Top{\G}$ is obtained from $\G$ by treating leaf variables as terminals, every root-to-leaf path in $\Parse_\G$ has at most one more edge than the corresponding path in $\Parse_{\Top{\G}}$.
  We conclude that the height of $\Parse_\G$ is at most $3+\max(0,\log_{\tau_v}\frac{\weight{S_\G}}{\tau_r b})$.
\end{proof}

%% file: upper/bound.tex
\section{Upper Bound for \RA}\label{sec:ub}

In this section, we explain how to efficiently answer \RA queries for strings produced by (RL)SLGs.
We first generalize these queries to weighted strings.

\defdsproblem{\RA Queries to Weighted Strings}{
  A string $T$ over a weighted alphabet $\Sigma$.
}{
  Given $i\in \iv{0}{\weight{T}}$, return $T[j]\in \Sigma$ for $j = \max\{j'\in \iv{0}{|T|}\;:\; \weight{T\iv{0}{j'}}\leq i\}$.
}

Observe that if $T$ is produced by an (RL)SLG $\G$, then a \RA query with argument $i$ reduces to retrieving an appropriate leaf of the parse tree \(\Parse_\G\).
The ancestors of this leaf are precisely nodes \((A,a)\) such that \(i \in \iv{a}{a+\weight{A}}\).
These nodes form a root-to-leaf path in the parse tree $\Parse_\G$.
The natural strategy is to traverse this path, and a single step of such a traversal can be modeled using the following $\child$ query:

\defdsproblem{\(\child_A\) query\label{query:child}}{
  A variable \(A\in \Var_\G\) in a weighted (RL)SLG \(\G\).
}{
  Given a parse tree node \((A,a)\) with symbol \(A\) and an argument \(i\in \iv{a}{a+\weight{A}}\), return the child \((B,b)\) of $(A,a)$ such that $i \in \iv{b}{b+\weight{B}}$.
}

Using $\child$ queries, \RA queries can be solved as follows.

\begin{algorithm}[htbp]
\KwIn{\(i\in\iv{0}{\weight{T}}\)}
\KwOut{The answer to the \RA query in $T$ with argument $i$.}

  \((A,a) \gets (S,0)\)\;
  \While{\(A\notin \Sigma\)}{
    \((A,a)\gets \child_A((A,a),i)\)\;
  }
  \Return \(A\)
  \caption{The \RA Algorithm}\label{alg:random}
\end{algorithm}

Naively, the \(\child_A\) query can be answered in time proportional to $|\rhs(A)|$,
so the running time of \cref{alg:random} can be bounded by the product of the height of $\G$ and the maximum production length of $\G$.
If the input SLG $\G$ is transformed using \cref{theorem:weighted-contracting-slg}, we can achieve random access time $\Oh(\log\weight{T})$, or even $\Oh(\log\frac{\weight{T}}{\weight{c}})$ if the resulting character is $c$.
A relatively simple extension allows implementing \(\child_A\) queries in $\Oh(|\rle(\rhs(A))|)$ time, so we can achieve the same bounds for RLSLGs using \cref{cor:contracting_rlslg}.

The structure of nice RLSLGs introduced in \cref{sec:structured} allows decreasing the grammar height at the expense of an increased production length.
As we show next, this does not interfere with constant-time implementation of \(\child_A\) queries.

\begin{lemma}\label{lemma:nicechild}
  Let \(A\) be a \(\tau\)-nice variable for some real number $\tau>1$.
  In the word-RAM model with word size $w\ge \Omega(\log \weight{A})$, we can construct in time \(\Oh(\tau)\) a data structure of \(\Oh(\tau\log \weight{A})\) bits that supports \(\child_A\) queries in constant time.
\end{lemma}
\begin{proof}
  Let \(C_A\coloneqq \rle(\rhs(A))\) with $m\coloneqq |C_A|$ and $(B_j,k_j)\coloneqq C_A[j]$ for $j\in [0\dd m)$.
  Define the sequence \(P_A[0\dd m]\) of weighted boundary-indices of the runs in the expansion of \(A\) by \(P_A[0]\coloneqq 0\) and
  \(P_A[j+1]\coloneqq P_A[j]+k_j\cdot \weight{B_j}\) for \(j\in[0\dd m)\).
  Let \(X\) be the multiset containing \(P_A[j]\) for every \(j\in\fragmentcc{1}{m}\) such that \(P_A[j]<\weight{A}\), represented by this nondecreasing sequence of boundaries.
  Recall that, for any $y\in \mathbb{R}$, we have $\rank_X(y)=|\{x\in X : x < y\}|$, where multiplicities are counted.
  Then for a query \(\child_A((A,a),i)\) let \(j\coloneqq \rank_X(i-a+1)\).
  Since all weights are non-negative integers, the boundary sequence is nondecreasing, and we have \(j=\rank_X(i-a+1)=|\{\ell\in \fragmentcc{1}{m} : P_A[\ell] \le i-a\}|\), meaning that $a+P_A[j] \le i < a+P_A[j+1] = a + P_A[j] + k_j\cdot \weight{B_j}$.
  In particular, \(P_A[j]<P_A[j+1]\), so \(\weight{B_j}>0\).
  The desired child is the copy of \(B_j\) whose weighted span contains~\(i\), namely
  \[
    \child_A((A,a),i)=\Big(B_j,\; a+P_A[j]+\left\lfloor{\tfrac{i-a-P_A[j]}{\weight{B_j}}}\right\rfloor\cdot\weight{B_j}\Big).
  \]
  This is correct because each of the \(k_j\) copies of \(B_j\) contributes weight \(\weight{B_j}\) and the offset selects the unique copy containing \(i\).
  Hence, \(\child_A\) reduces to a \(\rank_X\) query plus access to \(P_A\) and \(C_A\).

  \begin{claim}\label{clm:setrank}
    Consider a multiset $X$ of size $n$ from a universe \(\iv{0}{u}\), given as a nondecreasing sequence. Let \(s\in\fragmentcc{1}{u}\).
    If \(\rank_X({i+s})-\rank_X(i) \le (\log u)^{\Oh(1)}\) holds for every $i\in \fragmentcc{0}{u-s}$, then, in word RAM with words of size $\Omega(\log u)$,
    in time \(\Oh(n+u/s)\) we can construct an \(\Oh((n+u/s)\log u)\)-bit data structure that supports constant-time \(\rank_X\) queries.
  \end{claim}
  \begin{proof}
    Split \(\iv{0}{u}\) into \(\Oh(u/s)\) intervals of size \(s\), with the last interval possibly shorter.
    Each interval contains at most \(t\coloneqq(\log u)^{\Oh(1)}\) elements of \(X\), counting multiplicities.
    For every interval~\(e\), store the distinct elements of its intersection with $X$ in a fusion tree \(F_e\) \cite{Fredman1993, PT14}, together with an array \(\mathsf{mult}_e\), where \(\mathsf{mult}_e[r]\) is the total multiplicity of the first \(r\) stored distinct elements and \(\mathsf{mult}_e[0]=0\), and keep these structures in an array \(E\);
    since \(t=w^{\Oh(1)}\) when \(w\ge\Omega(\log u)\), these support constant-time rank over the multiset and can be constructed in linear time.
    Let
    \(\mathsf{pref}[e]\) be the total multiplicity of all elements in earlier intervals.
    All arrays are computable in \(\Oh(n+u/s)\) time and use \(\Oh((n+u/s)\log u)\) bits in total.
    Given \(y\in \fragmentcc{0}{u}\), return \(n\) if \(y=u\); otherwise, compute \(e\coloneqq \floor{y/s}\), use \(F_e\) to count the number \(r\) of distinct stored values below \(y\), and return \(\rank_X(y)\coloneqq \mathsf{pref}[e]+\mathsf{mult}_e[r]\).
    All steps are constant time, proving the claim.
  \end{proof}

  If \(\tau \ge \weight{A}\), we can answer \(\rank_X\) by a direct table: store the values
  \(\rank_X(y)\) for all \(y\in \fragmentcc{0}{\weight{A}}\).
  This uses \(\Oh(\weight{A}\log \weight{A})=\Oh(\tau\log \weight{A})\) bits and \(\Oh(\weight{A})=\Oh(\tau)\) time,
  so it already meets the required bounds. We therefore assume \(\tau < \weight{A}\) below.

  By \(\tau\)-niceness (\cref{def:nice}), \(|C_A|\le 2d\tau\) and every substring of \(\rhs(A)\) of weight at most \(\frac1\tau\weight{A}\) contains at most \(2d\ceil{\log \tau}\) runs.
  Let \(n\coloneqq |X|\le |C_A|\), $u\coloneqq \weight{A}$, and \(s\coloneqq \floor{\frac1\tau\weight{A}}\) so that $n+u/s \leq \Oh(\tau)$.
  For every \(i\in\fragmentcc{0}{u-s}\), the value \(\rank_X(i+s)-\rank_X(i)\) is the number of indices \(j\in\fragmentcc{1}{m}\) such that \(P_A[j]\in\iv{i}{i+s}\).
  If there are \(q\ge 2\) such indices \(j_1<\cdots<j_q\), then the substring of \(\rhs(A)\) from run \(j_1\) to run \(j_q-1\) has weight \(P_A[j_q]-P_A[j_1]\le s\le\weight{A}/\tau\), and hence \(q-1\le 2d\ceil{\log\tau}\) by \(\tau\)-niceness.
  Thus, \(\rank_X(i+s)-\rank_X(i)\le 2d\ceil{\log\tau}+1=(\log u)^{\Oh(1)}\).
  Consequently, \(X\) satisfies the premise of \cref{clm:setrank}, so we can answer \(\rank_X\) (and hence \(\child_A\)) in constant time with \(\Oh(\tau\log \weight{A})\) bits and \(\Oh(\tau)\)-time preprocessing.
\end{proof}

Combining \cref{lem:nice-height,lemma:nicechild,alg:random}, we get the following result:

\begin{lemma}\label{lem:ra_simple}
  Let \(\G\) be a weighted RLSLG that produces a string \(T\) of weight~$\weight{T} \le 2^{\Oh(w)}$, where $w$ is the machine word size.
  Given $\G$ and a real parameter $\tau > 1$, in \(\Oh(|\G|\tau)\) time
  we can construct an \(\Oh(|\G|\tau\log \weight{T})\)-bit data structure that answers \RA queries in time
  \[\Oh\Big(\max\Big(1,\;\log_\tau \tfrac{\weight{T}}{|\G|\tau\weight{c}}\Big)\Big),\]
  where $c$ is the returned character.
\end{lemma}
\begin{proof}
  Apply \cref{lemma:niceconstruct} to construct a \((|\G|\tau,\tau)\)-nice (RL)SLG $\cH$ of size $\Oh(|\G|\tau)$ that produces~$T$; this step takes $\Oh(|\G|\tau)$ time.
  Recall that \(\cH\) has \(|\Var_\cH|\in\Oh(|\G|)\) many variables.
  For each variable $A\in \Var_\cH$, construct the data structure of \cref{lemma:nicechild} supporting constant-time $\child_A$ queries.
  The total size of these data structures is $\Oh(|\G|\tau \log \weight{T})$ and their total construction time is $\Oh(|\G|\tau)$.

  By \cref{lem:nice-height}, the path from the root of $\Parse_\cH$ to any leaf with value $c$ does not exceed $2+\max\Big(0, \log_\tau \frac{\weight{T}}{|\G|\tau\weight{c}}\Big)$.
  Consequently, \cref{alg:random} has the desired time complexity.
\end{proof}

For strings over small unweighted alphabets, we can make use of leafy grammars.
\begin{lemma}\label{lem:ra_unweighted}
  Let \(\G\) be an RLSLG that produces a string \(T\) over $\iv{0}{\sigma}$ of length~$n \le 2^{\Oh(w)}$, where $w$ is the machine word size.
  Given $\G$ and a real parameter $\tau > 1$, in \(\Oh(|\G|\tau)\) time
  we can construct an \(\Oh(|\G|\tau\log n)\)-bit data structure that answers \RA queries in time
  \[\Oh\Big(\max\Big(1,\;\log_\tau \tfrac{n\log\sigma}{|\G|\tau \log n}\Big)\Big).\]
\end{lemma}
\begin{proof}
  Denote \(b \coloneqq \ceil{\log n/\log \sigma}\).
  Apply \cref{corr:constructleafyandnice} to construct a \(b\)-leafy (RL)SLG \(\cH\) that produces $T$ and whose top part is \((|\G|\tau,\tau)\)-nice.
  This step takes $\Oh(|\G|\tau + |\G|\tau + |\G|b\log\sigma / w)=\Oh(|\G|\tau)$ time, the top part has size $|\Top{\cH}|\leq \Oh(|\G|\tau)$, and the number of leaf variables is $\Oh(|\G|)$.

  For each top variable $A\in \Vartop{\cH}$, construct the $\child_A$ data structure of \cref{lemma:nicechild}.
  This takes $\Oh(|\Top{\cH}|)=\Oh(|\G|\tau)$ time and $\Oh(|\G|\tau \log n)$ space.
  For each leaf variable $A\in \Varleaf{\cH}$, \cref{corr:constructleafyandnice} already constructs an $\Oh(b \log \sigma)$-bit representation of $\rhs_\cH(A)=\expand[\cH]{A}$.
  This suffices for constant-time $\child_A$ queries because $\child_A((A,a),i)=(c,i)$ for $c=\rhs_\cH(A)[i-a]$.

  By \cref{lem:leafy-nice-height}, the parse tree $\Parse_\cH$ has height at most $3+\max(0,\log_\tau \tfrac{n}{|\G|\tau b})$, so \cref{alg:random} has the desired time complexity.
\end{proof}

A straightforward combination of \cref{lem:ra_simple,lem:ra_unweighted} lets us get the best of both worlds.

\begin{theorem}\label{thm:ra}
  Let \(\G\) be a weighted RLSLG that produces a string \(T\) over $\iv{0}{\sigma}$ of weight~$\weight{T} \le 2^{\Oh(w)}$, where $w$ is the machine word size.
  Given $\G$ and a real parameter $\tau > 1$, in \(\Oh(|\G|\tau)\) time
  we can construct an \(\Oh(|\G|\tau\log \weight{T})\)-bit data structure that answers \RA queries in time
  \[\Oh\Big(\max\Big(1,\;\log_\tau \tfrac{\weight{T}\log\sigma}{|\G|\tau \log \weight{T}}\Big)\Big).\]
  The query time can also be bounded as follows based on the weight of the returned character $c$:
  \[\Oh\Big(\max\Big(1,\;\log_\tau \tfrac{\weight{T}}{|\G|\tau\weight{c}}\Big)\Big).\]
\end{theorem}
\begin{proof}
  Observe that random access to a weighted string $T$ is equivalent to random access to an unweighted string $\hat{T}$ in which every occurrence of a character $c$ is replaced by $c^{\weight{c}}$.
  In other words, every character $c$ that occurs anywhere in the productions of $\G$ can be replaced with a new variable $C$ such that $\rhs(C)\coloneqq c^{\weight{c}}$.
  Consequently, in $\Oh(|\G|)$ time we can build an RLSLG of size $\Oh(|\G|)$ that produces $\hat{T}$.

  By \cref{lem:ra_unweighted}, random access to $\hat{T}$ can be supported in
  $\Oh(\max(1,\log_\tau \tfrac{\weight{T}\log \sigma}{|\G|\tau \log\weight{T}}))$ time using an $\Oh(|\G|\tau \log\weight{T})$-bit data structure that can be constructed in $\Oh(|\G|\tau)$ time.
  The data structure of \cref{lem:ra_simple} has the same size and construction time bounds, so we can build both and run the two query algorithms simultaneously.
  The answer can be returned as soon as one of these two query algorithms terminates, and this happens within the specified time bounds.
\end{proof}

We now derive \cref{thm:upper}; in fact, the same derivation gives the analogous statement for RLSLGs.
Let \(\G\) be an (RL)SLG of size \(g\) generating an unweighted string \(T\in[0\dd\sigma)^n\).
Since \(w\ge\Omega(\log(n\sigma))\) implies \(n\le 2^{\Oh(w)}\), \cref{thm:ra} applies with \(\weight{T}=n\).
Set \(\tau\coloneqq Mw/(g\log n)\); the assumption \(g\log n<Mw\) gives \(\tau>1\).
Applying \cref{thm:ra} with this value of \(\tau\), we obtain construction time \(\Oh(g\tau)=\Oh(Mw/\log n)\) and space \(\Oh(g\tau\log n)=\Oh(Mw)\) bits, which is \(\Oh(M)\) machine words.
Finally, the query time becomes
\[
  \Oh\!\left(\max\!\left(1,\log_\tau\frac{n\log\sigma}{g\tau\log n}\right)\right)
  =
  \Oh\!\left(\max\!\left(1,\frac{\log \frac{n\log\sigma}{Mw}}{\log\frac{Mw}{g\log n}}\right)\right).
\]
This is the bound of \cref{thm:upper}, with the convention that time bounds are at least constant.

%% file: upper/traversal.tex
\newcommand{\rt}{\mathsf{root}}
\newcommand{\mdl}{\mathsf{middle}}
\newcommand{\lft}{\mathsf{left}}
\newcommand{\rgt}{\mathsf{right}}
\newcommand{\woff}{\mathsf{offset}}

\section{Traversing the String}\label{sec:traversal}

Given an (RL)SLG \(\G\) producing a string \(T\), we want a data structure that maintains a pointer to a position in \(T\) and supports forward/backward traversal while returning the traversed substring.
We allow each step to advance by \(b\) consecutive leaves (characters), where \(b\) is a parameter; later we set \(b=\Theta(\min(w,\tau\log n)/\log\sigma)\) so that each returned block fits in \(\Oh(1)\) machine words.
Our goal is constant time per step and pointer construction time matching the random-access trade-offs.

For \(b=1\), there is a known solution to this problem. Lohrey, Maneth and Reh show how to traverse the parse tree leaf by leaf, reaching the next one in constant time each \cite{Lohrey2017}.
We extend their result to weighted strings and RLSLGs.
In the realm of small (unweighted) alphabets we improve on the running time of the traversal algorithm.
The algorithm of \cite{GKPS05} supports traversal in an even more general setting, but at the cost of quadratic preprocessing time, whereas ours is linear. Furthermore, we allow a pointer to move both forward and backward, which generalizes their forward-only result.

It is considerably slower to retrieve some substring from the parse tree than the time it takes to return it, by up to a factor of \(w/\log\sigma\). We improve this gap.

To that end we define character pointers, which support this traversal operation, both forward and backward. In this section we provide an efficient implementation of this interface and show how to apply it to leafy grammars.

\begin{definition}\label{def:leaf_pointer}
  A \emph{character pointer} \(p\) to an unweighted string \(T\) is an object that
  always points at some position \(p.\pos\in\iv{0}{|T|}\), provides constant-time access to $p.\pos$ and $T[p.\pos]$, and supports the following traversal operations:
  \begin{itemize}
    \item \(p.\forward()\) that returns a new character pointer \(p'\) such that \(p'.\pos \coloneqq (p.\pos+1) \bmod |T|\),
    \item \(p.\backward()\) that returns a new character pointer \(p'\) such that \(p'.\pos \coloneqq (p.\pos-1) \bmod |T|\).
  \end{itemize}
  A character pointer can be obtained using the following constructor:
  \begin{itemize}
    \item\(\new(i)\) that, given $i\in [0\dd |T|)$, constructs a character pointer \(p\) such that \(p.\pos = i\).
  \end{itemize}
  If $T$ is weighted, the character pointer provides constant-time access to \(p.\woff \coloneqq \weight{T[0\dd p.\pos)}\) instead of \(p.\pos\).
  Moreover, we adapt the constructor analogously to \RA queries:
  \begin{itemize}
    \item\(\new(i)\), given $i\in [0\dd \weight{T})$, constructs a character pointer \(p\) such that
    \[
      p.\pos = \max\{j\in [0\dd |T|) : \weight{T[0\dd j)} \le i\}.
    \]
  \end{itemize}
\end{definition}
\begin{remark}
  The symmetric operations \(\backward\) and \(\forward\) are called \emph{left} and \emph{right} in~\cite{Lohrey2017}.
\end{remark}

\subsection{Implementing Character Pointers}
In this section, we implement character pointers to grammar-compressed strings with constant-time traversal; the trade-off between construction time and the size of the underlying data structure is the same as for \RA in \cref{thm:ra}.

Our solution generalizes~\cite{Lohrey2017} in that it supports run-length encoding and efficient construction for an arbitrary position.
Our implementation relies on the parse tree $\Parse_\G$, whose leaves, read from left to right, represent the subsequent characters of $T$.
Internally, our pointers allow traversing the entire $\Parse_\G$, not just its leaves.
We formalize this using the following interface:

\begin{definition}
  Let \(\G\) be a weighted (RL)SLG that produces a string \(T\).
  A \emph{node pointer} \(p\) on the parse tree \(\Parse_\G\) of \(\G\) is a data structure that
  always points at some node \(p.\node\) in \(\Parse_\G\), provides constant-time access to $p.\node=(p.\symb,p.\woff)$, and supports the following operations:
  \begin{itemize}
    \item $\mathsf{root}(\G)$ constructs a node pointer to the root of $\Parse_\G$;
    \item \(\pushLeft{p}\) and \(\pushRight{p}\) construct a node pointer \(p'\) to the leftmost (or rightmost, respectively) leaf in the subtree of $\Parse_\G$ rooted at $p.\node$;
    \item \(\popLeft{p}\) and \(\popRight{p}\) construct a node pointer \(p'\) to the lowest ancestor of \(p.\node\) that is either the root of $\Parse_\G$ or has a left (or right, respectively) sibling;
    \item \(\pop{p}\) constructs a node pointer \(p'\) to the parent of \(p.\node\). This operation can only be called if \(p.\node\) is not the root of $\Parse_\G$.
    \item \(\pushChild{p}{i}\) constructs a node pointer \(p'\) to the child $p'.\node=\child_{p.\symb}(p.\node,i)$ of $p.\node$, so that $i\in [p'.\woff\dd \allowbreak p'.\woff+\weight{p'.\symb})$; see \cref{query:child}. This operation can only be called if $p.\node$ is not a leaf and $i\in [p.\woff\dd p.\woff+\weight{p.\symb})$.
  \end{itemize}
\end{definition}

\newcommand{\tn}{t_{\mathsf{n}}}
\newcommand{\tc}{t_{\mathsf{c}}}

We now describe an efficient implementation of character pointers using this interface.
\begin{lemma}\label{lemma:char_pointers_in_tn}
  Consider a weighted (RL)SLG $\G$ that produces a string $T$ of weight $\weight{T}\le 2^{\Oh(w)}$, where $w$ is the machine word size.
  If a data structure supports node pointers to $\G$ in $\tn$ time per operation,
  then it also supports character pointers to $T$ with $\Oh(\tn)$-time traversal and $\Oh(\tn \cdot (1+h_{p.\pos}))$-time construction, where $h_{p.\pos}$ is the depth of the leaf representing $T[p.\pos]$ in $\Parse_\G$.
\end{lemma}
\begin{proof}
  The implementation presented here closely follows the one first given in~\cite{Lohrey2017}.
  Most notably, our \(\pushRight{\cdot}\) operation is equivalent to their \emph{expand-right}. We extend the result to runs, larger right-hand sides, and separate the implementation of the node pointer operations (especially \(\pushChild{\cdot}{\cdot}\)) from the traversal operations.

  To construct a character pointer $\new(i)$, we initialize $p\coloneqq \mathsf{root}(\G)$ and repeat $p\coloneqq \pushChild{p}{i}$ until $p.\symb\in \Sigma_\G$.
  After $h_{p.\pos}$ iterations, \(p\) points at the desired leaf of $\Parse_\G$.

  We will now show how to implement the traversal operations.
  Since \(\forward\) and \(\backward\) are symmetric operations, we discuss $\forward$ only;
  $\backward$ admits an analogous implementation.

  Let us define a helper operation \(\rightSibling{p}\): Given a pointer to a node with a right sibling, this operation constructs a pointer to this sibling.
  To implement cyclic traversal, we assume that the root is itself its right sibling.
  Observe that \(\pushLeft{\rightSibling{\popRight{p}}}\) is well-defined and executes the \(p.\forward()\) operation assuming that $p.\node$ is a leaf of $\Parse_\G$.

  It remains to show that we can construct \(p'\coloneqq \rightSibling{p}\) using constantly many node pointer operations.
  If $p.\node$ is the root of $\Parse_\G$, then we can return $\rightSibling{p} = p$ by definition.
  Otherwise, we have \(\rightSibling{p} = \pushChild{\pop{p}}{p.\woff + \weight{p.\symb}}\).
\end{proof}

We next discuss an efficient implementation of node pointers.

\begin{lemma}\label{lemma:node_pointers_in_tc}
  Consider a weighted (RL)SLG $\G$ that produces a string $T$ of weight $\weight{T}\le 2^{\Oh(w)}$, where $w$ is the machine word size.
  Suppose that, for every $A\in \Var_\G$, we have $|\rhs(A)|\ge 2$ and $\child_A$ queries can be answered in $\tc$ time.
  Then, in $\Oh(|\G|)$ time, we can construct a $\Oh(|\G|\log |\G|)$-bit data structure that supports node pointer operations to $\G$ in $\Oh(\tc)$ time.
\end{lemma}
\begin{proof}
  \newcommand{\lvl}{\mathsf{lvl}}

  Our implementation resembles that of \cite{Lohrey2017}.
  Let \(\G_L\) be the \emph{leftmost-child} forest with node set $\Symb_\G$, root set $\Sigma_\G$, and the parent of each $A\in \Var_\G$ defined as the leftmost symbol of \(\rhs(A)\).
  Symmetrically, the \emph{rightmost-child} forest has node set $\Symb_\G$, root set $\Sigma_\G$, and the parent of each $A\in \Var_\G$ defined as the rightmost symbol of \(\rhs(A)\).
  Observe that $\G_L$ and $\G_R$ are indeed forests of size $\Oh(|\Symb_\G|)\le\Oh(|\G|)$.

  \newcommand{\ancestor}{\mathsf{ancestor}}

  The level $F.\lvl(v)$ of a node $v$ in a forest $F$ is its distance to the root, i.e., the number of edges on the path from $v$ to the root of the same connected component.
  A level-ancestor query \(F.\ancestor(v,\ell)\), given a node $v$ and a level \(\ell\in [0\dd F.\lvl(v)]\), returns the unique ancestor of $v$ at level \(\ell\).
  Standard level-ancestor data structures \cite{BV94,Bender2004} support constant-time level ancestor queries using $\Oh(|F|\log |F|)$ bits and $\Oh(|F|)$ preprocessing time.
  We build such a data structure for both $\G_L$ and $\G_R$, and we also store $\G_L.\lvl(A)$ and $\G_R.\lvl(A)$ for every $A\in \Symb_\G$; this takes $\Oh(|\G|\log |\G|)$ bits and $\Oh(|\G|)$ preprocessing time.

  Recall that a parse tree (\cref{def:parsetree}) of an (RL)SLG \(\G\) is defined over nodes \((A,a)\in\Symb_\G\times\iv{0}{n}\). It has a unique root (the start symbol of \(\G\)) and every node of a variable has children corresponding to its right-hand sides.
  We classify those nodes into four types.
  The root is of type \(\rt\). Any other node has a parent in $\Parse_\G$ and depending on its position among the siblings, it is classified as a leftmost child ($\lft$), rightmost child ($\rgt$), or a middle child ($\mdl$, otherwise).

  We will now define \emph{node pointers}, each of which points to some node of this tree.
  Our definition ensures that for each node there is a unique canonical pointer for it.
  Our algorithm accesses the parse tree via this pointer interface.
  For every node \(\nu\), we define its \emph{node pointer}
  inductively, depending on the type of $\nu$:
  \begin{itemize}
    \item If $\nu$ is the root, then its pointer is \(p_\nu \coloneqq \rt(\nu)\).
    \item If $\nu$ is a middle child, then its pointer is \(p_\nu \coloneqq \mdl(\nu,p')\), where $p'$ points to the parent of $\nu$.
    \item If $\nu$ is a leftmost child, then its pointer is \(p_\nu\coloneqq\lft(\nu,p')\), where $p'$ points to the nearest ancestor of $\nu$ that is not a leftmost child.
    \item If $\nu$ is a rightmost child, then its pointer is \(p_\nu\coloneqq\rgt(\nu,p')\), where $p'$ points to the nearest ancestor of $\nu$ that is not a rightmost child.
  \end{itemize}

  Constructing a pointer to the root $(S_\G,0)$ naively works in constant time.

  To implement the remaining operations, consider a pointer $p$ to $p.\node = \nu = (A,a)$.
  For \(\pushChild{p}{i}\), we first retrieve the parse tree node \((B,b) \coloneqq \child_A(\nu,i)\) for the resulting pointer.
  \begin{itemize}
    \item  If \(a=b\), then \((B,b)\) is the leftmost child of \((A,a)\).
    \begin{itemize}
      \item If \(p=\lft(\nu,p')\),  then we return \(\lft((B,b),p')\);
      \item otherwise, we return \(\lft((B,b),p)\).
    \end{itemize}
    \item If \(b+\weight{B} = a+\weight{A}\), then \((B,b)\) is the rightmost child of \((A,a)\).
    \begin{itemize}
      \item If \(p=\rgt(\nu,p')\),  then we return \(\rgt((B,b),p')\);
      \item otherwise, we return \(\rgt((B,b),p)\).
    \end{itemize}
    \item In the remaining case \((B,b)\) is a middle child, and we return \(\mdl((B,b),p)\).
  \end{itemize}
  We implement \(\pop{p}\) depending on the type of \(p\), which cannot be $\rt$.
  \begin{itemize}
    \item If \(p=\mdl(\nu,p')\), then return \(p'\).
    \item If \(p=\lft(\nu,p')\), we have to differentiate between two cases based on $p'.\node=(A',a')$.
    \begin{itemize}
      \item If \(\G_L.\lvl(A') = \G_L.\lvl(A)+1\), then return \(p'\).
      \item Otherwise, let \(A'' \coloneqq \G_L.\ancestor(A',\G_L.\lvl(A)+1)\) and return \(\lft((A'',a),p')\).
    \end{itemize}
    \item If \(p=\rgt(\nu,p')\), we have to differentiate between two cases based on $p'.\node=(A',a')$.
    \begin{itemize}
      \item If \(\G_R.\lvl(A') = \G_R.\lvl(A)+1\), then return \(p'\).
      \item Otherwise, let \(A'' \coloneqq \G_R.\ancestor(A',\G_R.\lvl(A)+1)\) and return \(\rgt((A'',a+\weight{A}-\weight{A''}),p')\).
    \end{itemize}
  \end{itemize}
  We implement \(\popLeft{p}\) depending on the type of \(p\).
  \begin{itemize}
    \item If \(p = \lft(\nu,p')\), then return \(p'\).
    \item Otherwise, return \(p\).
  \end{itemize}
  We implement \(\popRight{p}\) depending on the type of \(p\).
  \begin{itemize}
    \item If \(p = \rgt(\nu,p')\), then return \(p'\).
    \item Otherwise, return \(p\).
  \end{itemize}
  We implement \(\pushLeft{p}\) as follows:
  \begin{itemize}
    \item If \(\G_L.\lvl(A) = 0\), then return \(p\).
    \item Otherwise, let \(B \coloneqq \G_L.\ancestor(A,0)\).
    \begin{itemize}
      \item If \(p = \lft(\nu,p')\), then return \(\lft((B,a),p')\).
      \item Otherwise, return \(\lft((B,a),p)\).
    \end{itemize}
  \end{itemize}
  We implement \(\pushRight{p}\) as follows:
  \begin{itemize}
    \item If \(\G_R.\lvl(A) = 0\), then return \(p\).
    \item Otherwise, let \(B \coloneqq \G_R.\ancestor(A,0)\) and $b\coloneqq a+\weight{A}-\weight{B}$.
    \begin{itemize}
      \item If \(p = \rgt(\nu,p')\), then return \(\rgt((B,b),p')\).
      \item Otherwise, return \(\rgt((B,b),p)\).
    \end{itemize}
  \end{itemize}
All of these operations take $\Oh(1+\tc)$ time.
\end{proof}

As an immediate consequence, we can combine node-pointer support with character pointers.

\begin{corollary}\label{cor:char_ptrs_from_child}
  Consider a weighted (RL)SLG $\G$ that produces a string $T$ of weight $\weight{T}\le 2^{\Oh(w)}$, where $w$ is the machine word size.
  Suppose that $\child_A$ queries can be answered in $\tc$ time for every $A\in \Var_\G$.
  Then, in $\Oh(|\G|)$ time, we can build an $\Oh(|\G|\log |\G|)$-bit data structure that supports character pointers with $\Oh(\tc)$-time traversal and $\Oh(\tc\cdot(1+h_{p.\pos}))$-time construction, where $h_{p.\pos}$ is the length of the path from the root of $\Parse_\G$ to the leaf representing $T[p.\pos]$.
\end{corollary}
\begin{proof}
  If $\G$ contains a production with an empty right-hand side, we remove it by deleting all occurrences of that variable (iterating over variables in topological order). This preserves the generated string, does not increase the grammar size, and can be done in $\Oh(|\G|)$ time; moreover it can only decrease the height. Observe that \(\child\)-queries never return empty symbols, so this preprocessing does not affect the oracle outputs.

  If $\G$ contains a unary production $A\to B$, we dissolve it by replacing all occurrences of $A$ with $B$, iterating over variables in topological order.
  This preserves the generated string, does not increase the grammar size, and can be done in $\Oh(|\G|)$ time; moreover, the lengths of node-to-root paths can only decrease.
  During this process, we compute for every dissolved symbol \(A\) its final representative \(\mu(A)\), obtained by following unary rules until the first symbol that is not dissolved; for all remaining symbols \(A\), we set \(\mu(A)\coloneqq A\).
  We keep the original \(\child\)-query oracle and post-process its outputs: if a query returns a node \((B,b)\), we return \((\mu(B),b)\) instead.
  Since \(\mu(B)\) is retrieved by a table lookup, this adds only constant overhead and preserves the \(\Oh(\tc)\)-time bound.

  We may therefore assume $|\rhs(A)|\ge 2$ for all $A\in\Var_\G$.
  The claim now follows by combining \cref{lemma:char_pointers_in_tn,lemma:node_pointers_in_tc}.
\end{proof}

Combining \cref{cor:char_ptrs_from_child} with the results of \cref{sec:ub,sec:structured}, we get the desired guarantees for character pointers; this result generalizes \cref{lem:ra_simple}.

\begin{corollary}\label{cor:char_ptrs}
  Let \(\G\) be a weighted RLSLG that produces a string \(T\) of weight~$\weight{T} \le 2^{\Oh(w)}$, where $w$ is the machine word size.
  Given $\G$ and a real parameter $\tau > 1$, in \(\Oh(|\G|\tau)\) time
  we can build an \(\Oh(|\G|\tau\log \weight{T})\)-bit data structure that implements character pointers with constant-time traversal and construction time
  \[\Oh\Big(\max\Big(1,\;\log_\tau \tfrac{\weight{T}}{|\G|\tau\weight{T[p.\pos]}}\Big)\Big).\]
\end{corollary}
\begin{proof}
  Apply \cref{lemma:niceconstruct} to construct a \((|\G|\tau,\tau)\)-nice RLSLG \(\cH\) which also produces \(T\).
  For every variable $A\in \Var_\cH$, apply \cref{lemma:nicechild} to construct a data structure supporting constant-time \(\child_A\) queries.
  These components take $\Oh(|\G|\tau \log\weight{T})$ bits in total and are constructed in $\Oh(|\G|\tau)$ time.
  By \cref{lem:nice-height}, the path in $\Parse_\cH$ from the root to the leaf representing \(T[p.\pos]\) is of length at most \(2+\max\left(0,\log_{\tau}\frac{\weight{T}}{|\G|\tau \weight{T[p.\pos]}}\right)\).
  Hence, \cref{cor:char_ptrs_from_child} yields an $\Oh(|\cH|\log|\cH|)\le\Oh(|\G|\tau\log\weight{T})$-bit data structure supporting character pointer operations in the claimed time. The construction time is $\Oh(|\G|\tau+|\cH|)=\Oh(|\G|\tau)$.
\end{proof}

\subsection{Fast Character Pointers}
When applied on top of leafy grammars we can report multiple symbols at a time.
To this end, we define fast character pointers.

\begin{definition}\label{def:leafy_pointers}
  A \emph{fast character pointer} \(p\) to a non-empty string \(T\) is a character pointer to $T$ (see \cref{def:leaf_pointer}) also supporting the following \emph{fast traversal} operations:
  \begin{itemize}
    \item \(\forward(m)\) which, given $m\in \Zp$, retrieves \(T^\infty\iv{p.\pos}{p.\pos+m}\)\footnote{We define $T^\infty$ as an infinite string indexed by $\mathbb{Z}$ so that $T^\infty[i]=T[i\bmod |T|]$ for every $i\in \mathbb{Z}$.} and returns a fast character pointer \(p'\) such that \(p'.\pos \coloneqq (p.\pos+m) \bmod |T|\); and
    \item \(\backward(m)\) which, given $m\in \Zp$, retrieves \(T^\infty\iv{p.\pos-m}{p.\pos}\) and returns a fast character pointer \(p'\) such that \(p'.\pos \coloneqq (p.\pos-m) \bmod |T|\).
  \end{itemize}
\end{definition}

Any character pointer supporting traversal operations in $t_\mathsf{t}$ time can support fast traversal operations in $\Oh(m\cdot t_\mathsf{t})$ time.
We provide a more efficient implementation for unweighted grammar-compressed strings over small alphabets, exploiting bit-parallelism of the word RAM.

\begin{lemma}\label{lem:traverse_leafy}
  Let \(\G\) be an RLSLG that produces a string \(T\in [0\dd \sigma)^n\).
  Suppose that the machine word satisfies $w\ge\Omega(\log n)$, the RLSLG is $b$-leafy for some positive integer \(b \le \Oh(w/\log\sigma)\), and the \(\child_A\) queries on variables \(A\in \Vartop{\G}\) can be answered in time \(\tc\).

  Then, in \(\Oh(|\G|)\) time, we construct an $\Oh(|\G|\log n)$-bit data structure supporting fast character pointers to $T$ with \(\Oh(\tc\cdot(1+ m/b))\)-time traversal and \(\Oh(\tc\cdot (1+h_{p.\pos}))\)-time construction, where $h_{p.\pos}$ is the length of the path from the root of $\Parse_\G$ to the leaf representing $T[p.\pos]$.
\end{lemma}
\begin{proof}
  Consider the character pointers of \cref{cor:char_ptrs_from_child} applied to the weighted string $\Top{T}$ produced by \(\Top{\G}\).
  The \(\child_A\)-queries on symbols of \(\Vartop{\G}\) can be answered in time \(\tc\), and the answers to these queries are the same in $\G$ and $\Top{\G}$.
  Let \(q\) be such a character pointer to \(\Top{T}\). It supports \(\Oh(\tc)\)-time traversal and \(\Oh(\tc\cdot (1+\Top{h_{q.\pos}}))\)-time construction, where \(\Top{h_{q.\pos}}\) is the depth of the corresponding leaf in \(\Parse_{\Top{\G}}\) (note that \(\Top{h_{q.\pos}}\le h_{p.\pos}\)).

  A character pointer $q$ to $\Top{T}$ gives constant-time access to the symbol \(\Top{T}[q.\pos]\in\Varleaf{\G}\); we set \(q.\rhs \coloneqq \rhs(\Top{T}[q.\pos])=\expand[\G]{\Top{T}[q.\pos]}\), which is stored explicitly as a packed string.

  We then build a fast character pointer \(p\) to \(T\) by augmenting \(q\) with an offset \(\leafPos \in \iv{0}{|q.\rhs|}\) indicating the position inside \(q.\rhs\); thus a fast pointer is represented as the pair \(p=(q,\leafPos)\) such that \(p.\pos = q.\woff + \leafPos\), where \(q.\woff\) is the prefix weight available in constant time by the weighted interface.
  Given a query index \(i\in[0\dd |T|)\), we invoke the weighted constructor for character pointers to \(\Top{T}\), which returns a pointer \(q\); we then set \(\leafPos \coloneqq i-q.\woff\).
  This adds only constant time on top of the \(\Oh(\tc\cdot (1+\Top{h_{q.\pos}}))\)-time construction, whose total time is \(\Oh(\tc\cdot (1+h_{p.\pos}))\) because $h_{p.\pos}= 1+\Top{h_{q.\pos}}$.

  We argue that this implementation supports fast character pointers on \(\G\).
  Without loss of generality, we show how to compute \(p.\forward(m)\) by unpacking \(p=(q,\leafPos)\), updating local copies, outputting the traversed characters, and returning a new fast pointer.
  In the algorithm, \(r\) denotes an append-only output buffer rather than a packed string kept in one word.
  We maintain \(|r|\) as the number of characters already written to the buffer; an assignment of the form \(r\gets r\cdot x\) means that a packed string \(x\) is appended to the end of the output buffer in $\Oh(1+|x|\log \sigma/ w)$ time.

  \begin{algorithm}[H]
    \((q,\leafPos)\gets p\)\;
    \If{\(\leafPos+m < |q.\rhs|\)}{
      \(\leafPos\gets \leafPos + m\)\;
      \Return \(((q,\leafPos),q.\rhs\iv{\leafPos-m}{\leafPos})\)
    }\Else{
      \(r\gets q.\rhs\iv{\leafPos}{|q.\rhs|}\)\;
      \(q\gets q.\forward()\)\;
      \While{\(|r|+|q.\rhs|\le m\)}{
        \(r\gets r\cdot q.\rhs\)\;
        \(q\gets q.\forward()\)\;
      }
      \(\leafPos\gets m-|r|\)\;
      \Return \(((q,\leafPos), r\cdot q.\rhs\iv{0}{\leafPos})\)
    }
    \caption{\(p.\forward(m)\)}
  \end{algorithm}

  By definition of \(b\)-leafiness, \(|q.\rhs|\in \iv{b}{2b}\) always holds, so every packed block or slice that we manipulate before writing it to \(r\) has \(\Oh(b\log\sigma)\le\Oh(w)\) bits.
  Thus, slicing a leaf string and appending one such block to the output buffer costs constant time; the buffer itself may contain \(m\) characters, but it is only written sequentially and is never manipulated as a single packed word.
  Each iteration of the while-loop advances \(q\) to the next symbol of \(\Top{T}\) and appends one full leaf string, contributing at least \(b\) characters.
  Hence, the number of iterations is \(\Oh(1+m/b)\), and each iteration costs \(\Oh(\tc)\) due to the underlying traversal on \(\Top{T}\), yielding total time \(\Oh(\tc\cdot(1+m/b))\).

  The correctness follows by cases: if the remaining suffix of \(q.\rhs\) has length strictly larger than \(m\), we return exactly that substring and update \(\leafPos\) accordingly, keeping \(\leafPos<|q.\rhs|\). Otherwise, we output the suffix, advance \(q\) along \(\Top{T}\) collecting full leaf strings while \(|r|+|q.\rhs|\le m\), and finally output the needed prefix of the last leaf; the updated \(\leafPos\) then points to the correct offset in the new \(q.\rhs\).
  In all cases, the output string has length \(m\) and the returned pointer \((q,\leafPos)\) represents the position exactly \(m\) characters after the starting position (modulo \(|T|\)).
\end{proof}

With all those tools prepared, we can now implement fast-character pointers to RLSLGs.

\begin{corollary}\label{cor:traverse}
  Let \(\G\) be an RLSLG producing an unweighted string \(T\in [0\dd \sigma)^n\).
  In the word RAM with word size $w \ge \Omega(\log (n\sigma))$, given \(\G\) and a parameter \(\tau > 1\), in \(\Oh(|\G|\tau)\) time
  we can construct an \(\Oh(|\G|\tau \log n)\)-bit data structure supporting fast character pointers to \(T\) with traversal time
  \[\Oh\!\left(1+\tfrac{m\log\sigma}{\min(w,\tau\log n)}\right)\]
  and construction time
  \[
    \Oh\Big(\max\big(1,\log_\tau \tfrac{n\log\sigma}{|\G|\tau \log n}\big)\Big).
  \]
\end{corollary}
\begin{proof}
  Let \(L\coloneqq \min(w,\tau\log n)\) and \(b\coloneqq\min(n,\ceil{L/\log\sigma})\).
  Construct a \((|\G|\tau,\tau)\)-nice \(b\)-leafy grammar \(\cH\) (\cref{corr:constructleafyandnice}) that produces all strings that \(\G\) produces.
  If \(b=n\), then the leaf representation stores \(T\) explicitly in \(\Oh(1)\) packed words.
  Let \(\ell\coloneqq\max(1,\floor{L/\log\sigma})\).
  We can then implement fast traversal directly on this packed cyclic string by repeatedly writing the next \(\min(\ell,m')\) characters of \(T^\infty\), where \(m'\) is the number of characters still to output.
  Each such chunk has \(\Oh(L)\le\Oh(w)\) bits and can be assembled in constant time from the packed representation of \(T\) using slicing, concatenation, and string powers.
  This takes \(\Oh(1+m/\ell)=\Oh(1+m\log\sigma/L)\) time, and pointer construction is constant-time.

  Hence, in the rest of the proof, we assume \(b<n\), so \(b=\ceil{L/\log\sigma}\).
  By \cref{lem:leafy-nice-height}, the height of \(\cH\) is at most
  \[
    3+\max\left(0,\log_\tau \tfrac{n}{|\G|\tau b}\right)
    = \Oh\Big(\max\big(1,\log_\tau \tfrac{n\log\sigma}{|\G|\tau \log n}\big)\Big),
  \]
  where we use \(b\ge L/\log\sigma \ge\Omega(\log n / \log\sigma)\).
  For each top symbol \(A\in \Vartop{\cH}\) we construct a data structure that supports \(\child_A\)-queries in constant time (\cref{lemma:nicechild}).
  For each leaf symbol \(A\in \Varleaf{\cH}\), we store \(\rhs_{\cH}(A)=\expand[\cH]{A}\) explicitly as a packed string, as required by \cref{lem:traverse_leafy}.
  The total space of these components is \(\Oh(|\cH|\log n)\le\Oh(|\G|\tau\log n)\), and the total construction time remains \(\Oh(|\G|\tau)\), since \(|\cH|\le \Oh(|\G|\tau)\) and the leaf right-hand sides are already produced by \cref{corr:constructleafyandnice}.
  Finally, \(b\le n\) by definition and \(b\log\sigma\le\Oh(w)\), so \cref{lem:traverse_leafy} applies and yields traversal time
  \[
    \Oh(1+m/b)=\Oh(1+m\log\sigma/L).\qedhere
  \]
\end{proof}

This traversal immediately yields efficient substring extraction.

\begin{corollary}\label{cor:substring}
  Let \(\G\) be an RLSLG producing an unweighted string \(T\in [0\dd \sigma)^n\).
  In the word RAM with word size $w \ge \Omega(\log (n\sigma))$, given \(\G\) and a parameter \(\tau > 1\), in \(\Oh(|\G|\tau)\) time
  we can construct an \(\Oh(|\G|\tau \log n)\)-bit data structure that supports extracting any substring \(T\iv{i}{i+m}\) in time
  \[
    \Oh\Big(\max\big(1,\log_\tau \tfrac{n\log\sigma}{|\G|\tau \log n}\big)
      + \tfrac{m\log\sigma}{\min(w,\tau\log n)}\Big).
  \]
\end{corollary}
\begin{proof}
  Use \cref{cor:traverse} to build the fast pointer structure. To answer \(T\iv{i}{i+m}\), construct \(\mathsf{new}(i)\) and apply \(\forward(m)\).
  The time bound is the sum of the construction time and the traversal time from \cref{cor:traverse}.
\end{proof}

To derive \cref{thm:retrieve}, use the same substitution as for \cref{thm:upper}: set \(g\coloneqq|\G|\) and \(\tau\coloneqq Mw/(g\log n)\).
The assumption \(g\log n<Mw\) gives \(\tau>1\), and the construction time and space become \(\Oh(Mw/\log n)\) and \(\Oh(M)\) words, respectively.
The initialization term is the same as in \cref{thm:upper}, and the traversal term satisfies
\[
  \tfrac{m\log\sigma}{\min(w,\tau\log n)}
  =
  \tfrac{m\log\sigma}{w}\cdot \max\!\left(1,\tfrac{g}{M}\right)
  \le
  \tfrac{m\log\sigma}{w}\cdot\tfrac{M+g}{M}.
\]
Together with the convention that time bounds are at least constant, this yields \cref{thm:retrieve}.

%% file: upper/rankandselect.tex
\section{Rank and Select in Grammar-Compressed Strings}\label{sec:rank_select}

In this section, we extend \RA queries of \cref{sec:ub} to \emph{rank} and \emph{select} queries.
For a finite set $X$ from a totally ordered universe $(U,\prec)$, the $\rank_X(u)$ query, given $u\in U$, returns $|\{x \in X : x \prec u\}|$, that is, the number of elements of $X$ strictly smaller than $u$.
The $\select_X(r)$ query, given $r\in [0\dd |X|)$, returns the unique $x\in X$ with $\rank_X(x)=r$ or, in other words, the $r$th smallest element of $X$ (0-based).
In the context of strings $T\in \Sigma^*$, we define $T_c \coloneqq \{i\in [0\dd |T|) : T[i]=c\}$ for every $c\in \Sigma$, and set $\rank_{T,c}\coloneqq \rank_{T_c}$ as well as $\select_{T,c}\coloneqq \select_{T_c}$.

We use the following standard implementations of these queries.

\begin{fact}[{\cite[Theorem 5]{Munro2014}, \cite[Lemma 2.3]{Babenko2014}}]\label{lemma:bitvector_rankselect}
  For every set $X\sub \iv{0}{u}$, there is an $\Oh(u)$-bit data structure supporting $\rank_X$ and $\select_X$ queries in $\Oh(1)$ time.
  Moreover, for any $t\ge 2$, after $\Oh(t)$-time preprocessing (shared across all $X$ and $u$), such a data structure can be constructed in $\Oh(\ceil{u/\log t})$ time if $X$ is given as a bitmask $B\in \{0,1\}^u$ satisfying $B[x]=1$ if and only if $x\in X$.
\end{fact}

\begin{fact}[Elias--Fano Encoding \cite{F71,E74}; see also {\cite[Section 2]{PV17}}]\label{claim:EF}
  For every set $X\sub \iv{0}{u}$, there is an $\Oh(|X|\log \frac{2u}{|X|})$-bit data structure supporting $\select_X$ queries in $\Oh(1)$ time.
  Such a data structure can be constructed in $\Oh(|X|)$ time when $X$ is given as a sorted sequence.
\end{fact}

\subsection{Prefix Sum Queries}
As a convenient stepping stone, we show that our tools are applicable to more general $\PrefixSum$ queries, defined for an arbitrary mapping $\Phi : \Sigma \to \cM$ from the alphabet $\Sigma$ to a monoid $(\cM,\oper,\cZero)$.
For $T\in \Sigma^*$, we set $\Phi(T)=\Phi(T[0])\oper \cdots \oper \Phi(T[|T|-1])=\Oper_{i=0}^{|T|-1}\Phi(T[i])$.
In the context of an SLG $\G$ over $\Sigma$, we also denote $\Phi(A)=\Phi(\expand[\G]{A})$ for every $A\in \Var_\G$, and then lift $\Phi$ to $\Symb_\G^*$ by concatenation.
This generalizes our conventions of weight functions, which map $\Sigma$ to $(\Zz,+,0)$.

For unweighted strings, prefix sum queries ask to compute $\Phi(T[0\dd i))$ given $i\in [0\dd |T|]$.
A standard generalization to weighted strings is formulated as follows:

\defdsproblem{\(\PrefixSum_{T,\Phi}(i)\) Queries}{
  A string $T\in \Sigma^n$, a weight function $\weight{\cdot}$, and a mapping \(\Phi\) from $\Sigma$ to a monoid $(\cM,\oper,\cZero)$.
}{
  Given an index \(i\in\fragmentcc{0}{\weight{T}}\), return \(\Phi(T[0\dd j))\), where \vspace{-.3cm}\[j = \max\{j'\in \fragmentcc{0}{|T|}\;:\; \weight{T\iv{0}{j'}}\leq i\}.\]
}

We henceforth assume that the elements of $\cM$ have an $\rM$-bit representation that supports $\oper$ in constant time.
Whenever we handle RLSLGs, we also assume that, for an integer $k\in \Zz$ and element $x\in \cM$, the $k$-fold sum $k\cdot x \coloneqq x \oper \cdots \oper x$ can also be computed in constant time.
\begin{remark}\label{rmk:monoid}
In the context of $\PrefixSum_{T,\Phi}$ queries, we can restrict $\cM$ to the elements of the form $\Phi(T[i\dd j))$ for $0\le i \le j \le |T|$ and replace all the remaining elements by a single absorbing sentinel element; this transformation does not change the results of any $\PrefixSum_{T,\Phi}$ queries yet often makes the assumption about constant-time operations feasible.
For example, in the case of $(\Zz,+,0)$, we can cap all calculations at $\Phi(T)$ so that the elements can be represented in $\rM \coloneqq \ceil{\log(1+\Phi(T))}$ bits and the operations take constant time in word RAM with $w\ge \rM$.
\end{remark}

\begin{lemma}\label{lemma:prefsum}
  Let \(\G\) be a weighted (RL)SLG that produces a string \(T\in \Sigma^*\) of weight~$\weight{T} \le 2^{\Oh(w)}$, where $w$ is the machine word size.
  Given $\G$, a mapping $\Phi : \Sigma\to \cM$ to a monoid, and a real parameter $\tau > 1$, in \(\Oh(|\G|\tau)\) time we can construct an \(\Oh(|\G|\tau(\rM+\log \weight{T}))\)-bit data structure that answers $\PrefixSum_{T,\Phi}(i)$ queries in \(\Oh(1)\) time for \(i=\weight{T}\), and for \(i\in \fragmentco{0}{\weight{T}}\) in time
  \[\Oh\Big(\max\Big(1,\;\log_\tau \tfrac{\weight{T}}{|\G|\tau\weight{c}}\Big)\Big),\]
  where $c$ is the result of the \RA query for $i$.
\end{lemma}
\begin{proof}
  First, we use \cref{lemma:niceconstruct} to convert $\G$ to a $(|\G|\tau,\tau)$-nice (RL)SLG $\cH$ that produces~$T$.
  Next, we compute and store $\Phi(A)$ for every $A\in \Symb_\cH$.
  For terminals, this is just the given value $\Phi(A)$.
  For variables, we rely on $\Phi(A) = \Oper_{j=0}^{m-1} k_j \cdot \Phi(B_j)$ when $\rle(\rhs(A))=((B_j,k_j))_{j=0}^{m-1}$ and our assumption about $\Oh(1)$-time operations~on~$\cM$.

  In order to adapt the \RA algorithm to \PrefixSum, we also need to implement $\PrefixSum_{\rhs(A),\Phi}$ queries for every $A\in \Var_\cH$.
  For this, we precompute the prefix sums $S_A[0\dd m]$ defined with $S_A[j] = \Phi(B_0^{k_0}\cdots B_{j-1}^{k_{j-1}})$ if $\rle(\rhs(A))=((B_j,k_j))_{j=0}^{m-1}$,
  and we extend the implementation of $\child_A$ queries from \cref{lemma:nicechild}.
  The implementation of these queries already identifies an index $j\in [0\dd m)$ and an exponent $k\in [0\dd k_j]$ such that the relevant prefix of $\rhs(A)$ is $B_0^{k_0}\cdots B_{j-1}^{k_{j-1}}\cdot B_j^k$.
  We can thus return $S_A[j]\oper (k\cdot \Phi(B_j))$ as an answer to the $\PrefixSum_{\rhs(A),\Phi}(i)$.

  Now, $\PrefixSum_{T,\Phi}$ queries can be implemented using \cref{alg:prefsum}, which is a variant of \cref{alg:random} that keeps track of $p=\PrefixSum_{T,\Phi}(a)$ while at node $(A,a)$ of the parse tree~$\Parse_\cH$.
  When we descend from $(A,a)$ to $(B,b)=\child_{A}((A,a),i)$, then $\PrefixSum_{T,\Phi}(b)=\PrefixSum_{T,\Phi}(a)\oper\PrefixSum_{\rhs(A),\Phi}(i-a)$, and the latter term is computed in constant time as described above.
  If \(i=\weight{T}\), then we can simply return the precomputed value \(\Phi(S_\cH)=\Phi(T)\).

  \begin{algorithm}[htbp]
    \KwIn{\(i\in\fragmentcc{0}{\weight{T}}\)}
    \KwOut{The prefix sum in $T$ for argument $i$.}
    \lIf{\(i=\weight{T}\)}{%
      \Return \(\Phi(T)\)%
    }
    \((A,a) \gets (S_\cH,0)\)\;
    \(p\gets \cZero\)\;
    \While{\(A\notin \Sigma\)}{
      \(p\gets p\oper\PrefixSum_{\rhs(A),\Phi}(i-a)\)\;
      \((A,a)\gets \child_A((A,a),i)\)\;
    }
    \Return \(p\)
    \caption{The \(\PrefixSum_{T,\Phi}(i)\) Algorithm}\label{alg:prefsum}
  \end{algorithm}

  In the complexity analysis, we only analyze the extra cost compared to the implementation of \RA queries in the proof of \cref{lem:ra_simple}.
  This is $\Oh(|\cH|)\le\Oh(|\G|\tau)$ preprocessing time and $\Oh(|\cH|\rM)\le\Oh(|\G|\tau\rM)$ bits for the arrays $S_A$ and the values $\Phi(A)$ stored for every $A\in \Symb_\cH$.
  As in \cref{lem:ra_simple}, the query time is proportional to the length of the traversed path of $\Parse_\cH$.
\end{proof}

In order to support \PrefixSum queries in leafy grammars, we only need an efficient implementation for the expansions of leaf variables.

\newcommand{\tleaf}{t_{\mathsf{leaf}}}
\begin{lemma}\label{lemma:prefsum-leafy}
  Consider a $b$-leafy (RL)SLG \(\G\) that produces a string \(T\in \Sigma^*\) of weight~$\weight{T} \le 2^{\Oh(w)}$, where $w$ is the machine word size, and a mapping $\Phi : \Sigma\to \cM$ to a monoid.
  Suppose that, for every $A\in \Varleaf{\G}$, we are given an oracle supporting $\PrefixSum_{\expand[\G]{A},\Phi}$ queries in time $\tleaf$.
  Then, given $\G$, $\Phi$, and a real parameter $\tau > 1$, one can in $\Oh(|\Top{\G}| \tau + |\Varleaf{\G}|\tleaf)$ time construct an \(\Oh(|\Top{\G}|\tau(\rM+\log \weight{T}))\)-bit data structure that answers $\PrefixSum_{T,\Phi}$ queries in time
  \[\tleaf+\Oh\Big(\max\Big(1,\;\log_\tau \tfrac{\weight{T}}{|\Top{\G}|\tau b}\Big)\Big).\]
\end{lemma}
\begin{proof}
  We apply prefix sum queries to compute $\Phi(A)$ for each $A\in \Varleaf{\G}$.
  Our data structure consists of the components of \cref{lem:ra_simple,lemma:prefsum}, applied to the (RL)SLG $\Top{\G}$ generating a weighted string $\Top{T}$.
  They efficiently answer \RA and \PrefixSum queries on \(\Top{T}\) with respect to both $\weight{\cdot}$ and $\Phi$.
  If \(i=\weight{T}\), then we can simply return the precomputed value \(\Phi(T)\).
  Otherwise, applying these queries for \(i\in\iv{0}{\weight{T}}\), we obtain a variable $A\in \Varleaf{\G}$, an index $a\in [0\dd \weight{T})$, and a partial sum $\PrefixSum_{\Top{T},\Phi}(i)$.\footnote{In fact, \cref{alg:prefsum} in the proof of \cref{lemma:prefsum} already computes the other two values.}
  Note that $(A,a)$ is a node in $\Parse_\G$ and $i\in [a\dd a+\weight{A})$, whereas the partial sum satisfies $\PrefixSum_{\Top{T},\Phi}(i)=\PrefixSum_{\Top{T},\Phi}(a)$.
  Consequently, the desired partial sum can be computed as
  \[\PrefixSum_{T,\Phi}(i) = \PrefixSum_{\Top{T},\Phi}(i) \oper \PrefixSum_{\expand[\G]{A},\Phi}(i-a).\]

  The queries on $\Top{T}$ take $\Oh(\max(1,\;\log_\tau\frac{\weight{T}}{|\Top{\G}|\tau \weight{A}})) \le \Oh(\max(1,\;\log_\tau\frac{\weight{T}}{|\Top{\G}|\tau b}))$ time because $\weight{A}\ge \explen{A}\ge b$, whereas the final $\PrefixSum_{\expand[\G]{A},\Phi}(i-a)$ takes $\tleaf$ time.
  The construction time and data structure size are directly inherited from \cref{lem:ra_simple,lemma:prefsum}, except that constructing $\Phi(A)$ for all $A\in \Varleaf{\G}$ takes $\Oh(|\Varleaf{\G}|\tleaf)$ extra time.
\end{proof}

\subsection{Rank Queries}

In order to answer $\rank_{T,1}$ queries in binary strings, we fix a mapping \(\Phi:\{0,1\}\to (\Zz,+,0)\) defined by \(\Phi(1)=1\) and \(\Phi(0)=0\).
Then \(\PrefixSum_{T,\Phi}(i)=\rank_{T,1}(i)\) for all \(i\in\fragmentcc{0}{n}\).

\begin{lemma}\label{lem:rank}
  Let \(\G\) be an (RL)SLG producing an unweighted binary string \(T\in\{0,1\}^n\) of length \(n\le 2^{\Oh(w)}\).
  Given \(\G\) and a real parameter \(\tau > 1\), in \(\Oh(|\G|\tau)\) time, we can construct an \(\Oh(|\G|\tau\log n)\)-bit data structure that supports \(\rank_{T,1}\) queries in \(\Oh\left(\max\left(1,\log_\tau\frac{n}{|\G|\tau}\right)\right)\) time.
\end{lemma}
\begin{proof}
  Applying \cref{lemma:prefsum} to \(\G\) with the unit weight function and the mapping \(\Phi:\{0,1\}\to (\Zz,+,0)\) defined above yields, in \(\Oh(|\G|\tau)\) time, a data structure of bit size
  \[
    \Oh(|\G|\tau(\rM+\log n))
  \]
  that answers \(\PrefixSum_{T,\Phi}\), and hence \(\rank_{T,1}\), in time \(\Oh\left(\max\left(1,\log_\tau\frac{n}{|\G|\tau}\right)\right)\).
  Since \(\Phi(T)\le n\), by \cref{rmk:monoid} we have \(\rM=\Oh(\log n)\), so the space bound becomes \(\Oh(|\G|\tau\log n)\).
\end{proof}

\begin{theorem}\label{thm:rank}
  Let \(\G\) be an (RL)SLG producing an unweighted binary string \(T\in\{0,1\}^n\) of length \(n\le 2^{\Oh(w)}\).
  Given \(\G\) and a real parameter \(\tau>1\), in \(\Oh(|\G|\tau)\) time, we can construct an \(\Oh(|\G|\tau\log n)\)-bit data structure that supports \(\rank_{T,1}\) queries in time \(\Oh\Big(\max\Big(1,\;\log_\tau \tfrac{n}{|\G|\tau \log n}\Big)\Big)\).
\end{theorem}
\begin{proof}
  Set \(b\coloneqq \max(1,\ceil{\min(\log n, 2\log(|\G|\tau))})\) and use \cref{lemma:leafyconstruct} to obtain, in \(\Oh(|\G|(1+b/w))\le \Oh(|\G|)\) time, a \(b\)-leafy RLSLG \(\cH\) producing \(T\) with \(|\Top{\cH}|,|\Varleaf{\cH}|\le\Oh(|\G|)\).
  Adding unused top variables of constant-size right-hand sides if necessary, we assume that \(|\Top{\cH}|=\Theta(|\G|)\) without changing the produced string.
  By \cref{lemma:prefsum-leafy}, it remains to support \(\PrefixSum_{\expand[\cH]{A},\Phi}\) queries for every leaf variable \(A\in \Varleaf{\cH}\).
  Since each \(\expand[\cH]{A}\) is a bit vector of length \(\Theta(b)\), \cref{lemma:bitvector_rankselect} yields constant-time \(\rank_{\expand[\cH]{A},1}\), and hence \(\PrefixSum_{\expand[\cH]{A},\Phi}\), queries for all leaves in \(\Oh(|\G|b)\le\Oh(|\G|\log n)\) bits.
  Choosing \(t\coloneqq |\G|\tau\) in \cref{lemma:bitvector_rankselect}, after \(\Oh(|\G|\tau)\)-time shared preprocessing each leaf structure can be constructed in \(\Oh(\lceil b / \log (|\G|\tau)\rceil)=\Oh(1)\) time, for a total of \(\Oh(|\G|)\le \Oh(|\G|\tau)\).
  Therefore, we can apply \cref{lemma:prefsum-leafy} with \(\tleaf\le\Oh(1)\), which adds \(\Oh(|\Top{\cH}|\tau+|\Varleaf{\cH}|\tleaf)=\Oh(|\G|\tau)\) construction time and \(\Oh(|\Top{\cH}|\tau(\rM+\log n))=\Oh(|\G|\tau\log n)\) bits of space.

  It remains to simplify the query time.
  If \(|\G|\tau < \sqrt n\), then \(\log_\tau \frac{n}{|\G|\tau b}\le \log_\tau n\le \Oh(\log_\tau \frac{\sqrt n}{\log n})\le \Oh(\log_\tau \frac{n}{|\G|\tau\log n})\).
  Otherwise, \(|\G|\tau\ge \sqrt n\), and the choice of \(b\) gives \(b\ge\log n\), so \(\log_\tau \frac{n}{|\G|\tau b}\le \log_\tau \frac{n}{|\G|\tau\log n}\).
  In both cases, we conclude that \(\max(1, \log_\tau \frac{n}{|\G|\tau b}) \le \Oh(\max(1,\log_\tau \frac{n}{|\G|\tau\log n}))\).
\end{proof}

\paragraph{Arbitrary Alphabets.}
For every character \(a\in \Sigma\), let \(T^{(a)}\in\{0,1\}^n\) be obtained from \(T\) by mapping \(a\) to \(1\) and every other character to \(0\).
Applying the same morphism to each terminal of an (RL)SLG \(\G\) producing \(T\) yields an (RL)SLG \(\G^{(a)}\) of the same size producing \(T^{(a)}\).
Since \(\rank_{T,a}(i)=\rank_{T^{(a)},1}(i)\), we can build the binary data structure of \cref{thm:rank} independently for all \(a\in \Sigma\).
This yields an \(\Oh(|\G|\tau\sigma\log n)\)-bit data structure supporting all \(\rank_{T,a}\) queries in time \(\Oh\!\left(\max\!\left(1,\log_\tau \frac{n}{|\G|\tau\log n}\right)\right)\), with total construction time \(\Oh(|\G|\tau\sigma)\).

\subsection{Select Queries}

Similarly to \cite{BCJT15}, in order to implement $\select_{T,1}$ queries, we extract all $1$s in $T$ and associate each of them with the number of preceding zeros.
In the spirit of grammar transformations, we build a grammar deriving such a difference representation of the set $T_1=\{i\in [0\dd |T|) : T[i]=1\}$.

\newcommand{\df}{\mathsf{df}}
Formally, we define a mapping $\df: \{0,1\}^* \to Z^+$, where $Z\coloneqq\{z_i : i\in \Zz\}$,\footnote{We use $Z$ instead of directly $\Zz$ to avoid confusion between concatenation and integer multiplication.} with $\df(0^k)=z_k$ and $\df(0^k1R)=z_k\df(R)$ for $k\in \Zz$ and $R\in \{0,1\}^*$. Thus, \mbox{$\df(0^{k_0}10^{k_1}\cdots 10^{k_m})=z_{k_0}z_{k_1}\cdots z_{k_m}$.}
\begin{observation}\label{obs:Phi}
Consider a mapping $\Phi$ from $Z$ to a monoid $(\Zz,+,0)$ defined by $\Phi(z_i)=i$.
Then, $\select_{T,1}(r) = \Phi(\df(T)\fragmentcc{0}{r})+r=\PrefixSum_{\df(T),\Phi}(r+1)+r$ holds for every $r\in [0\dd |T_1|)$.
\end{observation}
\begin{proof}
  Write \(\df(T)=z_{k_0}z_{k_1}\cdots z_{k_m}\) as defined above so that \(k_j\) is the number of zeros immediately preceding the \(j\)-th one for \(j\in[0\dd m)\).
  The position of the \(r\)-th one is therefore the number of preceding zeros plus the number of preceding ones, that is, \(\sum_{j=0}^{r} k_j+r=\Phi(\df(T)\fragmentcc{0}{r})+r=\PrefixSum_{\df(T),\Phi}(r+1)+r\).
\end{proof}

\begin{lemma}\label{claim:selectgrammar}
  Given an (RL)SLG $\G$ representing a string $T\in \{0,1\}^*$, one can in $\Oh(|\G|)$ time construct an (RL)SLG $\G'$ of size $\Oh(|\G|)$ representing the string $\df(T)$.
\end{lemma}
\begin{proof}
    By \cref{lemma:compute_normalform}, we can assume without loss of generality that \(\G\) is in the normal (RL)SLG form.
    To construct $\G'$, we process the symbols $A\in \Symb_\G$ in the topological order.
    For each symbol, we add a constant number of symbols to $\G'$, ensuring that each right-hand side has constant (run-length encoded) size.
    We also adhere to the following invariants, which classify symbols of $\G$ into two types:
    \begin{enumerate}[label=(\arabic*)]
      \item\label{type:short} If $|\df(\expand[\G]{A})|=1$, with $\df(\expand[\G]{A})=z_{k_A}$ for $k_A\in \Zz$, then $z_{k_A}\in \Sigma_{\G'}$.
    \item\label{type:long} If $|\df(\expand[\G]{A})|\ge 2$, with $\df(\expand[\G]{A})=z_{\ell_A}\cdot M_A \cdot z_{r_A}$ for some $\ell_A,r_A\in \Zz$ and $M_A\in Z^*$, then $z_{\ell_A},z_{r_A}\in \Sigma_{\G'}$ and there is a symbol $A'\in \Symb_{\G'}$ with $\expand[\G']{A'}=M_A$.
    \end{enumerate}

    Our algorithm implements the following cases:
    \begin{itemize}
      \item If $A = 0$, then $\df(\expand[\G]{A})=z_1$, and we add $z_{k_A}=z_1$ to $\Sigma_{\G'}$.
      \item If $A=1$, then $\df(\expand[\G]{A})=z_0z_0$, and we add $z_{\ell_A}=z_{r_A}=z_0$ to $\Sigma_{\G'}$, and a new symbol $A'$ with $\rhs_{\G'}(A')\coloneqq \emptystring$ to $\Var_{\G'}$.
      \item If $A\in \Var_{\G}$ with $\rhs_{\G}(A)=BC$, there are four sub-cases depending on the types of $B$ and $C$.
      \begin{itemize}
        \item If $B$ and $C$ are of type \ref{type:short}, then $A$ is also of type \ref{type:short} with $k_A = k_B+k_C$. Hence, we add $z_{k_A}$ to $\Sigma_{\G'}$.
        \item If $B$ is of type \ref{type:short} and $C$ is of type \ref{type:long},
        then $A$ is of type \ref{type:long} with $\ell_A = k_B+\ell_C$, $M_A = M_C$, and $r_A = r_C$. Hence, we add $z_{\ell_A}$ to $\Sigma_{\G'}$ and set $A'\coloneqq C'$.
        \item If $B$ is of type \ref{type:long} and $C$ is of type \ref{type:short},
        then $A$ is of type \ref{type:long} with $\ell_A = \ell_B$, $M_A = M_B$, and $r_A = r_B + k_C$. Hence, we add $z_{r_A}$ to $\Sigma_{\G'}$ and set $A'\coloneqq B'$.
        \item If $B$ and $C$ are both of type \ref{type:long}, then $A$ is also of type \ref{type:long} with $\ell_A=\ell_B$, $M_A = M_B z_{r_B+\ell_C} M_C$, and $r_A = r_C$.
        Hence, we add $z_{r_B+\ell_C}$ to $\Sigma_{\G'}$ and a new variable $A'$ to $\Var_{\G'}$ with $\rhs_{\G'}(A')\coloneqq B' z_{r_B+\ell_C} C'$.
      \end{itemize}
      \item If $A\in \Var_\G$ with $\rhs_\G(A)=B^k$ and $k\ge 3$, there are two sub-cases depending on the type of $B$.
      \begin{itemize}
        \item If $B$ is of type \ref{type:short}, then $A$ is also of type \ref{type:short} with $k_A = k \cdot k_B$. Hence, we add $z_{k_A}$ to $\Sigma_{\G'}$.
        \item If $B$ is of type \ref{type:long}, then $A$ is also of type \ref{type:long} with $\ell_A = \ell_B$, $M_A = M_B (z_{r_B+\ell_B} M_B)^{k-1}$, and $r_A = r_B$.
        Hence, we add $z_{r_B+\ell_B}$ to $\Sigma_{\G'}$, a new helper variable $D_A$ to $\Var_{\G'}$ with $\rhs_{\G'}(D_A)\coloneqq z_{r_B+\ell_B}B'$, and a new variable $A'$ to $\Var_{\G'}$ with $\rhs_{\G'}(A')\coloneqq B'D_A^{k-1}$.
        The helper is needed because the power in an RLSLG right-hand side must be a power of a single symbol, not of the block \(z_{r_B+\ell_B}B'\).
      \end{itemize}
    \end{itemize}
    It is straightforward to verify that this implementation satisfies the invariants and can be implemented in $\Oh(1)$ time per symbol (note that we store $k_A$ or $\ell_A,r_A,A'$ for every processed symbol).

    Depending on the type of the starting symbol $S$ of $\G$, we either set $z_{k_S}$ as the starting symbol of $\G'$ or create a new variable $S''$ with $\rhs_{\G'}(S'')\coloneqq z_{\ell_S} S' z_{r_S}$ as the starting symbol of $\G'$.
    In either case, we are guaranteed that $\G'$ produces $\df(T)=\df(\expand[\G]{S})$.
\end{proof}

\begin{lemma}\label{lem:select_simple}
  Let \(\G\) be an (RL)SLG producing an unweighted binary string \(T\in\{0,1\}^n\) of length $n\le 2^{\Oh(w)}$.
  Given $\G$ and a real parameter \(\tau> 1\), in \(\Oh(|\G|\tau)\) time, we can construct an \(\Oh(|\G|\tau\log n)\)-bit data structure that supports \(\select_{T,1}\) queries in \(\Oh\left(\max\left(1,\log_\tau\frac{n}{|\G|\tau}\right)\right)\) time.
\end{lemma}
\begin{proof}
  By \cref{obs:Phi}, every query \(\select_{T,1}(r)\) reduces to computing \(\PrefixSum_{\df(T),\Phi}(r+1)\) and adding \(r\), where \(\Phi : Z \to (\Zz,+,0)\) is given by \(\Phi(z_i)\coloneqq i\).

  Using \cref{claim:selectgrammar}, we construct in \(\Oh(|\G|)\) time an RLSLG \(\G'\) of size \(\Oh(|\G|)\) producing \(\df(T)\).
  We equip \(\df(T)\) with the unit weight function. Since \(|\df(T)|\le n+1\le 2^{\Oh(w)}\), the assumptions of \cref{lemma:prefsum} are satisfied.
  Moreover, by \cref{rmk:monoid}, we may cap all values of the monoid \((\Zz,+,0)\) at \(n\), so \(\rM=\Oh(\log n)\).

  Applying \cref{lemma:prefsum} to \(\G'\) yields, in \(\Oh(|\G'|\tau)\le \Oh(|\G|\tau)\) time, a data structure of bit size
  \[
    \Oh(|\G'|\tau(\rM+\log |\df(T)|))\le \Oh(|\G|\tau\log n)
  \]
  that answers \(\PrefixSum_{\df(T),\Phi}\) queries in time
  \[
    \Oh\!\left(\max\!\left(1,\log_\tau \frac{|\df(T)|}{|\G'|\tau}\right)\right)
    \le
    \Oh\!\left(\max\!\left(1,\log_\tau \frac{n}{|\G|\tau}\right)\right),
  \]
  where the first inequality uses that every character of \(\df(T)\) has weight \(1\), and the second uses \(|\df(T)|\le n+1\) and \(|\G'|\le \Oh(|\G|)\).
  Finally, adding the query index \(r\) takes constant time, so the same asymptotic bound holds for \(\select_{T,1}\).
\end{proof}

\begin{theorem}\label{thm:select}
  Let \(\G\) be an (RL)SLG producing an unweighted binary string \(T\in\{0,1\}^n\) of length \(n\le 2^{\Oh(w)}\).
  Given \(\G\) and a real parameter \(\tau>1\), in time \[\Oh\left(|\G|\left(\tau+\tfrac{\log n}{\max\left(1,\,\log(\frac{n}{|\G|\tau\log n})\right)}\right)\right) \leq \Oh\left(|\G|(\tau+\log n)\right),\]
  we can construct a data structure of size \(\Oh(|\G|\tau\log n)\) bits that supports \(\select_{T,1}\) queries in time \[\Oh\left(\max\left(1,\tfrac{\log \frac{n}{|\G|\tau \log n}}{\log \tau}\right)\right).\]
\end{theorem}
\begin{proof}
  If \(n \le 2|\G|\tau\log n\), then we store \(T\) explicitly as a bit vector and augment it with the constant-time \(\select\) structure of \cref{lemma:bitvector_rankselect}.
  This uses \(\Oh(n)\leq \Oh(|\G|\tau\log n)\) bits and supports \(\select_{T,1}\) queries in constant time.
  To construct the packed bit vector, we first apply \cref{cor:substring} with the same parameter \(\tau\), viewing \(\G\) as an RLSLG if necessary, and extract \(T[0\dd n)\).
  The temporary substring data structure is built in \(\Oh(|\G|\tau)\) time, and the extraction takes
  \[
    \Oh\!\left(1+\frac{n}{\min(w,\tau\log n)}\right)
    \le \Oh(1+n/\log n)
    \le \Oh(|\G|\tau)
  \]
  time, where we use \(w\ge\Omega(\log n)\), \(\tau> 1\), and the assumption of this case.
  We then discard the temporary substring data structure.
  Choosing \(t\coloneqq \ceil{\sqrt{n}}\) in \cref{lemma:bitvector_rankselect}, the shared preprocessing takes \(\Oh(\sqrt n)\le \Oh(n/\log n)\) time and the actual construction takes \(\Oh(n/\log t)=\Oh(n/\log n)\) time.
  Hence, the whole data structure is constructed in \(\Oh(|\G|\tau)\) time.

  We henceforth assume \(n>2|\G|\tau\log n\), and set
  \[
    b \coloneqq \left\lfloor \frac{\log n}{\log\frac{n}{|\G|\tau\log n}}\right\rfloor
  \]
  so that \(1\le b\le \log n\).

  Let \(\Phi : Z\to (\Zz,+,0)\) be defined by \(\Phi(z_i)\coloneqq i\).
  By \cref{obs:Phi}, it suffices, given a query index \(r\), to answer \(\PrefixSum_{\df(T),\Phi}(r+1)\) and add \(r\) to the answer.

  Using \cref{claim:selectgrammar}, we construct in \(\Oh(|\G|)\) time an RLSLG \(\G'\) of size \(\Oh(|\G|)\) producing \(\df(T)\).
  Next, \cref{lemma:leafyconstruct} yields a \(b\)-leafy RLSLG \(\cH\) producing \(\df(T)\), with \(|\Top{\cH}|,|\Varleaf{\cH}|\le \Oh(|\G|)\).
  Adding unused top variables of constant-size right-hand sides if necessary, we assume that \(|\Top{\cH}|=\Theta(|\G|)\) without changing the produced string.

  By \cref{lemma:prefsum-leafy}, it remains to support \(\PrefixSum_{\expand[\cH]{A},\Phi}\) for every leaf variable \(A\in\Varleaf{\cH}\).
  Write
  \(
    \expand[\cH]{A}= z_{k_0}\cdots z_{k_{m_A-1}}
  \),
  where \(m_A\in [b\dd 2b)\), and define the set of prefix sums
  \[
    P_A \coloneqq \left\{r +\sum_{j=0}^{r-1} k_j : r\in [0\dd m_A]\right\}.
  \]
  If \(u_A\coloneqq m_A+\sum_{j=0}^{m_A-1} k_j\), then \(P_A\subseteq [0\dd u_A]\) has \(m_A+1\) elements and, by construction,
  \[
    \PrefixSum_{\expand[\cH]{A},\Phi}(r)=\select_{P_A}(r)-r
  \]
  holds for every \(r\in [0\dd m_A]\).
  Therefore, \cref{claim:EF}, applied with universe size \(u_A+1\), gives a constant-time data structure for leaf \(A\) constructible in $\Oh(m_A)$ time.
  Its bit size is \(\Oh\big((m_A+1)\log\frac{2(u_A+1)}{m_A+1}\big)\le \Oh\big(m_A\log \frac{2u_A}{m_A}\big)\) because \(m_A\ge 1\) and \(u_A\ge m_A\).
  Let \(m\coloneqq \sum_{A\in\Varleaf{\cH}} m_A\), so \(m \le \Oh(|\G|b)\).
  Every leaf variable of \(\cH\) occurs at least once in the top string produced by \(\Top{\cH}\).
  Choosing one occurrence of each leaf variable, we obtain pairwise disjoint substrings of \(\df(T)\).
  For an occurrence of a leaf \(A\) with \(\expand[\cH]{A}= z_{k_0}\cdots z_{k_{m_A-1}}\), the quantity
  \(u_A-m_A=\sum_{j=0}^{m_A-1} k_j\) is exactly the \(\Phi\)-sum of the corresponding substring of \(\df(T)\).
  Consequently,
  \[
    m\le |\df(T)|\le n+1
    \qquad\text{and}\qquad
    \sum_{A\in\Varleaf{\cH}} (u_A-m_A)\le n.
  \]
  The concavity of the logarithm yields
  \[
    \sum_{A\in\Varleaf{\cH}} m_A\log\tfrac{2u_A}{m_A}
    \le
    m\log\left(\tfrac{2}{m}\sum_{A\in\Varleaf{\cH}} u_A\right)
    \le
    m\log\left(\tfrac{2}{m}(n+m)\right)
    \le
    m\log\tfrac{6n}{m}.
  \]
  As \(x\mapsto x\log\tfrac{6n}{x}\) is increasing on \(x\in [1\dd 2n]\) and \(m\le \min(n+1,\Oh(|\G|b))\), this is
  \(
    \Oh\!\left(|\G|b\log\tfrac{n}{|\G|b}\right)
  \).
  To see that this is \(\Oh(|\G|\tau\log n)\), let
  \(\alpha\coloneqq n/(|\G|\tau\log n)>2\) and \(B\coloneqq \log n/\log\alpha\)
  so that \(b=\lfloor B\rfloor\) and, since \(B\ge 1\), we have \(B/2\le b\le B\).
  Consequently,
  \[
    b\log\tfrac{n}{|\G|b}
    \le
    B\log\tfrac{2n}{|\G|B}
    =
    \tfrac{\log n}{\log\alpha}\log(2\alpha\tau\log\alpha)
    \le \Oh(\log n \cdot (1+\log \tau)) \le \Oh(\tau\log n).
  \]
  The top structures of \cref{lemma:prefsum-leafy} use \(\Oh(|\Top{\cH}|\tau(\rM+\log n))=\Oh(|\G|\tau\log n)\) bits, where \(\rM=\Oh(\log n)\) by \cref{rmk:monoid}.
  Hence, the leaf and top structures fit within the space bound.

  For the preprocessing time, constructing \(\G'\) takes \(\Oh(|\G|)\) time.
  Moreover, \cref{lemma:leafyconstruct} constructs \(\cH\) in time
  \[
    \Oh\!\left(|\G|\cdot (1+b\log |\Sigma_{\G'}|/w)\right)
    \le \Oh(|\G|(1+b)),
  \]
  because \(|\Sigma_{\G'}|\le n+1\) and \(w=\Omega(\log n)\).
  The application of \cref{lemma:prefsum-leafy} adds \(\Oh(|\Top{\cH}|\tau+|\Varleaf{\cH}|\tleaf)=\Oh(|\G|\tau)\) preprocessing time.
  Finally, \cref{claim:EF} constructs the leaf data structure for \(A\) in \(\Oh(m_A)\) time, so the total leaf preprocessing time is
  \[
    \Oh\!\left(\sum_{A\in\Varleaf{\cH}} m_A\right)=\Oh(|\G|b).
  \]
  Altogether, the data structure is constructed in \(\Oh(|\G|(\tau+b))\) time, which is
  \[
    \Oh\!\left(|\G|\left(\tau+\tfrac{\log n}{\log\frac{n}{|\G|\tau\log n}}\right)\right).
  \]

  We can now apply \cref{lemma:prefsum-leafy} with \(\tleaf\le \Oh(1)\).
  Using \(|\Top{\cH}|=\Theta(|\G|)\), we obtain query time
  \[
    \Oh\!\left(1+\max\left(1,\log_\tau\tfrac{n}{|\G|\tau b}\right)\right)
    \le
    \Oh\!\left(\max\left(1,\log_\tau\tfrac{n}{|\G|\tau\log n}\right)\right),
  \]
  where the last step follows \(n>2|\G|\tau\log n\) and the choice of $b$.
  Finally, by \cref{obs:Phi}, the same asymptotic query bound holds for \(\select_{T,1}\).
\end{proof}

\paragraph{Arbitrary Alphabets.}
The mapping \(T\mapsto T^{(a)}\) satisfies \(\select_{T,a}(r)=\select_{T^{(a)},1}(r)\) for every \(a\in \Sigma\).
Hence, applying \cref{thm:select} independently to the grammars \(\G^{(a)}\) for all \(a\in \Sigma\) yields an \(\Oh(|\G|\tau\sigma\log n)\)-bit data structure supporting all \(\select_{T,a}\) queries in time \[\Oh\!\left(\max\!\left(1,\frac{\log \frac{n}{|\G|\tau \log n}}{\log \tau}\right)\right).\]
The total construction time becomes
\[
  \Oh\!\left(|\G|\sigma\left(\tau+\tfrac{\log n}{\max\left(1,\log\frac{n}{|\G|\tau\log n}\right)}\right)\right)
  \le \Oh(|\G|\sigma(\tau+\log n)).
\]
This is the source of the extra \(\Oh(|\G|\sigma\log n)\) term compared to the binary \(\rank\) construction.

Finally, we derive \cref{thm:ranksel} by combining the two arbitrary-alphabet constructions above.
Set \(g\coloneqq|\G|\) and \(\tau\coloneqq Mw/(g\log n)\).
The assumptions \(g\log n<Mw<n\) give \(\tau>1\) and \(n/(g\tau\log n)=n/(Mw)>1\), while \(w\ge\Omega(\log(n\sigma))\) ensures \(n\le2^{\Oh(w)}\).
The combined space bound is \(\Oh(g\tau\sigma\log n)=\Oh(Mw\sigma)\) bits, or \(\Oh(M\sigma)\) machine words.
The construction time is dominated by the select construction and is
\[
  \Oh(g\sigma(\tau+\log n))
  =
  \Oh\!\left(\tfrac{Mw\sigma}{\log n}+g\sigma\log n\right).
\]
Both \(\rank\) and \(\select\) query times become
\[
  \Oh\!\left(\max\!\left(1,\frac{\log\frac{n}{Mw}}{\log\frac{Mw}{g\log n}}\right)\right),
\]
which is the query bound in \cref{thm:ranksel}, with the convention that query times are at least constant.

%% file: appendix.tex
\section{Proof of Corollary~\ref{thm:block}}\label{section:appendix}

\begingroup
\makeatletter
\let\thmBlock@origcorollary\corollary
\def\corollary{\@ifnextchar[\thmBlock@skipopt\thmBlock@origcorollary}
\def\thmBlock@skipopt[#1]{\thmBlock@origcorollary}
\thmBlock*
\makeatother
\endgroup

\begin{proof}
Define $\phi_d\coloneqq \frac{n\log \sigma}{d\log n}$ and $\phi_s\coloneqq \frac{n\log \sigma}{s\log n}$.
Since $d\tau\le s$, we have
\begin{equation}\label{eq:block-phi}
  \phi_d=\phi_s\cdot \tfrac{s}{d}
  \qquad\text{and}\qquad
  \phi_s\le \phi_d.
\end{equation}
Using the general bound
\[
  g_{\rle}(T)\le \Oh\!\left(\delta(T)\log\tfrac{n \log \sigma}{\delta(T) \log n}\right)
\]
the assumption $d\ge \delta(T)$, and \eqref{eq:block-phi}, we get
\begin{equation}
  g_{\rle}(T) \le \Oh(d \log\tfrac{n \log \sigma}{d \log n}) = \Oh(d \log\phi_d) = \Oh\!\left(d\log \phi_s+d\log\tfrac{s}{d}\right). \label{eq:block-rle}
\end{equation}
Because $\frac{s}{d}\ge \tau\ge 2$, we have $\log(\frac{s}{d})\le \Oh(\frac{s}{d\tau}\log\tau)$, so \eqref{eq:block-rle} simplifies to
\begin{equation}\label{eq:block-g}
  g_{\rle}(T)\le \Oh\!\left(d \log \phi_s+\tfrac{s}{\tau}\log\tau\right).
\end{equation}
Hence, $T$ is generated by an RLSLG of size $g$ satisfying \eqref{eq:block-g}.

Choose
\begin{equation}\label{eq:block-M}
  M\coloneqq \left\lceil s+\tfrac{g\tau}{\log\tau}\right\rceil.
\end{equation}
By \eqref{eq:block-g} and \eqref{eq:block-M}, we have
\begin{equation}\label{eq:block-M2}
  M
  = \Oh\!\left(s+\tfrac{g\tau}{\log\tau}\right)
  \le \Oh(s+d\tau \log_\tau \phi_s).
\end{equation}
Moreover, since $\tau\ge 2$, we have
\[M\ge s \qquad \text{and}\qquad M\ge \tfrac{g\tau}{\log\tau}> g.\]
We apply \cref{thm:upper} (in its RLSLG variant) with word size $w\coloneqq \ceil{\log n}$ so that $Mw>g\log n$.

If $Mw\ge n\log\sigma$, we can store $T$ explicitly in packed form, using $n\log\sigma/w\le M$ words and answering \RA in $\Oh(1)$ time.

Otherwise, $g\log n < Mw < n\log\sigma$, so the RLSLG variant of \cref{thm:upper} yields an $\Oh(M)$-word data structure.
To simplify its query time bound, note that \eqref{eq:block-M2} and $w=\Theta(\log n)$ give
\begin{equation}\label{eq:block-query-bounds}
  \tfrac{n\log \sigma}{Mw}\le\Oh(\phi_s)
  \qquad\text{and}\qquad
  \tfrac{Mw}{g\log n}\ge\Omega\!\left(\tfrac{\tau}{\log\tau}\right).
\end{equation}
Therefore, the query time is
\[
  \Oh\left(\frac{\log \frac{n \log \sigma}{Mw}}{\log\frac{Mw}{g\log n}}\right)
  \le \Oh\left(\frac{\log \phi_s}{\log \frac{\tau}{\log\tau}}\right)
  = \Oh\left(\tfrac{\log \phi_s}{\log\tau}\right)
  = \Oh(\log_\tau \phi_s),
\]
where the middle step uses $\log(\tau/\log\tau)=\Theta(\log\tau)$ for $\tau\ge 2$.

For optimality, assume $n\ge \sigma\ge 2$, $\tau \ge \log^{\epsilon} n$, and $d \ge \log^{2+\epsilon} n$.
Consider any $\Oh(s+d\tau\log_\tau \phi_s)$-word data structure supporting \RA on all strings with $z(T) \le g(T)\le d$ and answering each query by probing $t$ words using a nondeterministic or bounded-error randomized algorithm.
If $\phi_s\le \tau^2$, then $\log_\tau \phi_s \le 2$ and thus $t\ge 1\ge\Omega(\log_\tau \phi_s)$.
Otherwise, we may assume its size is exactly $M=\Theta(s+d\tau\log_\tau \phi_s)$, and we apply \cref{thm:lower} with grammar size $g\coloneqq d$, word size $w\coloneqq \ceil{\log n}$, and a sufficiently small constant $\epsilon'>0$ (e.g., $\epsilon'\le \frac12\epsilon$).
The conditions of \cref{thm:lower} hold since $n\ge\sigma\ge2$, the range of $s$ gives $\phi_d\ge\tau\ge2$, which together with $n\ge\sigma$ implies $d\le n$, $g=d \ge \log^{2+\epsilon} n$, and
\[
  Mw=\Theta((s+d\tau\log_\tau \phi_s)\cdot \log n)\ge d\tau\log n\ge g\log n \cdot (w\log n)^{\epsilon'}.
\]
Moreover, the upper bound $s\le d\tau\log_\tau \phi_d$ from the statement and \eqref{eq:block-phi} imply
\begin{equation}\label{eq:block-lb-M}
  M=\Theta(s+d\tau\log_\tau \phi_s)\le \Oh(d\tau\log_\tau \phi_d).
\end{equation}
Substituting this into \cref{thm:lower}, we obtain
\[
  t
  \ge \Omega\left(\frac{\log \frac{n\log\sigma}{Mw}}{\log \frac{Mw}{g\log n}}\right)
  \ge \Omega\left(\frac{\log \frac{\phi_d}{\tau\log_\tau \phi_d}}{\log(\tau\log_\tau \phi_d)}\right).
\]
It remains to simplify the last expression.
Since $\phi_d \ge \phi_s>\tau^2$, we have
\[
  \log \tfrac{\phi_d}{\tau\log_\tau \phi_d}
  \ge \log \tfrac{\phi_s}{\tau\log_\tau \phi_s}
  = \log \phi_s-\log\tau-\log\log \phi_s + \log \log \tau
  \ge \tfrac12\log \phi_s - \log\log \phi_s
  = \Omega(\log \phi_s),
\]
where the penultimate step uses $\phi_s>\tau^2$.
On the other hand, $n\ge\sigma$ implies $\phi_d\le n$, and $\tau\ge \log^\epsilon n$ yields
\[
  \log(\tau\log_\tau \phi_d)
  = \log\tau+\log\log_\tau \phi_d
  \le  \log\tau+\log\log n
  \le \log \tau +\epsilon^{-1}\log \tau
  \le \Oh(\log\tau).
\]
Therefore,
\[
  t\ge \Omega\!\left(\tfrac{\log \phi_s}{\log\tau}\right)
   = \Omega(\log_\tau \phi_s).\qedhere
\]
\end{proof}